\documentclass[11pt,letter]{article}
\usepackage[margin=1in]{geometry}
\usepackage[utf8]{inputenc}
\usepackage{amsmath}
\usepackage{amssymb}
\usepackage{amsthm}
\usepackage{dsfont}
\usepackage{color}
\usepackage{amsmath}
\usepackage{amsfonts}
\usepackage{amssymb}
\usepackage{bbm}
\usepackage{relsize}
\usepackage[procnumbered, linesnumbered,algoruled]{algorithm2e}
\usepackage[compact]{titlesec}
\usepackage{lineno}
\usepackage{enumitem}

\usepackage[hidelinks]{hyperref}

\newcommand{\posA}{\vspace{0.08in}}

\newcommand{\ceil}[1]{{\left\lceil#1  \right\rceil}}
\newcommand{\comment}[1]{}

\newcommand{\poly}{\textnormal{poly}}

\newcommand{\cA}{{\mathcal{A}}}
\newcommand{\cB}{{\mathcal{B}}}
\newcommand{\cG}{{\mathcal{G}}}
\newcommand{\cI}{{\mathcal{I}}}
\newcommand{\cS}{{\mathcal{S}}}
\newcommand{\cM}{{\mathcal{M}}}

\newcommand{\cj}{{\mathcal{J}}}

\newcommand{\cC}{{\mathcal{C}}}
\newcommand{\bx}{{\bar{x}}}
\newcommand{\by}{{\bar{y}}}

\newcommand{\OPT}{\textnormal{OPT}}
\newcommand{\greedy}{\textnormal{\textsf{Greedy}}}

\newcommand{\eps}{{\varepsilon}}

\newcommand{\floor}[1]{\left\lfloor #1 \right\rfloor}

\DeclareMathOperator*{\argmin}{arg\,min}

\begin{document}
	\newtheorem{thm}{Theorem}[section]
	\newtheorem{prop}[thm]{Proposition}
	\newtheorem{assm}[thm]{Assumption}
	\newtheorem{lem}[thm]{Lemma}
	\newtheorem{obs}[thm]{Observation}
	\newtheorem{cor}[thm]{Corollary}
	\newtheorem{lemma}[thm]{Lemma}
	\newtheorem{definition}[thm]{Definition}
	\newtheorem{theorem}[thm]{Theorem}
	\newtheorem{proposition}[thm]{Proposition}
	\newtheorem{observation}[thm]{Observation}
	\newtheorem{claim}[thm]{Claim}
		\newtheorem{example}[thm]{Example}
	\newtheorem{defn}[thm]{Definition}
	\newcommand{\ariel}[1]{{\color{red} (Ariel :#1)}}
	\newcommand{\ilan}[1]{{\color{blue} Ilan :\color{magenta}#1}}
	\def \II   {{\mathcal I}}
	\newcommand{\one}{\mathbbm{1}}

	\begin{titlepage}
	\title{
		 Bin Packing with Partition Matroid can be Approximated within $o(OPT)$ Bins
	}
	\author{Ilan Doron-Arad\thanks{Computer Science Department, 
		Technion, Haifa, Israel. \texttt{idoron-arad@cs.technion.ac.il}}
	\and
	Ariel Kulik\thanks{CISPA Helmholtz Center for Information Security, Saarland Informatics Campus, Germany. \texttt{ariel.kulik@cispa.de}} 
	\and
	Hadas Shachnai\thanks{Computer Science Department, 
		Technion, Haifa, Israel. \texttt{hadas@cs.technion.ac.il}}
}

	\maketitle

	\begin{abstract}
\label{sec:abstract} 

We consider the Bin Packing problem with a partition matroid constraint.
The input is a set of items of sizes in $(0,1]$, and a partition matroid over the 
items. The goal is to pack all items in a minimum number of unit-size bins, such that each bin forms an independent set in the matroid. The problem is a generalization of both Group Bin Packing and Bin Packing with Cardinality Constraints. 
Bin Packing with Partition Matroid naturally arises
in resource allocation to ensure fault tolerance and security,
as well as in harvesting computing capacity.
 
Our main result is a polynomial-time algorithm that packs the items in $\OPT + o(\OPT)$ bins, where $\OPT$ is the minimum number of bins required for packing the given instance. This matches the best known result for the classic Bin Packing problem
up to the function hidden by $o(\OPT)$.  As special cases, our result improves upon the existing APTAS for Group Bin Packing and generalizes the AFTPAS for Bin Packing with Cardinality Constraints. 
Our approach is based on rounding a solution for a   configuration-LP formulation of the problem. The rounding takes a novel point of view of {\em prototypes} in which items are interpreted as placeholders for other items and applies {\em fractional grouping} to modify a fractional solution (prototype) into one having nice integrality properties.


\end{abstract}

	\thispagestyle{empty}
\end{titlepage}

	\tableofcontents
\thispagestyle{empty}
	\newpage
\setcounter{page}{1}

\section{Introduction}\label{Introduction}
\label{sec:intro}
The {\em bin packing (BP)} problem involves 
packing a set of items in a minimum number of containers (bins) of the same (unit) size. Bin Packing is one of the most studied problems in combinatorial optimization. Indeed, 
 in many real-life scenarios, a solution for BP is essential for optimizing the allocation of resources. In this paper, we consider the Bin Packing problem with a partition matroid constraint.
 The input is a set of items of sizes in $(0,1]$, and a partition matroid over the 
 items. The goal is to pack all items in a minimum number of unit-size bins, such that each bin forms an independent set in the matroid.

 Formally, a {\em bin packing with partition matroid (BPP)} instance is a tuple $\mathcal{I} = (I,\mathcal{G},s,k)$, where~$I$ is a set of items, $\mathcal{G}$ is a partition of $I$ into groups,
  $s:I \rightarrow (0,1]$ gives the item sizes, and $k:\cG \rightarrow \mathbb{N}_{>0}$ is a cardinality constraint for the groups.  The {\em instance matroid} of $\cI$ is the partition matroid $\cM=(I, \cS)$ where $\cS=\{S\subseteq I	~|~\forall G\in \cG:~|S\cap G|\leq k(G)\}$. 
  A {\em configuration} of the instance $\cI$ is 
a subset of items $C \subseteq I$ such that $\sum_{\ell \in C} s(\ell) \leq 1$ and $C\in \cS$. That is,  the total size of items in $C$ is at most one and for each {\em group} $G \in \mathcal{G}$ the configuration~$C$ contains at most $k(G)$ items from $G$. 
 A {\em packing} of $\mathcal{I}$ is a partition  of $I$ into $m$ subsets called {\em bins} $(A_1, \ldots, A_m)$ such that $A_b$ is a configuration
 for all $b \in [m]$.\footnote{For any $n \in \mathbb{N}$ we denote by $[n]$ the set $\{1,2,\ldots,n\}$.}  
     The objective is to find a packing of all items in a minimal number of bins. Indeed, the special case where each group consists of a {\em single} item is the classic Bin Packing problem.
 
 Bin Packing with Partition Matroid has natural applications in resource allocation on the Cloud to ensure fault tolerance~\cite{EAFR} and security~\cite{guerine2019provenance},
as well as in harvesting computing capacity in distributed systems~\cite{boinc04,JOUR}.
For a simple example, consider a set $I$ of unit time jobs accessing shared memory of size $Q$. Each job $\ell \in I$ has some memory requirement, as well as the type $R \in {\cal R}$ of processor on which it can execute, where ${\cal R}$ is the set of distinct
processor types.
Assume the number of available processors of type $R$ is $k(R)$, for all
$R \in {\cal R}$. We seek a feasible schedule 
of $I$ that minimizes the maximum completion time of any job (or, the {\em makespan}). Then each time unit can be viewed as a single bin of capacity $Q$, and  each job $\ell$ as an item whose size is equal to the memory requirement of job $\ell$.
The goal is to use as few bins (= time units) as possible, while ensuring that in each bin at most $k(R)$ jobs execute, for all $R \in {\cal R}$.
By setting the bin capacities to $1$ and scaling down the item sizes accordingly, we have an instance of BPP.

 Let $\OPT=\OPT(\II)$ be the value of an optimal solution for an instance~ $\II$ of a minimization problem~$\mathcal{P}$. 
 As in the Bin Packing problem, we distinguish between {\em absolute} and {\em asymptotic} approximation.
 For $\alpha \geq 1$, we say that $\cA$ is an absolute $\alpha$-approximation algorithm for 
 $\mathcal{P}$ if for any instance $\II$ of~$\mathcal{P}$ it holds that  $ \cA (\II)/\OPT(\II) \leq \alpha$, where $\cA(\II)$ is the value of the solution returned by $\cA$. Algorithm
 $\cA$  is an {\em asymptotic} $\alpha$-approximation algorithm for 
 $\mathcal{P}$ if  for any instance $\II$ it holds that $\cA (\II) \leq \alpha \OPT(\II) +o(\OPT(\II))$.
An  {\em asymptotic polynomial-time
	approximation scheme (APTAS)} is a family of algorithms $(\cA_{\eps})_{\eps>0}$ such that, for every $\eps>0$, $\cA_{\eps}$ is a polynomial time asymptotic $(1+\eps)$-approximation algorithm for $\mathcal{P}$.  An {\em asymptotic {\bf fully} polynomial-time
	approximation scheme (AFPTAS)} is an APTAS $(\cA_{\eps})_{\eps>0}$  such that $\cA_{\eps} (\II)$ runs in time $\poly(|\II|, \frac{1}{\eps})$, where $|\II|$ is the encoding length of the instance $\II$; that is, there is a bivariate polynomial  $p$ such that $\cA_{\eps}(\II)$ runs in time $p(|\II|,\frac{1}{\eps})$ or less.

 As Bin Packing is known to admit an $o(OPT)$ 
 {\em additive} approximation~\cite{karmarkar1982efficient,hoberg2017logarithmic}, this 
 gives rise to the following question:
 Is Bin Packing with Partition Matroid harder to solve than classic BP?
 While key techniques for solving BP break down in the presence of a partition matroid constraint, we show that by applying different tools we can obtain  a result similar to 
 the best known result for BP, up to the function hidden
 by the additive approximation.
 Our main result is the following.
 
 \begin{theorem}
 	\label{thm:main}
 There is a polynomial-time algorithm that, given a BPP instance $\mathcal{I}$, returns a packing of $\mathcal{I}$ in $\OPT(\cI) +O\left( \frac{\OPT(\cI)}{\left(\ln \ln \OPT(\cI) \right)^{1/17}} \right)=\OPT(\mathcal{I})+o\left( \OPT(\mathcal{I})\right)$ bins.
 \end{theorem}

We obtain the result by initially deriving an AFPTAS for BPP. Then Theorem~\ref{thm:main} follows by using the AFTPAS with~$\eps$ depending on the input instance. 
As a special case, 
Theorem~\ref{thm:main} improves upon the recent APTAS of Doron-Arad et al.~\cite{DKS21} for {\em group bin packing (GBP)}, where the cardinality constraint of all groups is equal to one (i.e., $k(G)=1~\forall G \in \cG$). Theorem~\ref{thm:main} also generalizes the AFTPAS of Epstein and Levin~\cite{epstein2010afptas} for {\em bin packing with cardinality constraints (BPCC)}, the special case of BPP where all items belong to a single group (i.e., $\cG=\{I\}$).

We remark that we did not attempt to optimize the function hidden by $o(\OPT(\cI))$ in Theorem~\ref{thm:main}; namely, our techniques may be used to derive
better additive factor. We leave such efforts for the journal version of this paper.

\subsection{Techniques}
\label{sec:technique}

The classic asymptotic approximation schemes for Bin Packing (see, e.g., \cite{Vega1981BinPC, karmarkar1982efficient}) rely on the nice property that {\em small} items can be added to 
a partial 
packing of the instance with little overhead, using simple algorithms such as First-Fit (see, e.g., \cite{Williamson_Shmoys}). For example, assume we are given a set of items $I$, a size function	 $s:I\rightarrow [0,1]$ and a partition $A_1,\ldots, A_m $ and $S$ of $I$. Furthermore, assume that  $s(A_b) =\sum_{\ell \in A_b} s(\ell) \leq 1$ for any $b\in [m]$,
and $s(\ell)\leq \delta$ for any $\ell \in S$, for some $0<\delta <1$.\footnote{For any function $f:I\rightarrow \mathbb{R}$ and $A\subseteq I$ we define $f(A)= \sum_{\ell \in A} f(\ell)$.}  Then, using First-Fit one can easily find a partition $D_1,\ldots, D_{m'}$ of $I$ such that $s(D_b)\leq 1 $ for all $b\in [m']$ and $m'\leq \max\{ m , (1+2\delta ) s(I)+1\}$. 
 Consequently, the main focus of 
 these schemes is on packing efficiently items of sizes larger than $\delta$, where $\delta>0$  is some small parameter value used by the algorithm. 

We note that packing small items in the presence of a partition matroid constraint is more involved. As we show below, an efficient packing of the small items can still be obtained using a relatively simple algorithm; however, the setting in which the algorithm can be applied is more restricted, and items cannot be easily added to a partial packing of the instance (i.e., a set of configurations). Furthermore, the quality of such a packing depends on the  {\em cardinality bound} of the BPP instance   $\mathcal{\cI} = (I,\cG,s,k)$, defined by $V(\mathcal{I}) = \max_{G \in \cG} \ceil{\frac{|G|}{|k(G)|}}$. Formally,
\begin{lemma}
\label{thm:greedy}
Given a BPP instance $\cI = (I, \cG,s,k)$ and $\delta\in (0,0.5)$, such that $s(\ell) \leq \delta$ for all $\ell \in I$, there is an algorithm  $\greedy$ that	returns in polynomial time a packing of $\cI$ in at most $(1+2\delta) \cdot \max\left\{ s(I),~V(\mathcal{I})\right\}+2$ bins.
\end{lemma} 

To obtain a packing of the input instance $\II$ of BPP, our algorithm partitions the set of items $I$ into subsets $I_1,\ldots, I_r$; for each subset $I_j$  it generates a  partial packing $A'_1,\ldots, A'_m$ (i.e., $A'_1,\ldots ,A'_m$ are configurations, and $ \bigcup_{b\in [m]} A'_b\subseteq I_j$ ) which can be extended to a packing  of $I_j$ (i.e, one which contains all the items in $I_j$) using $\greedy$. For this procedure to work, a crucial property is
 that $s(\ell) \ll 1- s(A'_b)$ for every $\ell \in I_j\setminus \bigcup_{b' \in [m]} A'_{b'}$ and $b\in [m]$. 
Furthermore, for any $G\in \cG$, $|A'_b\cap G|$ has to be upper bounded by a uniform value
for all $b\in [m]$,
 and $|G\cap (I_j\setminus \bigcup_{b\in [m]} A'_b)|$ must be small.
Thus, a main part of the algorithm deals with  
generation of the above partition of $I$ into $I_1,\ldots,I_r$ and the corresponding partial packings.\footnote{See the details in Section~\ref{sec:overview}.}

Let $\cC$ denote the set of configurations of a BPP instance $\cI=(I,\cG, s,k)$. 
The novelty of our algorithm lies in an interpretation of vectors $\bx\in \mathbb{R}_{\geq 0}^{\cC}$ (that is, vectors which have a non-negative entry for each configuration $C\in \cC$)  as  {\em prototypes}.
The prototype $\bx\in \mathbb{R}_{\geq 0}^{\cC}$ serves as a blueprint for a (fractional) packing. Within the context of prototypes, each item $\ell \in C$ of a configuration~$C$ is interpreted as a placeholder (or, ``slot'') for items which can replace it. Also, some items may be added, thus utilizing the available capacity of the  configuration~$C$ (given by $1-s(C)$). The value $\bx_C$ represents a solution in which $\bx_C$ configurations match the blueprint corresponding to
$C$. We associate a polytope with each prototype; a non-empty polytope ensures the prototype can be used to obtain a solution for the given BPP instance.

The partition of $I$ into $I_1,\ldots, I_r$ and $S$, and the generation of the partial packings ($A'_1,\ldots, A'_m$) follow from integrality properties of the polytope (associated with a prototype). However, for the integrality properties to be useful,  both the number of configurations on the support of the prototype (i.e, $\{C\in \cC~|~\bx_C>0\}$) and  the number of items in each configuration on the support must be small.

We obtain an initial prototype by solving a standard configuration-LP formulation of the problem. In this initial prototype, each item serves as a placeholder for itself.  The algorithm  modifies the prototype sequentially using two steps: {\em eviction} and {\em shifting}. The eviction step reduces the number of items per configuration on the support of the prototype. The shifting step reduces the overall number of distinct items used by configurations on the support. Consequently, our algorithm can be viewed as a multi-step rounding process for the solution of the configuration-LP. 

The shifting step applies the {\em fractional grouping} technique
introduced in~\cite{FKS21}. 
The technique, originally developed for solving a constrained submodular maximization problem, can be viewed as a {\em fractional} version of
{\em linear shifting}~\cite{fernandez1981bin}. Given the items sorted in non-increasing order by sizes, suppose that each item $\ell \in I$ is associated with a weight $w_\ell \in (0,1]$. For a subset of items $I' \subseteq I$ and an integer $\tau \geq 1$, we define a partition of $I'$ into $\tau$ classes $\Delta_1, \ldots , \Delta_\tau$ of roughly the same (fractional) weight, where the weight of class $i$ is given by 
$w(\Delta_i) = \sum_{\ell \in \Delta_i} w_\ell$.\footnote{Our scheme applies fractional grouping separately to certain groups $G \in \mathcal{G}$
	(see the details in Section~\ref{sec:AlgPrototype}).}
In our setting, the weights $w_{\ell}$ are determined by the prototype (after the eviction). 
Following the fractional grouping, each of the items in a class $\Delta_i$ can be replaced (in the prototype) by one of a few representatives from the same class, thus reducing the number of distinct items in configurations on the support of the prototype, while ensuring the polytope associated with the prototype remains non-empty. 

Linear shifting can be viewed as
the special case of fractional grouping in which the weight  of all  items is equal to one (i.e., $w_{\ell}=1$). 
A key idea of the technique is that the size of each item in $\Delta_{i}$ can be rounded up to the maximal size of an item in this set. By {\em shifting} the items in $\Delta_{i}$ to the bins used by items of $\Delta_{i-1}$, it is shown that the increase in item sizes does not incur a significant increase in the number of bins required for packing the instance. The technique works well for bin packing problems in which the set $I'$ on which we apply  grouping is known in advance, as in classic Bin Packing~\cite{fernandez1981bin}. In other applications~\cite{DKS21,bansal2016improved}, the set of items $I'$ is unknown to the algorithm. To overcome this difficulty, the algorithms of~\cite{DKS21, bansal2016improved}  guess properties of the sets $\Delta_1,\ldots, \Delta_\tau$ which suffice for efficient grouping. These guesses render the running times exponential in $1/\eps$.
Thus, the running times of these schemes are polynomial only when $\eps$ is {\em fixed}.
Fractional grouping overcomes this difficulty by avoiding the 
guessing step, thus eliminating the exponential dependence on $1/\eps$. Hence, 
the use of fractional grouping is a key to obtaining an asymptotic {\em fully} polynomial-time approximation scheme (AFPTAS) from which we derive the result in Theorem~\ref{thm:main}. Unlike 	previous works~\cite{FKS21, KMS22}, this paper gives the first constructive application of fractional grouping.

\subsection{Prior Work}
\label{sec:related}

The classic Bin Packing problem 
is known to be NP-hard and cannot be approximated within a ratio 
better than $\frac{3}{2}$, unless P=NP. 
This ratio is achieved by the simple First-Fit Decreasing algorithm~\cite{S94}.
The paper~\cite{Vega1981BinPC} 
presents an APTAS for Bin Packing, which uses at most $(1+\eps)\OPT+1$ bins, for any fixed $\eps \in (0, 1/2)$. The paper~\cite{karmarkar1982efficient}  
gives an AFPTAS and approximation algorithm that uses at most $\OPT+O(\log^2(\OPT))$ bins. 
The additive factor was improved in~\cite{Ro13} to
$O(\log \OPT \cdot \log \log \OPT)$, and later to
$O(\log(\OPT))$~\cite{hoberg2017logarithmic}.
For comprehensive surveys of known results for BP see, e.g.,~\cite{C+13,C+17}.

As the literature on Bin Packing and its variants is immense, we review below results that relate to packing under {\em partition matroid constraint}.
The hardness of approximation of BPP (with respect to absolute approximation ratio) follows from the hardness of BP.

The special case of Group Bin Packing (in which $k(G)=1$ for all $G \in \cG$) was first studied in~\cite{JO97}. The approximation ratio of $2.7$ obtained in this paper  was later improved
to better constants in several papers (e.g.,~\cite{mccloskey2005approaches,allornothing}).
The best known result for general instances is an APTAS due to~\cite{DKS21}. 
The algorithm of~\cite{DKS21}  is based on extensive guessing of  properties of an optimal solution which are then used as guidance  for the assignment of items to bins.
In particular, the algorithm  does not round a solution for a configuration-LP, and cannot be viewed as a rounding algorithm in general.  
The extensive guessing leads to running time that is exponential in $\frac{1}{\eps}$.  
For a special case of GBP where the maximum cardinality of a group is some constant, an AFPTAS follows from a result of~\cite{jansen1999approximation}.

GBP was studied also in the context of scheduling on identical machines. 
Das and Wiese~\cite{DW17} introduced the problem of makespan minimization with bag constraints.
In this generalization of the classic makespan minimization problem, each job belongs to a {\em bag}. The goal is to schedule the jobs on a set of $m$ identical machines, for some $m \geq 1$, such that no two jobs in the same bag are assigned to the same machine, and the makespan is minimized. For the classic problem of makespan minimization with no bag constraints, there are known 
{\em polynomial time approximation scheme (PTAS)}~\cite{hochbaum1987using,MS:10} as well as
{\em efficient polynomial time approximation scheme (EPTAS)}~\cite{hochba1997approximation,alon1998approximation,MS:5,JKV16}. Das and Wiese~\cite{DW17} developed a PTAS for the problem with bag constraints. Later, Grage et al.~\cite{Jansen_et_al:2019} obtained an EPTAS.

Another special case of BPP is Bin Packing with Cardinality Constraints, in which $|\cG|=1$, i.e., we have a single group $G$ with $k(G)=k$, for some integer $k \geq 1$.
BPCC has been studied since the 1970's~\cite{KSS75,KSS77,KP99, CKP03}. The best known result is an AFPTAS due to~\cite{epstein2010afptas}, which relies on rounding a non-standard configuration-LP formulation of the problem. 
The techniques used in this work bear some similarities to the techniques of~\cite{epstein2010afptas}.
For example, our interpretation for a configuration, 
in which unused capacity is an available capacity for additional items,
is similar to the concept of {\em windows} used in~\cite{epstein2010afptas}.
The AFTPAS of~\cite{epstein2010afptas} utilizes the property that if $k<\frac{1}{\eps}$ then
only $\frac{1}{\eps}$ items fit into a bin; thus, linear shifting can be applied to the whole instance.We note that in the presence of multiple  groups, 
the cardinality bound for each group may be small (or even equal to $1$), yet the number of items that can be packed in a single bin may be arbitrarily large.  This is one of several
hurdles encountered when attempting to adapt the algorithm of~\cite{epstein2010afptas} to our setting of BP with a partition matroid constraint.

In the {\em matroid partitioning} problem, we are given a ground set $U$ and a matroid $({\cM},{\cS})$, where ${\cS}$ is a family of subsets of $U$, known as the  
the {\em independent sets} of the matroid. We seek a partition of $U$ to as few independent sets as possible. The problem is polynomially solvable for any matroid ${\cM}$ over $U$ using a combinatorial algorithm (see, e.g.,~\cite{E65,GW92}). When $\cM$ is a partition matroid, the matroid partitioning problem can be viewed as a variant of BPP with unbounded bin capacities.

To the best of our knowledge, Bin Packing with Partition Matroid is studied here
for the first time.

\subsection{Organization of the Paper}

In Section~\ref{sec:preliminaries} we give some definitions and notation.
 Section~\ref{sec:overview} gives an overview of our scheme, that is applied to a {\em structured} instance, and the main lemmas used in its analysis.
 In Sections~\ref{sec:eviction}$-$\ref{sec:alg_pack}
 we present the main components of the scheme: algorithm $\textsf{Evict}$ (Section~\ref{sec:eviction}), algorithm $\textsf{Shift}$
 (Section~\ref{sec:AlgPrototype}), algorithm $\textsf{Partition}$ (Section~\ref{sec:FromPolytope}), and algorithm $\textsf{Pack}$ (Section~\ref{sec:alg_pack}).
 In Section~\ref{sec:reduction} we present algorithms \textsf{Reduce} and \textsf{Reconstruct}, which handle the structuring of the instance and 
 the transformation of the solution for the structured instance into a solution for the original instance, respectively. We conclude in Section~\ref{sec:discussion} with a summary and directions for future work. 
 
 For clarity of the presentation, we defer the formal proofs to the Appendix. Also, algorithm $\textsf{Greedy}$ is presented in Appendix~\ref{sec:greedy}.

\section{Preliminaries}
\label{sec:preliminaries}

Let $\cA$ be an algorithm  that accepts as input $\eps > 0$. 
We say the running time of $\cA$ is
$\textnormal{poly} (| \cI |, \frac{1}{\eps})$ if there is a two-variable polynomial $p(x,y)$ such that $p(|\cI|,\frac{1}{\eps})$ is an upper bound on the running time of $\cA(\cI, \eps)$. To allow a simpler presentation of the results
we assume throughout the paper that  the set of items $I$ is $\{1,2,\ldots,n\}$, and the items are sorted in non-increasing order by sizes, i.e., $s(1)\geq s(2)\geq \ldots s(n)$.

Our scheme initially transforms a given BPP instance into one having a structure which depends on the parameter 
$\eps >0$.  In this new instance, only a small number of groups may contain relatively large items. 
Let $K:(0,0.1) \to \mathbb{R}$, where
  $K(\eps) =  \eps^{-{\eps^{-2}}}$ for all $\eps \in (0,0.1)$. We use
  $K$ for defining a structured instance.  
\begin{definition}
	\label{def:structuring}
	Given a BPP instance $\cI = \left(I, \mathcal{G},s,k\right)$ and $\eps\in (0,0.1)$, we say that $\cI$
	is $\eps$-{\em structured} if there is $\cB \subseteq \cG$ such that  $|\cB| \leq K(\eps)$ and for all $G \in \cG \setminus \cB$ and $\ell \in G$ it holds that~$s(\ell) < \eps^2$. 
\end{definition}

Following the structuring step, our 
scheme proceeds to solve BPP on the structured instance. As a final step,
the packing found for the structured instance is transformed into 
a packing of the original instance. This is formalized in the next result.
\begin{lemma}
	\label{lem:reductionReconstruction}
	There is a pair of algorithms, $\textsf{Reduce}$ and $\textsf{Reconstruct}$, which satisfy the following.
	\begin{enumerate}
		
		\item 
		Given a BPP instance $\cj$ and $\eps >0$ such that $\eps^{-1}\in \mathbb{N}$,
		algorithm $\textsf{Reduce}$ returns in time $\textnormal{poly} (| \cI |, \frac{1}{\eps})$ an $\eps$-structured BPP instance $\cI$, where $\OPT(\cI) \leq \OPT(\cj)$.\label{condition:reduction1}
		
		\item 
		Given a BPP instance $\cj$, $\eps >0$ such that $\eps^{-1}\in \mathbb{N}$, and a packing $A'$ for $\cI= \textsf{Reduce}(\cj, \eps)$ of size $m'$,
		algorithm $\textsf{Reconstruct}$ returns in time $\textnormal{poly} (| \cI |, \frac{1}{\eps})$ a packing $A$ for the instance $\cj$ of size $m$, where $m \leq m'+13 \eps\cdot \OPT(\cj)+1$.\label{condition:reduction2} 
	\end{enumerate}
	
\end{lemma}

The structured instance $\cI$ is constructed from $\mathcal{J}$ by reassigning items of size at least $\eps^2$ from all but a few groups to a new
group. The reconstruction algorithm modifies  the packing of $\cI$ such that each bin in the solution is a configuration of $\mathcal{J}$. The proof of Lemma~\ref{lem:reductionReconstruction} (given in Section~\ref{sec:reduction}) is 
inspired by ideas of~\cite{DW17,Jansen_et_al:2019,DKS21}.
By Lemma~\ref{lem:reductionReconstruction}, an AFTPAS for $\eps$-structured BPP instances  implies an AFTPAS for general BPP instances.

\newcommand{\genalg}{\textnormal{\textsf{Gen-AFPTAS}}}

\section{Approximation Algorithm for $\eps$-Structured Instances}
\label{sec:overview}

Our algorithm uses a configuration-LP  
relaxation of the given BPP instance.
For $\eps \in (0,0.1)$ such that $\eps^{-1}\in\mathbb{N}$, let $\mathcal{I} = (I,\mathcal{G},s,k)$ be an $\eps$-structured BPP instance. Recall that a {\em configuration} of $\mathcal{I}$ is a subset of items $C \subseteq I$ such that $\sum_{\ell \in C} s(\ell) \leq 1$, and $|C \cap G| \leq k(G)$ for all $G \in \mathcal{G}$. Let $\mathcal{C}(\cI)$ be the set of all configurations of $\mathcal{I}$; we use $\cC$ when the instance $\cI$ is clear from the context. Given $S \subseteq I$, a partition $(A_1, \ldots, A_m)$ of $S$ is a {\em packing} (or, a {\em packing of} $S$) if $A_b \in \cC$ for all $b \in [m]$; if $S = I$ then we say that $A$ is a {\em packing of} $\cI$. For $\ell \in I$, let $\mathcal{C}[\ell] = \{C \in \mathcal{C}~|~\ell \in C\}$ be the set of configurations of $\mathcal{I}$ that contain $\ell$. Our algorithm initially solves the following configuration-LP relaxation of the problem.

\begin{equation}
	 \label{C-LP}
	\begin{aligned}
		~~~~~ \min\quad        & ~~~~~\sum_{C \in \mathcal{C}} \bar{x}_C                                                           \\
		\textsf{s.t.\quad} & ~~~\sum_{~C \in \mathcal{C}[\ell]} \bar{x}_C = 1   & \forall \ell \in I~~~~~\\  
		& ~~~~~\bar{x}_C \geq 0 ~~~~~~~~&~~~~~~~~~~~~~~~~~~~~ \forall C \in  \mathcal{C}~~~~
	\end{aligned}
\end{equation}

A solution for the LP~(\ref{C-LP}) assigns to each configuration  
$C \in  \mathcal{C}$ a real number  $\bar{x}_C \in [0,1]$ which indicates the fractional selection of $C$ for the solution such that each item is fully {\em covered}.  For simplicity, denote by $\|\bar{x}\| = \sum_{C \in \mathcal{C}} \bar{x}_C  $ the $\ell_1$-norm of $\bar{x}$ (that is, the objective value in~(\ref{C-LP})).  Note that the configuration-LP~(\ref{C-LP}) has an exponential number of variables; thus, it cannot be solved in polynomial time by applying standard techniques. A common approach for solving such linear programs is to use a {\em separation oracle} for the dual program.

Consider the {\em configuration maximization problem (CMP)} in which we are given a BPP instance $\mathcal{I}=(I,\mathcal{G},s,k)$ and a weight function $w:I\rightarrow \mathbb{R}_{\geq 0}$; the objective is to find a configuration $C\in \mathcal{C}$ such that $\sum_{\ell \in C} w(\ell)$ is maximized. By a well known connection between separation and optimization, an FPTAS for CMP implies an FPTAS for the configuration-LP of $\mathcal{I}$~\cite{grigoriadis2001approximate,fleischer2011tight,grotschel2012geometric,plotkin1995fast}. CMP
can be solved via an easy reduction 
to {\em knapsack with partition matroid}, which admits an FPTAS~\cite{doron2023fptas}. Thus, we have
\begin{lemma}
\label{configurationLP}
There is an algorithm  $\textnormal{\textsf{SolveLP}}$ which given a BPP instance $\cI$ and $\eps>0$, returns  in time $\textnormal{poly} (| \cI |, \frac{1}{\eps})$ a solution for the configuration-LP of $\mathcal{I}$ of value at most $(1+\eps)\OPT$, where $\OPT$ is the value of an optimal solution for the configuration-LP of $\mathcal{I}$.
\end{lemma}

We give the proof of Lemma~\ref{configurationLP} in Appendix~\ref{app:omitted}. A key component in our scheme is the construction of a {\em prototype} of a packing. 

\begin{definition}
	\label{def:prototype}
	Given a BPP instance $\cI$,  a {\em prototype} is a vector $\bar{x} \in \mathbb{R}_{\geq 0}^{\cC}$. 
\end{definition}

In particular, a solution for~\eqref{C-LP} is a prototype. 
In the context of a prototype, each configuration $C \in \mathcal{C}$ is considered as a set of {\em slots}: each slot  $\ell \in C$ is a placeholder for an item, where items that {\em fit} in the place of a slot $\ell$ are those in the group of $\ell$ of smaller or equal size. 
For any $\ell \in I$
define  $\textsf{group}(\ell) = G$, where $G\in \cG$ is the unique group such that $\ell\in G$. Now, we define the subset of items that fit in place of a slot $\ell \in I$ by

\begin{equation}
	\label{fitl}
	\textsf{fit}(\ell) = \{\ell' \in \textsf{group}(\ell)~|~ s(\ell') \leq s(\ell)\}~~~~~~~~~~~~~~\forall \ell \in I.
\end{equation}

 Each configuration is viewed as a set of slots, representing subsets of items that can 
 replace the slots in the actual packing of the instance.
 We associate a polytope with the prototype, in which additional items (i.e, items which to not replace a slot) may be assigned to the unused capacity of a configuration $C$. To enable efficient packing of the instance, these additional items
 must be {\em small} relative to the unused capacity of $C$. To this end, we define the set of items which {\em fit with} a set of slots $C \in \mathcal{C}$ as \begin{equation}
	\label{fitS}
	\textsf{fit}(C) = \{\ell \in I~|~   s(\ell) \leq \min \{\eps^2, \eps \cdot (1-s(C))\}\}~~~~~~~~~~~~~~\forall C \in \mathcal{C}.
\end{equation} 

 In words, $\textsf{fit}(C)$ contains items of sizes smaller than $\eps^2$ and also at most $\eps$-fraction of the unused capacity of $C$. Given a prototype $\bar{x}$, the {\em $\bar{x}$-polytope} is a relaxation for a packing of the instance that assigns items fractionally, either as replacement for slots or in addition to a set of slots. We define the set of {\em types} of the instance $\mathcal{I}$ to be $I\cup \mathcal {\cC}$. The set of types includes the {\em slot-types}, i.e.,  slots (items) in $I$, and {\em configuration-types}, i.e., configurations in $\mathcal{C}$. A point in the $\bar{x}$-polytope
 has an entry for each pair of an item $\ell \in I$ and a type $t \in I \cup {\cal C}$ which
  represents the fractional {\em assignment} of the item to the type. Formally,

 \begin{definition}
 	\label{def:polytope}
 	Given a BPP instance $\cI$, the set $\cC$ of configurations for $\cI$, and a prototype $\bar{x}$ of $\cI$, the $\bar{x}$-polytope is the set containing all points $\bar{\gamma} \in [0,1]^{I \times (I\cup \mathcal{C})}$ which satisfy the following constraints.

\begin{align}
		\displaystyle \bar{\gamma}_{\ell,t} = 0   ~~~~~~~~~~~~~~~~~~~~~~~~~~~~~~~~~~ ~~~~~~~~~~	~~~~~~~~~~\forall \ell \in I , t \in I \cup \mathcal{C}\textnormal{ s.t. } \ell \notin \textnormal{\textsf{fit}}\left(t \right)~~~~ \label{F1}\\
		\rule{0pt}{1.8em}
		\displaystyle \sum_{\ell \in I} \bar{\gamma}_{\ell,C} \cdot  s(\ell) \leq \left(1-s(C)\right) \cdot \bar{x}_C ~~~~~~~~~~~~~~~~~~~	~~~~~~\forall C \in \mathcal{C} ~~~~~~~~~~~~~~~~~~~~~~~~~~~~~~~~ \label{F2}\\
		\rule{0pt}{1.8em}
		\displaystyle    \sum_{\ell \in G} \bar{\gamma}_{\ell,C} \leq \bar{x}_C \cdot \left(k(G)-|C \cap G|\right)~~~~  ~~~~~~~~~~	~~~~~~~~~~\forall G \in \mathcal{G}, C \in \mathcal{C}~~~~~~~~~~~~~~~~~~~~~~~\label{F5}\\
		\rule{0pt}{1.8em}
		\displaystyle \sum_{\ell \in I} \bar{\gamma}_{\ell,j} \leq \sum_{\substack{C \in \mathcal{C}[j]}} \bar{x}_C~~~~~~~~~~~~~~~~~~~~~ ~~~~~~~~~	~~~~~~~~~~\forall j \in I   ~~~~~~~~~~~~~~~~~~~~~~~~~~~~~~~~~\label{F3}\\
		\rule{0pt}{1.8em}
		\displaystyle   \sum_{t \in I \cup \mathcal{C}} \bar{\gamma}_{\ell,t} \geq 1 ~~~~~~~~~~~~~~~~~~~~~~~~~~~~~ ~~~~~~~~~~~~~~~~~~~~	\forall \ell \in I~~~~~~~~~~~~~~~~~~~~~~~~~~~~~~~~~ \label{F4}
	\end{align}
	 \end{definition}
	
	Constraints~(\ref{F1}) indicate that an item $\ell \in I$ cannot be assigned to type $t$ if $\ell$ does not fit in~$t$. 
	Constraints~(\ref{F2}) set an upper bound on the total (fractional) size of items assigned to each configuration-type $C \in \mathcal{C}$. This bound is equal to the residual capacity of $C$ times the number of bins packed with $C$, given by $\bar{x}_C$.
	Constraints~(\ref{F5}) bound the number of items in each group $G$ assigned to configuration-type $C$; at most $k(G)-|C \cap G|$ items in $G$ can be added to $C$ without violating the cardinality constraint of $G$. Constraints~(\ref{F3}) bound the number of items assigned to slot-type $j \in I$ by the total selection of configurations containing $j$ in $\bar{x}$.
	Finally, constraints~(\ref{F4}) guarantee that each item is fully assigned to the types.
	We note that if a prototype $\bar{x}$ is a solution for (\ref{C-LP}) then
	the $\bar{x}$-polytope is non-empty;
	in particular, it contains the point 
	$\bar{\gamma}$
	where $\bar{\gamma}_{\ell, j}=1$ $\forall \ell, j \in I$ 
	such that 
	$\ell=j$, and $\bar{\gamma}_{\ell, t}=0$ otherwise.
	 
Let $\textsf{supp}(\bar{x}) = \{C \in \mathcal{C}~|~ \bar{x}_C>0\}$ be the {\em support} of $\bar{x}$. Throughout this paper, we use prototypes $\bar{x}$ for which $\textsf{supp}(\bar{x})$ is polynomial in the input size; thus, these prototypes have
sparse representations.
The next lemma shows that if the prototype $\bar{x}$ has a small support, and each configuration in the support contains a few items then the vertices of the $\bar{x}$-polytope are almost integral.
Thus, given a vertex $\bar{\lambda}$ of such $\bar{x}$-polytope, the items assigned fractionally by $\bar{\lambda}$ can be packed using only a small number of extra bins.

 \begin{lemma}
 	\label{O(1)}
 	Let  $\cI$ be a BPP instance, $k \geq 1$ an integer and $\bar{x}$ a prototype of $\mathcal{I}$ such that $|C| \leq k$
 	for all $C \in \textnormal{\textsf{supp}}(\bar{x})$, and $\bar{x}_C \in \mathbb{N}$.
 	Then given a vertex $\bar{\lambda}$
 	of the $\bar{x}$-polytope for which constraints \eqref{F4} hold with equality, $$\left|\left\{\ell \in I~|~ \exists t \in I \cup \cC \text{ s.t. } \bar{\lambda}_{\ell,t} \in (0,1)\right\}\right| \leq 8k^2 \cdot |\textnormal{\textsf{supp}}(\bar{x})|^2.$$ 
 \end{lemma}

The proof of Lemma~\ref{O(1)} bears some similarity to a proof of~\cite{DKS21_arXiv}, which shows the integrality properties of a somewhat different polytope. We give the complete proof in Appendix~\ref{sec:poly}. We now describe the main components of our scheme, which converts an initial prototype $\bar{x}$ (defined by a solution for the configuration-LP) into another prototype $\bar{z}$, and then constructs a packing based on $\bar{z}$. The necessary conditions for a prototype allowing to construct an efficient packing are given below.  Let $Q:(0,0.1)\to \mathbb{R}$ where $Q = \exp(\eps^{-17})$ for all $\eps\in (0,0.1)$. 

\begin{definition}
	\label{def:goodprototype}
	Given $\eps \in (0, 0.1)$ and an $\eps$-structured BPP instance $\cI$, a prototype $\bar{x}$ of $\cI$ is a {\em good prototype} if the $\bar{x}$-polytope is non-empty, $|\textnormal{\textsf{supp}}(\bar{x})| \leq Q(\eps)$, and $|C| \leq \eps^{-10}$ for all $C \in \textnormal{\textsf{supp}}(\bar{x})$. 
\end{definition}
Observe that a solution $\bar{x}$ for the configuration-LP of $\mathcal{I}$
is not necessarily a good prototype, since it may have a support of large size.
Given this initial prototype, our scheme generates a good prototype in two steps.
In the first step, algorithm {\sf Evict} constructs an intermediate prototype $\bar{y}$ 
which selects only configurations containing
a small number of items. This results in a small increase in the total number of bins used. Also, the $\bar{y}$-polytope is non-empty. Given $\eps>0$, we say that an item $\ell\in I$ is $\eps$-large if $s(\ell)\geq \eps^2$. We use $L(\eps,\II)$ to denote the set of $\eps$-large items of an instance $\II$. If $\II$ and $\eps$ are known by context we simply use $L$ (instead of $L(\eps,\II)$). The properties of $\bar{y}$ are summarized in the next lemma (see the details of algorithm {\sf Evict} in Section~\ref{sec:eviction}).

\begin{lemma}
	\label{lem:eviction}
	There is an algorithm $\textsf{Evict}$ that given $\eps \in (0,0.1)$ such that $\eps^{-1}\in \mathbb{N}$, an $\eps$-structured BPP instance $\cI$,  and a solution $\bar{x}$ for the configuration-LP~\eqref{C-LP}, returns in time $\textnormal{poly} (| \cI |, \frac{1}{\eps})$ a prototype $\bar{y}$ which satisfies the following. (i) There exists $\bar{\gamma}$ in the $\bar{y}$-polytope such that $\bar{\gamma}_{\ell,j} = 0$ for all $\ell,j \in I, \ell \neq j$; (ii) for all $C \in \textnormal{\textsf{supp}}(\bar{y})$ it holds that $|C| \leq \eps^{-10}$, and $s(C \setminus L) \leq \eps$; (iii) $\|\bar{y}\| \leq (1+\eps)\|\bar{x}\| ; $ (iv) $\sum_{C\in \cC[\ell]} \by_C\leq 2$ for every $\ell \in I$.
\end{lemma}

Observe  that  property (i) of Lemma~\ref{lem:eviction}	allows  $\bar{\gamma}_{\ell,C}>0$ for item $\ell \in  	I$ and a set of slots $C\in \cC$. Note that \textsf{Evict} does not return the vector $\bar{\gamma}$ and only guarantees its existence. 
In the second step, our scheme uses algorithm $\textsf{Shift}$ to obtain a good prototype. This is done by a novel constructive use of fractional grouping over a small subset of carefully chosen groups. This step is formalized in the next lemma.

\begin{lemma}
\label{lem:Cnf}	
Given $\eps \in (0, 0.1)$ such that $\eps^{-1}\in \mathbb{N}$, let $\cI$ be an $\eps$-structured BPP instance. Furthermore, let $\bar{y}$ be a prototype of $\mathcal{I}$ with non-empty $\bar{y}$-polytope which satisfies the following. $(i)$ For all $C \in \textsf{supp}(\bar{y})$ it holds that $|C| \leq \eps^{-10}$ and $s(C \setminus L) \leq \eps$;
$(ii)$ there exists $\bar{\gamma}$ in the $\bar{y}$-polytope such that 
$\bar{\gamma}_{\ell,j} = 0$ for all $\ell,j \in I$ where $\ell \neq j$; (iii) $\sum_{C\in \cC[\ell]} \by_C\leq 2$ for every $\ell \in I$. Then, given $\cI$, $\bar{y}$, and $\eps$, Algorithm $\textsf{Shift}$ returns in time $\textnormal{poly} (| \cI |, \frac{1}{\eps})$ a good prototype $\bar{z}$ such that $\|\bar{z}\| \leq (1+5\eps)\|\bar{y}\|+Q(\eps)$.
\end{lemma}

Algorithm $\textsf{Shift}$ is presented in Section~\ref{sec:AlgPrototype}. Given a good prototype $\bar{z}$, our scheme proceeds to find a partition of the items into slot-types and configuration-types using the following construction.

 For configurations $S,C \in \mathcal{C}$, we say that $S$ is {\em allowed} in $C$ if each item $\ell \in S$ can be mapped to a distinct slot $j \in C$, such that $\ell \in \textsf{fit}(j)$. 
 We consider a packing of a subset of $I$ such that the bins in the packing are partitioned into a bounded number of {\em categories}. Each category is associated with 
 $(i)$ a configuration $C \in \cC$ such that all bins in the category are allowed in $C$, and $(ii)$ a {\em completion}: a subset of (unpacked) items bounded by total size and number of items per group, where each item fits with $C$. Also, we require that each item $\ell \in I$ is either in this packing or in a completion of a category. The above constraints are analogous to constraints \eqref{F1}-\eqref{F4} of the $\bar{x}$-polytope, that is used for finding  an assignment of the items to slots and configurations. Formally,

 \begin{definition}
 	\label{def:partition}
 	
 	Given $\eps \in (0, 0.1)$ and an $\eps$-structured BPP instance $\cI = (I,\cG,s,k)$, an $\eps$-{\em nice partition} $\mathcal{B}$ of $\mathcal{I}$ is a packing $(A_1,..., A_m)$ of a subset of $I$, a subset of configurations $\mathcal{H} \subseteq \cC$, categories $\left(B_C\right)_{C \in \mathcal{H}}$, and completions $\left(D_C\right)_{C \in \mathcal{H}}$ such that the following conditions hold.
 	
 	\begin{itemize}
 		\item $|\mathcal{H}| \leq \eps^{-22}Q^2(\eps)$.
 		
 		\item $\{B_C\}_{C \in \mathcal{H}}$ is a partition of $\{A_i~|~i \in [m]\}$
 		
 		\item $\{D_C\}_{C \in \mathcal{H}}$ is a partition of $I \setminus \bigcup_{i \in [m]} A_i$.
 		
 		\item For any $C \in \mathcal{H}$ and $A \in B_C$ it holds that $A$ is allowed in $C$
 		
 		\item  For any $C \in \mathcal{H}$ and $G \in \cG$ it holds that:
 		
 		\begin{enumerate}
 			\item  $D_C \subseteq \textnormal{\textsf{fit}}(C)$.
 			
 			\item $s(D_C) \leq \left(1-s(C)\right) \cdot |B_C|$.
 			
 			\item $|D_C \cap G| \leq |B_C| \cdot \left(k(G)-|C \cap G|\right)$.
 		\end{enumerate}
 	\end{itemize} 
 
 The {\em size} of $\mathcal{B}$ is $m$, and for all $C \in \mathcal{H}$ we say that $B_C$ is the category of $C$ and $D_C$. 
 \end{definition}  
 
  Algorithm {\sf Partition} initially rounds up the entries of $\bar{z}$ to obtain the prototype $\bar{z}^*$.
  It then finds a vertex $\bar{\lambda}$ of the $\bar{z}^*$-polytope, which is almost integral by  Lemma~\ref{O(1)}. Thus, with the exception of a small number of items, each item is fully assigned either to a slot or to a configuration. Algorithm {\sf Partition} uses $\bar{\lambda}$ to construct an $\eps$-nice partition. We generate a category for each $C \in \textnormal{\textsf{supp}}(\bar{z}^*)$ and define $D_C = \{ \ell \in I~| \bar{\lambda}_{\ell,C}=1 \}$ to be the set of all items assigned to $C$. We also generate $\bar{z}^*_C$ copies (bins) of each configuration and replace its slots by items via matching.
  
\begin{lemma}
	\label{FromPolytope}
	There is an algorithm $\textsf{Partition}$ that given $\eps \in (0, 0.1)$ such that $\eps^{-1}\in \mathbb{N}$, an $\eps$-structured BPP instance $\cI$, and a good prototype $\bar{z}$ of $\mathcal{I}$, returns  in time $\textnormal{poly} (| \cI |, \frac{1}{\eps})$ an $\eps$-nice partition of $\mathcal{I}$ of size at most $\|\bar{z}\|+\eps^{-22}Q^2(\eps)$.
\end{lemma}

Algorithm $\textsf{Partition}$ is presented in Section~\ref{sec:FromPolytope}. Given an $\eps$-nice partition of size $m$, a packing of the instance in roughly $m$ bins is obtained using the next lemma. 

\begin{lemma}
 \label{lem:GREEDY}
 There is a polynomial-time algorithm \textsf{Pack} which given $\eps \in (0, 0.1)$ such that $\eps^{-1}\in \mathbb{N}$, an $\eps$-structured BPP instance $\cI$, and $\eps$-nice partition of $\mathcal{I}$ of size $m$, returns  in time $\textnormal{poly} (| \cI |, \frac{1}{\eps})$ a packing of $\mathcal{I}$ in at most $(1+2\eps)m+2\eps^{-22}Q^2(\eps)$ bins. 
 \end{lemma}

Algorithm \textsf{Pack} utilizes Algorithm \textsf{Greedy} to add the items in a completion of a category to the bins of this category, possibly using a few extra bins. Algorithm \textsf{Pack} is presented in Section~\ref{sec:alg_pack} and Algorithm \textsf{Greedy} is presented in Section~\ref{sec:greedy}. Using the above components, we construct an AFPTAS for $\eps$-structured BPP instances. The pseudocode of the scheme is given in Algorithm~\ref{alg:Fscheme}. We summarize in the next result.

\begin{algorithm}[h]
	\caption{$\textsf{AFPTAS}({\II}, \eps)$}
	\label{alg:Fscheme}
	
	\SetKwInOut{Input}{Input}
		\SetKwInOut{Output}{Output}
	\Input{An $\eps$-structured instance $\II$ and $\eps\in (0,0.1)$ such that $\eps^{-1}\in \mathbb{N}$}
	\Output{A packing of $\II$}

	Find a solution for the configuration-LP of $\mathcal{I}$; that is, $\bar{x} = \textsf{SolveLP}(\eps,\mathcal{I})$\label{LP}

	Let $\bar{y} = \textsf{Evict}(\eps,\cI,\bar{x})$ \label{step:evicAFPTAS}
	
	Find a good prototype $\bar{z} = \textsf{Shift}(\eps,\cI,\bar{y})$ \label{GetPolytope}

	Find an $\eps$-nice partition $\mathcal{B}$ of $\mathcal{I}$ by $\textsf{Partition}(\eps,\cI,\bar{z})$ \label{step:BPPnice}

	Return 
	 $\Phi = \textsf{Pack}(\eps,\mathcal{I},\mathcal{B})$  \label{step:greedy1}

\end{algorithm}

\begin{lemma}
\label{lem:AFPTAS}
Given $\eps \in (0, 0.1)$ such that $\eps^{-1}\in \mathbb{N}$, and an $\eps$-structured  BPP instance $\cI$, Algorithm~\ref{alg:Fscheme} returns in time $\poly (| \cI |, \frac{1}{\eps})$ a packing of  $\II$ in at most $(1+60\eps)\OPT(\II)+Q^3(\eps)$ bins.
\end{lemma}

The proof of Lemma~\ref{lem:AFPTAS} follows immediately from Lemmas~ \ref{configurationLP},  ~\ref{lem:eviction}, \ref{lem:Cnf}, ~\ref{FromPolytope}, and \ref{lem:GREEDY}. 
A formal proof is given in Appendix~\ref{app:omitted}. 
The next lemma is an immediate consequence of Lemmas~\ref{lem:AFPTAS} and~\ref{lem:reductionReconstruction}. 
\begin{lemma}
	\label{lem:gen_afptas}
There is an algorithm $\genalg$ which given a BPP instance $\cj$ and $\eps\in (0,0.1)$ such that $\eps^{-1}\in \mathbb{N}$, returns in time $\poly(|\cj|,\frac{1}{\eps})$ a packing of $\cj$ using at most $(1+130\eps)\cdot \OPT(\cj) + 3\cdot  Q^3(\eps)$  bins. 
\end{lemma}
We give the proof of Lemma~\ref{lem:gen_afptas} in Appendix~\ref{app:omitted}. 
Given a BPP instance $\mathcal{J}$, Theorem~\ref{thm:main} follows by applying $\genalg$ to $\mathcal{J}$ taking $\eps = \left(  \ln \ln W \right)^{-\frac{1}{17}}$, where 
$W=s(I) + V(\II) +\exp(\exp(100^{17}))=\Theta(\OPT(\cj))$. By the above selection of $\eps$, it holds that 
$Q^3(\eps)=o(\OPT(\mathcal{J}))$ and $\eps  \OPT(\mathcal{J})=o(\OPT\left( \mathcal{J}) \right)$. This yields the  $o(\OPT(\II)))$ additive approximation ratio. The proof of Theorem~\ref{thm:main} is given in Appendix~\ref{app:omitted}.

\section{Algorithm \textsf{Evict}}
\label{sec:eviction}

Our scheme uses algorithm \textsf{Evict} for reducing the number of items in each configuration in the support of a prototype. Let $\bar{x}$ be a solution for~\eqref{C-LP}. 
 Given a configuration $C \in \mathcal{C}$, algorithm \textsf{Evict} replaces $C$ by a set of slots consisting of items in $C$.
 Let $\ell_1, \ldots, \ell_{r}$ be the items in $C \setminus L$ 
 sorted in non-increasing order by sizes. Denote by $h \in [r]$ the minimum index of an item such that any item of larger index $\ell \in \{\ell_{h+1}, \ldots, \ell_{r}\}$ 
has size at most $\eps$ of the residual capacity from packing 
$C \cap L \cup \{\ell_1, \ldots, \ell_{h}\}$.\footnote{If $h=0$ then $\{\ell_1, \ldots, \ell_{h}\} =\emptyset$.}
 In case $h$ is larger than $\alpha= \eps^{-5}$ set
 $h = \alpha$. Formally, 
 \begin{equation}
	\label{h}
h= \min \bigg\{ \alpha, \min \big\{0 \leq h' \leq r~|~ \forall j \in \{h'+1, \ldots,r\} ~: \ell_j \in 
\textsf{fit} (C \cap L \cup \{\ell_1, \ldots, \ell_{h'}\} ) \big\}\bigg\}
\end{equation}

 Now, define $U_C = \{\ell_1, \ldots, \ell_h\}$ as the set of the first $h$ items in the above order. 
 Let $R(C) = C \cap L \cup U_C$ be the items in $C$ which are considered as slots by algorithm \textsf{Evict}. Also, let $\mathcal{L}_C = \{\ell \in U_C~|~s(\ell) \geq \frac{1-s(R(C))}{\eps}\}$ be all items in $U_C$ of relatively large size. 

Consider two items $\ell_i, \ell_k \in U_C$ such that $|U_C| > k > i+\eps^{-4}$. By \eqref{h}, the distance between the items in the sorted order implies a considerable difference in their sizes; namely, $s(\ell_i) > \frac{s(\ell_k)}{\eps^2}$.
Furthermore, the smaller item cannot be too small: $s(\ell_k) > \eps\left(1-s(R(C))\right)$, since otherwise $\ell_k \notin U_C$. It follows that $\ell_i \in \mathcal{L}_C$, which suggests that $\mathcal{L}_C$ is not much smaller than $U_C$.  

\begin{lemma}
	\label{lem:evictionHelp}
For any configuration $C \in \mathcal{C}$ such that $|U_C| = \alpha$, it holds that $|\mathcal{L}_C| \geq \eps^{-4}$.
\end{lemma}

Algorithm \textsf{Evict} constructs from $\bar{x}$ a prototype $\bar{y}$ as follows. For any $C \in \textsf{supp}(\bar{x})$, the value $\bar{x}_C$ is split among at most $|\mathcal{L}_C|$ new configurations, each generated by removing from $C$ a subset of items. Initially, all items in $C \setminus R(C)$ are removed. Then, an item from $\mathcal{L}_C$ may be removed, depending on whether $|U_C| = \alpha$ or not.

If $|U_C| = \alpha$, then the next largest item not in $U_C$ (that is, $\ell_{h+1}$) does not necessarily belong to $\textsf{fit}(R(C))$. However, because the size of items in $\mathcal{L}_C$ is relatively large, removing an item $\ell \in \mathcal{L}_C$ from $C$ guarantees that $C \setminus R(C) \subseteq \textsf{fit}(R(C) \setminus \{\ell\})$. This is useful for proving that the prototype $\bar{y}$ returned by algorithm \textsf{Evict} has non-empty $\bar{y}$-polytope. Otherwise, we have $|U_C| < \alpha$. By \eqref{h}, for any item  $\ell \in C \setminus R(C)$ it holds that $\ell \in \textsf{fit}(R(C))$. Hence, in this case, no need to remove an item from $\mathcal{L}_C$. 
In the above procedure we use a vector per configuration. Given a configuration $C\in \mathcal{C}$, define the {\em relaxation of $C$} as the following vector $\bar{w}^C \in \mathbb{R}_{\geq 0}^{\mathcal{C}}$.

\begin{enumerate}
	\item  	\label{eq:y=a} If $|U_C| = \alpha$, then for any $\ell \in \mathcal{L}_C$ define $\bar{w}^C_{R(C) \setminus \{\ell\}} = \frac{1}{|\mathcal{L}_C|-1}$. For any other configuration $C' \in \mathcal{C}$ such that there is no $ \ell \in \mathcal{L}_{C'}$ satisfying $C=R(C') \setminus \{\ell\}$, set $\bar{w}^C_{C'} = 0$. 

\item 	\label{eq:y<a} If $|U_C| < \alpha$, then set $\bar{w}^C_{R(C)} = 1$. For all $C' \in \mathcal{C} \setminus \{R(C)\}$ set $\bar{w}^C_{C'} = 0$.
\end{enumerate}

 The relaxation vector of a configuration 
 satisfies several properties that are useful for our algorithm. In particular, 
 the total size of $\bar{w}^C$ is only slightly larger than $1$, all items in $R(C)$ are
 covered, and uncovered items fit for placement with the slots. Finally, small items are added only to configurations that are almost full. This is formalized in the next lemma. 
 \begin{lemma}
	\label{lem:wProperties}
	For any $C \in \mathcal{C}$, the following hold for $\bar{w}^C$, the relaxation of $C$. 
	
	\begin{enumerate}
		
		\item $\|\bar{w}^C\| \leq 1+2\eps^4$.
		
		\item For all $\ell \in R(C)$,
		 it holds that 
		$\sum_{C' \in \mathcal{C}[\ell]} \bar{w}^C_{C'} \geq 1$. 
		
		\item For all $C ' \in \textnormal{\textsf{supp}}(\bar{w}^C)$,
		 it holds that 
		$C \setminus R(C) \subseteq\textnormal{\textsf{fit}}(C')$ and $s(C \setminus R(C)) \leq 1-s(C')$. 
		
		\item  For all $C ' \in \textnormal{\textsf{supp}}(\bar{w}^C)$ it holds that $s(C' \setminus L) \leq \eps$. 
		
		\item  For all $C ' \in \textnormal{\textsf{supp}}(\bar{w}^C)$ and $G \in \cG$ it holds that $\left|G \cap \left(C \setminus R(C)\right)\right| \leq k(G) -|G \cap C'|$. 
	\end{enumerate} 
\end{lemma}

\begin{algorithm}[htb]
	\caption{$\textsf{Evict} (\bar{x}$)}
	\label{Alg:eviction}

	For every $C\in \textsf{supp}(\bar{x})$ compute $\bar{w}^C$, the relaxation of $C$.\label{step:forsupp}\label{step:y}

	return $\bar{y} = \sum_{C \in \textsf{supp}(\bar{x})} \bar{x}_C \cdot \bar{w}^C$.\label{step:return}
	\end{algorithm}

Algorithm \textsf{Evict} computes a linear combination of the relaxations of all configurations $C \in \mathcal{C}$, where the coefficient of $\bar{w}^C$ is $\bar{x}_C$. The pseudocode of \textsf{Evict} is given in Algorithm~\ref{Alg:eviction}, and the proof of Lemma~\ref{lem:eviction} is given in Appendix~\ref{app:omittedEvict}.

 \section{Algorithm $\textsf{Shift}$}
\label{sec:AlgPrototype}

 Let $\bar{y}$ be a prototype satisfying the conditions of Lemma~\ref{lem:Cnf}.
 Algorithm \textsf{Shift} constructs from  $\bar{y}$ a good 
 prototype $\bar{z}$ by reducing the size of the support of $\bar{y}$. 
 The algorithm uses {\em classes}, i.e., partitions of groups into sets, where each set contains items of (roughly) similar size. Each class has up to $\eps^{-10}$
 representatives to which items in this class are mapped. In the prototype constructed by \textsf{Shift}, each configuration $C \in \textsf{supp}(\bar{y})$ is replaced by a set of
 representatives to which the items in $C$ are mapped.
 Recall that $L = \{\ell \in ~|~ s(\ell) > \eps^2\}$ is the set of {\em large} items in $I$; let $I \setminus L$ be  the set of {\em small} items. 
 
  We define classes for two types of groups, as explained below. Given an item $\ell \in I$,
   the {\em frequency} of $\ell$ is given by $f_{\bar{y}}(\ell) = \sum_{C \in \mathcal{C}[\ell]} \bar{y}_C$; that is, the fractional number of configurations
 containing $\ell$  w.r.t $\bar{y}$. 
 The frequency of a subset of items $I' \subseteq I$ is then
$f_{\bar{y}}(I') = \sum_{\ell \in I'} f_{\bar{y}}(\ell)$. 
To obtain a prototype $\bar{z}$ with small support, only items from a small number of classes may remain in the constructed configurations. Thus, items belonging to 
groups with low frequencies of small items, or groups with no large items, are discarded 
from their hosting configurations.  
 
 Specifically, for any $G \in \cG$, let $f_{\bar{y}}(G \setminus L)$ be the frequency of small items in $G$. Now, let $\eta = \min \{|\cG|,\eps^{-12}\}$, then the {\em significant} groups $\mathcal{G}_{\eta} = \{G_1, \ldots, G_{\eta}\}$ are the first $\eta$ groups in $\cG$, where the groups are sorted in non-increasing order by the frequency of small items in each.

 \begin{lemma}
 	\label{lem:significantDiscard}
 	
 	The following properties hold for $\bar{y}$.	\begin{enumerate}
 		\item $\sum_{\ell \in I \setminus L} f_{\bar{y}}(\ell) \cdot s(\ell) \leq \eps \cdot \|\bar{y}\|$.
 		
 		\item  For all $G \in \mathcal{G} \setminus  \mathcal{G}_{\eta}$ it holds that $f_{\bar{y}}(G \setminus L) \leq \eps \cdot \|\bar{y}\|$. 
 	\end{enumerate}

 \end{lemma} Recall that any configuration $C \in \textsf{supp}(\bar{y})$ satisfies (a) $s(C \setminus L) \leq \eps$, and (b) $|C|\leq \eps^{-10}$.  The first property in Lemma~\ref{lem:significantDiscard} follows from (a), and the second property follows from the selection of $\eta$ and (b).  

In the prototype $\bar{z}$ constructed by algorithm {\sf Shift},
the selection of the empty configuration (i.e., $\bar{z}_{\emptyset}$) is increased. In the $\bar{z}$-polytope,  this enables to fully assign the small items from non-significant groups to the empty-configuration type, as these items are discarded from the 
configurations selected by 
$\bar{z}$. Using Lemma~\ref{lem:significantDiscard}, a slight increase in  the selection of the empty-configuration suffices for the $\bar{z}$-polytope to be non-empty. However, this can be done effectively only for small items. Hence, in addition to the significant groups, we take for the configurations of the constructed prototype also groups containing large items. Formally, Define the {\em massive} groups as all groups containing at least one large item: 

\begin{equation}
 	\label{Agroups}
 	\cA = \{G \in \cG~|~ G \cap L \neq \emptyset\}.
 \end{equation} Recall that $\cI$ is an $\eps$-structured instance; thus, $K(\eps) = \eps^{-\eps^{-2}}$ is an upper bound on $|\cA|$. We define the {\em important groups} as $\cG_{\eta} \cup \cA$ containing all significant and massive groups. The construction of classes relies on the fractional grouping technique, which finds a partition of an important group into a small number of sets, each of roughly the same frequency according to $\bar{y}$. Items of the same set can be treated as having the same size, as in the classic linear shifting technique \cite{fernandez1981bin}. 
 
  We use fractional grouping separately for each 
  important group. Consider an important group $G \in \mathcal{G}_{\eta} \cup \cA$.  For the construction of  the small classes of $G$, we use the following inductive definition.\footnote{Recall that for any two items $\ell,j \in I$ such that $s(\ell) > s(j)$ it holds that $\ell < j$.} Let $\min (G)$ ($\max (G)$) denote the minimal (maximal) index of an item in $G$. Define $q_0 = 1+\max (G)$. For $i \geq 1$, let $\upsilon = 3\eps^{-2}$, and 
  
  \begin{equation}
	\label{eq:q}
		q_i = \max \biggl\{  \ell \in G ~\bigg|~ f_{\bar{y}} \bigl(\{j \in G ~|~ 
	\ell	  \leq j < q_{i-1} \}\bigr)~ \geq ~\eps^{\upsilon} \cdot \|\bar{y}\| \biggr\}.
\end{equation} 

If the maximum in \eqref{eq:q} is defined over a empty set then  $q_i =  \min (G)$,
and we set $\tau(G) = i$ to be the number of small classes of $G$.  In words, $q_i$ is the maximal index of an item in $G$ such that the frequency of the set of items with indices 
in the range $\{ q_i, \ldots , q_{i-1}-1 \}$
is at least $\eps^{\upsilon} \cdot \|\bar{y}\|$. Now, for any $i \in [\tau(G)]$, define $\Delta_i(G) = \{\ell \in G ~|~ q_i  \leq \ell <  q_{i-1} \}$ as the $i$-th {\em class} of $G$, which contains all items in $G$ of indices smaller than $q_{i-1}$ and at lease $q_i$. 
 
\begin{observation}
	\label{ob:FG}
	Given $G \in \cG_{\eta} \cup \cA$, the classes $\Delta_1(G), \ldots, \Delta_{\tau(G)}(G)$  of $G$ satisfy the following. 
	
\begin{enumerate}
	\item $\tau(G)\leq \eps^{-\upsilon-10}$
	
	\item For all $i \in [\tau-1]$ and $\ell \in \Delta_i(G)$ it holds that $\Delta_{i+1}(G) \subseteq \textnormal{\textsf{fit}}\left( \ell\right)$.
	
	\item For all $i \in [\tau]$ it holds that $f_{\bar{y}}(\Delta_i(G)) \leq \eps^{\upsilon} \cdot \|\bar{y}\|+2$.
	
	\item For all $i \in [\tau-1]$ it holds that $f_{\bar{y}}(\Delta_i(G)) \geq \eps^{\upsilon} \cdot \|\bar{y}\|$.
\end{enumerate}
\end{observation}

Observation~\ref{ob:FG} follows from \eqref{eq:q}. Specifically, the third propertyholds since  $f_{\bar{y}}(\ell) \leq 2$ for all $\ell \in I$, by the properties of $\bar{y}$. By Observation~\ref{ob:FG}, for any $G \in \cG_{\eta} \cup \cA$ and $ i \in [\tau(G)-1]$, each item in $\Delta_{i+1}(G)$ fits in place of a slot $\ell \in \Delta_i(G)$. This is useful for constructing a prototype $\bar{z}$ such that the $\bar{z}$-polytope contains a point in which items in $\Delta_{i+1}(G)$ are assigned to at most $\eps^{-10}$ representatives from $\Delta_i(G)$. Thus, given at most $\eps^{-10}$ representatives $\Delta'_i(G) \subseteq \Delta_i(G)$, we can define $\bar{z}$ such that $f_{\bar{z}} \left( \Delta'_i(G) \right) \approx f_{\bar{y}} \left(\Delta_{i+1}(G)\right)$ and the size of the support of $\bar{z}$ depends solely on a function of $\eps$. Finally, define the set of classes as all the classes of important groups; that is,

\begin{equation}
	\label{eq:classS}
	\mathcal{Q} =  \bigcup_{G \in \cG_{\eta}\cup \cA} \{\Delta_{1}(G), \ldots, \Delta_{\tau(G)}(G)\}. 
\end{equation} Observe that $\eps^{-\upsilon-10}\cdot \left( \eta+K(\eps) \right)$ is an upper bound on $|\mathcal{Q}|$: there are $\eta$ significant groups and at most $K(\eps)$ massive groups because $\cI$ is $\eps$-structured; in addition, each group has at most $\eps^{-\upsilon-10}$ classes, by Observation~\ref{ob:FG}. Note that the classes in $\mathcal{Q}$ do not intersect; moreover, except for the small items from non-important groups which are absent in these classes, each item belongs to exactly one class in $\mathcal{Q}$.

We now define a mapping from configurations to representatives from a subset of the classes. Observe that for groups $G \in \cG$ with $k(G) >1$, a configuration may contain more than  one item in $G$. Hence, given a configuration $C \in \textsf{supp}(\bar{y})$ and a class $\Phi \in \mathcal{Q}$, items in $C \cap\Phi$ are mapped to the subset of  $|C \cap \Phi|$  items in $\Phi$ of minimal size; that is, the number of items which belong to $\Phi$ in $C$ does not change by the mapping. We construct below a mapping for all items from important groups.

Given a subset of items $S \subseteq I$, sorted in increasing order by item indices, let 
$\textsf{last}_k(S)$ be the set of the last $k$ items in $S$, where $k \in \{0,1, \ldots , |S|\}$. For example, given $S = \{9, 13, 88, 103, 1093\}$, we have that $\textsf{last}_3(S) = \{88, 103, 1093\}$. Given $\Phi \in \mathcal{Q}$ and $C \in \textsf{supp}(\bar{y})$, the items in $C \cap \Phi$ are mapped to $\textsf{last}_{|C \cap \Phi|}(\Phi)$, the last $|C \cap \Phi|$ items in $\Phi$; thus, the number of items in $C \cap \Phi$ does not change.  Finally, given $C \in \textsf{supp}(\bar{y})$, the {\em projection} of $C$ is: 

\begin{equation}
		\label{eq:mapS,eq:mapM}
	P(C) = \bigcup_{\Phi\in \mathcal{Q}} \textsf{last}_{|C \cap\Phi|}(\Phi). 
\end{equation} By the next lemma, the prototype induced by replacing each configuration $C \in \textsf{supp}(\bar{y})$ by the projection of $C$ has a support of a small size.

 \begin{lemma}
	
	\label{lem:P(C)}
	For all $C \in \textnormal{\textsf{supp}}(\bar{y})$ the following hold for the projection of $C$.
	\begin{enumerate}

		\item   $s\left(P(C)\right) \leq s(C)$.
		
		\item $\textnormal{\textsf{fit}}(C)  \subseteq \textnormal{\textsf{fit}}(P(C))$. 
		
		\item $|P(C)| \leq \eps^{-10}$.
		
		\item $|\{P(C')~|~C' \in \textnormal{\textsf{supp}}(\bar{y})\}| \leq Q(\eps)-3K(\eps)$. 
		
		\item For all $G \in \cG$ it holds that $|G \cap P(C)| \leq |G \cap C|$. 
	\end{enumerate}
\end{lemma}

Properties 1-3 of Lemma~\ref{lem:P(C)} hold since for any $C \in \textsf{supp}(\bar{y})$ and $\Phi \in \mathcal{Q}$, the items in $C \cap \Phi$ are mapped by $P$ to $|C \cap \Phi|$ items of smaller or equal size from the same group. The fourth property holds since the number of classes is bounded, and $|C \cap \Phi| \leq \eps^{-10}$. The last bound follows 
by noting that the items of each class are mapped to the same number of items, whereas items from non-important groups are discarded from the configuration. 

 The projection function is utilized to construct a good prototype, where the fractional selection of a configuration $C \in \textsf{supp}(\bar{y})$ is added to the corresponding entry of the projection $P(C)$. To guarantee that the returned prototype $\bar{z}$ has non-empty $\bar{z}$-polytope, the selection of each configuration is slightly increased. In addition, the selection of the empty-configuration and of a small number of configurations of one slot is further augmented.

\begin{algorithm}[htb]
	\caption{$\textsf{Shift} (\bar{y}$)}
	\label{Alg:shifting}

Construct the classes $\mathcal{Q}$.\label{step:classesQ}

	Define $\bar{u} =  \left(1+\frac{2\eps^{-\upsilon}}{\|\bar{y}\|}\right)\left(  \sum_{C \in \textsf{supp}(\bar{y})} \bar{y}_C \cdot \mathbbm{1}^{P(C)}  \right)$.\label{step:u}

	return $\bar{z} = \bar{u}+\left( 4 \eps \cdot \|\bar{y}\| +\eps^{-3} \right)\cdot \mathbbm{1}^{\emptyset} +\sum_{G \in \cG_{\eta} \cup \cA} \left(  \eps^{\upsilon} \cdot \|\bar{y}\| + 2 \right)\cdot \mathbbm{1}^{r(G)}$.\label{step:shiftreturn}
\end{algorithm}
 
 Specifically, for each $G \in \cG$, let $r(G) = \{ \min G \}$ be the {\em maximal slot} of $G$; this is a configuration containing the item of maximal size (and minimal index) in $G$. For each $C \in \cC$ define the {\em indicator} of $C$, a vector $\mathbbm{1}^C \in [0,1]^{\cC}$ such that $\mathbbm{1}^C_C = 1$, and for all $C' \in \cC \setminus \{C\}$ define $\mathbbm{1}^C_{C'} = 0$. Algorithm \textsf{Shift} computes the prototype $\bar{z}$ in two steps. First, we construct a prototype $\bar{u}$ as  a linear combination of the indicators of all configurations $C \in \cC$, where $C$ contributes $ \left(1+\frac{2\eps^{-\upsilon}}{\|\bar{y}\|}\right) \cdot \bar{y}_C$ to $\bar{u}_{P(C)}$. Then, given $\bar{u}$, we construct $\bar{z}$ by adding $4\eps \|\bar{y}\|+\eps^{-3}$ to $\bar{u}_{\emptyset}$ and adding $ \eps^{\upsilon} \cdot \|\bar{y}\| + 2$ to $r(G)$ for each $G \in \cG_{\eta} \cup \cA$. It can be easily shown that $\|\bar{z}\| \leq (1+5\eps)\|\bar{y}\|+Q(\eps)$. The pseudocode of algorithm \textsf{Shift} is given in Algorithm~\ref{Alg:shifting}. The proof of Lemma~\ref{lem:Cnf} is given in Appendix~\ref{app:omittedShifting}.

 To show that the $\bar{z}$-polytope is non-empty, we construct a point $\bar{\psi} \in \mathbb{R}^{I \times (I \cup \mathcal{C})}_{\geq 0}$ which satisfies most of the constraints of the $\bar{u}$-polytope. Then, based on $\bar{\psi}$, we construct another point $\bar{\lambda} \in \mathbb{R}^{I \times (I \cup \mathcal{C})}_{\geq 0}$ and show that $\bar{\lambda}$ is in the $\bar{z}$-polytope. For all $C \in \cC$, let $P^{-1}(C) = \{C' \in \cC~|~ P(C') = C\}$ be all configurations mapped by the projection to $C$. The construction of $\bar{\psi}$ relies on a point $\bar{\gamma} \in \mathbb{R}^{I \times (I \cup \mathcal{C})}_{\geq 0}$ in the $\bar{y}$-polytope, where for all $\ell, j \in I$ such that $\ell \neq j$ it holds that $\bar{\gamma}_{\ell,j} = 0$.

 The point $\bar{\psi}$ is defined as follows. For all $\ell \in I$ and $C \in \cC$ define $\bar{\psi}_{\ell,C} = \sum_{C' \in P^{-1}(C)} \bar{\gamma}_{\ell,C'}$. That is, $\ell$ is assigned to $C$ as the fractional assignment of $\ell$ by $\bar{\gamma}$ to configurations projected to $C$.  Additionally, for any $G \in \cG_{\eta} \cup \cA$, $i \in [\tau(G)-1]$, $\ell \in \Delta_{i+1}(G)$, and $j \in \Delta_{i}(G)$ define $$\bar{\psi}_{\ell,j} =   \frac{f_{\bar{u}}(j)}{f_{\bar{u}}(\Delta_{i}(G))} \cdot \bar{\gamma}_{\ell,\ell}.$$ This is where the {\em shifting} comes into play: all items from the important group $G$ in class $\Delta_{i+1}(G)$ are assigned to at most $\eps^{-10}$ slots in the consecutive class. By \eqref{eq:mapS,eq:mapM}, Observation~\ref{ob:FG}, and Step~\ref{step:u} of Algorithm~\ref{Alg:shifting}, it follows that $f_{\bar{u}} \left(\Delta_{i}(G)\right) \approx f_{\bar{y}} \left(\Delta_{i+1}(G)\right)$ and $\Delta_{i+1}(G) \subseteq \textsf{fit}(\min \Delta_{i}(G))$. Hence, most of the constraints in the $\bar{u}$-polytope of items in $\Delta_{i+1}(G)$ are satisfied. For any other entry $(\ell',t') \in I \times (I \cup \cC)$ not defined above, define $\bar{\psi}_{\ell',t'} = 0$. 
 
	\begin{claim}
	\label{ob:psi}
	$\bar{\psi}$ satisfies constraints \eqref{F1}, \eqref{F2}, \eqref{F5}, \eqref{F3} of the $\bar{u}$-polytope. 
\end{claim} 

Constraint \eqref{F1} of the $\bar{u}$-polytope is satisfied for $\bar{\psi}$ by Observation~\ref{ob:FG} and Lemma~\ref{lem:P(C)}. Constraint \eqref{F2} is satisfied by Lemma~\ref{lem:P(C)}; and constraint \eqref{F3} follows from Observation~\ref{ob:FG}. As $\bar{\psi}$ may fail to satisfy constraint \eqref{F4} of the $\bar{u}$-polytope, we define $\bar{\mu} \in \mathbb{R}^{I}_{\geq 0}$ such that for all $\ell \in I$ $\bar{\mu}_{\ell} = \max \left\{0,1-\sum_{t \in I \cup \cC} \bar{\psi}_{\ell,t}\right\}$ is the fractional assignment that is missing for item $\ell$ to satisfy constraint \eqref{F4} for $\psi$.

	\begin{claim}
	\label{lem:mu}
	The following hold for $\bar{\mu}$.
	
	\begin{enumerate}
		\item $\textnormal{\textsf{supp}}(\bar{\mu}) \setminus L \subseteq \textnormal{\textsf{fit}}(\emptyset)$.
		
			\item For all  $\ell \in \textnormal{\textsf{supp}}(\bar{\mu})$ it holds that $\ell \in \textnormal{\textsf{fit}}\left(r(\textnormal{\textsf{group}}(\ell))\right)$.
		
		\item 	For all $G \in \cG$ it holds that $\sum_{\ell \in G} \bar{\mu}_{\ell} \leq \eps \cdot \|\bar{y}\|+2$.
		
		\item 	$\sum_{\ell \in I\setminus L} \bar{\mu}_{\ell}  \cdot s(\ell) \leq \eps\cdot\|\bar{y}\|$.
	\end{enumerate}
	
\end{claim} 

The proof of Properties 3 and 4 in Claim~\ref{lem:mu} follows from Lemma~\ref{lem:significantDiscard} and Observation~\ref{ob:FG}, since $\textsf{supp}(\bar{\mu})$ contains only (i) small items from non-important groups, and (ii) items from the first class of each important group.

Now, we define another point $\bar{\rho} \in \mathbb{R}^{I \times (I \cup \mathcal{C})}_{\geq 0}$ based on $\bar{\mu}$. For all $\ell \in I \setminus L$ define $\bar{\rho}_{\ell,\emptyset} = \bar{\mu}_{\ell}$. Also, for all $G \in \cG$, $\ell \in L \cap G$ define $\bar{\rho}_{\ell, r(G)} = \bar{\mu}_{\ell}$; for any other $(\ell,t) \in I \times \left(I \cup \cC\right)$, define $\bar{\rho}_{\ell,t} = 0$. Finally, using the definitions of $\bar{\rho}$ and $\bar{\psi}$, let $\bar{\lambda} = \bar{\rho}+\bar{\psi}$. By the above, all items that are not fully assigned by $\bar{\psi}$ are assigned to the empty-configuration or to the configuration containing a single item that is largest in the group. The constraints of the $\bar{z}$-polytope are satisfied, since the empty-configuration and the maximal slots of important groups have increased selection.  
The next claim follows from Claims~\ref{ob:psi} and~\ref{lem:mu}. 

	\begin{claim}
	\label{lem:shifting1}
	$\bar{\lambda}$ is in the $\bar{z}$-polytope. 
\end{claim}

\section{Algorithm $\textsf{Partition}$}
\label{sec:FromPolytope}

Algorithm \textsf{Partition} constructs an $\eps$-nice partition of small size based on the integrality properties of the polytope, as given in Lemma~\ref{O(1)}. Let $\bar{z}$ be 
 a good prototype (satisfying the conditions in Definition~\ref{def:goodprototype}).
Since the values of entries in the support of $\bar{z}$ are not necessarily integral, $\bar{z}$ may not satisfy the conditions of Lemma~\ref{O(1)}. Thus, we 
modify $\bar{z}$ to a prototype having integral entries.  Define the {\em integralization} of $\bar{z}$ as a prototype $\bar{z}^* \in \mathbb{R}^{\cC}_{\geq 0}$ such that, for all $C \in \cC$, $\bar{z}^*_C = \ceil{\bar{z}_C}$.

\begin{observation}
	\label{ob:y}
	The $\bar{z}^*$-polytope is non-empty, $\textnormal{\textsf{supp}} (\bar{z}^*) =\textnormal{\textsf{supp}} (\bar{z})$, and $\|\bar{z}^*\| \leq \|\bar{z}\|+|\textnormal{\textsf{supp}}(\bar{z})|$. 
\end{observation}

By Observation~\ref{ob:y}, it follows that $\bar{z}^*$ is a good prototype, and there is a vertex $\bar{\gamma}$ of the $\bar{z}^*$-polytope that satisfies constraint \eqref{F4} with equality. Let $F = \{\{\ell\}~|~ \ell \in I, \exists t \in I \cup \cC \text{ s.t. } \bar{\gamma}_{\ell,t} \in (0,1)\}$ be the set of  
items that are fractionally assigned to some type by $\bar{\gamma}$. Since $\bar{z}^*$ is a good prototype with integral entries, by Lemma~\ref{O(1)} $|F|$ is small. Thus, 
an $\eps$-nice partition is obtained by assigning the integral items in $\bar{\gamma}$ and then adding the {\em fractional} items,
with only a small increase in the size of the $\eps$-nice partition. 

We start by finding a packing for items assigned integrally by $\bar{\gamma}$ to slot-types via bipartite matching. Specifically, let $$V = \left\{\left(C,j,k\right) ~|~C\in \textsf{supp}(\bar{z}^*), j \in C, k \in \left[\bar{z}^*_C\right] \right\}$$ where each vertex in $V$ represents a slot within a configuration $C \in \textsf{supp}(\bar{z}^*)$, and an index for one of the $\bar{z}^*_C$ bins associated with $C$. Also, let $U = \{\ell \in I~|~ \exists j \in I \text{ s.t. } \bar{\gamma}_{\ell,j} = 1\}$ be all items assigned integrally to some slot-type by $\bar{\gamma}$. Now, define the {\em assignment graph} of $\bar{\gamma}$ as the bipartite graph $G = (U, V, E)$, where $E = \{\left(\ell, \left(C,j,k\right)\right) \in U \times V~|~ \bar{\gamma}_{\ell,j} = 1\}$. That is, there is an edge between any item $\ell \in U$ to $(C,j,k) \in V$ if $\ell$ is assigned integrally (i.e., completely) to $j$ by $\bar{\gamma}$. If the edge $(\ell,(C,j,k))$ is taken to the matching, we replace the slot $j$ by $\ell$ in the $k$-th bin associated with $C$. 

\begin{lemma}
	\label{lem:matchingU}
There is a matching in $G$ of size $|U|$. 
\end{lemma}
 The proof of Lemma~\ref{lem:matchingU} is based on finding a fractional matching in $G$, by taking $\frac{1}{d(\ell)}$ of an edge $(\ell,v) \in E$ to the matching, where $d(\ell)$ is the degree of $\ell$ in $G$. By constraint \eqref{F3} of the $\bar{z}^*$-polytope, this guarantees a feasible fractional matching of size $|U|$. Since $G$ is bipartite, there is also an integral matching of size $|U|$ in $G$ \cite{schrijver2003combinatorial}. Let $M$ be a matching in $G$ of size $|U|$.

We construct below a packing based on $M$. For all $C \in \textsf{supp}(\bar{z}^*)$ and $k \in [\bar{z}^*_C]$, define $$A_{C,k} = \{\ell \in U ~|~ \exists j \in I \text{ s.t. }  (\ell,(C,j,k)) \in M\}$$  as the {\em bin} of $C$ and $k$, which contains all items  coupled by the matching to the $k$-th bin associated with $C$.  The next lemma follows since $M$ is a matching in the assignment graph of $\bar{\gamma}$.

\begin{lemma}
	\label{lem:allowed}
For all $C \in \textnormal{\textsf{supp}} (\bar{z}^*)$ and $k \in [\bar{z}^*_C]$, $A_{C,k}$ is allowed in $C$. 
\end{lemma} 

We now define a packing of all items.
Given two tuples $S,T$, let $S \oplus T$ denote  the concatenation of $S$ and $T$.  The {\em packing of} $\bar{\gamma}$ and $M$ is given by

\begin{equation}
\label{Astar}
A^* = \big(h~|~h \in F\big) \oplus \big(A_{C,k} ~|~ C \in \textsf{supp}(\bar{z}^*), k \in [\bar{z}^*_C]\big). 
\end{equation} In words, $A^*$ contains a bin for every fractional item, as well as all the bins associated with configurations in the support of $\bar{z}^*$. It follows that $A^*$ is a packing of $U$ and all fractional items. 

We construct an $\eps$-nice partition below whose packing is $A^*$. Let $\mathcal{H} = \textsf{supp}(\bar{z}^*) \cup F$. By Lemma~\ref{O(1)}, we get $|F| \leq \eps^{-21}|\textsf{supp}(\bar{z}^*)|^2$, and it follows that $|\mathcal{H}| \leq \eps^{-22}Q^2(\eps)$ . For all $C \in \mathcal{H}$ define the {\em category of} $C$ as \begin{equation}
	\label{BC}
	B_C = \left\{A_{C,k}~|~ k \in [\bar{z}^*_{C}] \right\}  \cup \big\{ h \in F ~|~ h = C \big\}
\end{equation} that contains all the bins associated with $C$, and possibly a singleton of a fractional item.\footnote{If $C$ is a singleton of this item.} By \eqref{Astar} and \eqref{BC}, $\{B_C\}_{C \in \mathcal{H}}$ is a partition of the bins in $A^*$. In addition, for all $C \in \mathcal{H}$ define  $$D_{C} = \{\ell \in I~|~\bar{\gamma}_{\ell,C} = 1\}$$ as all items assigned integrally to $C$ by $\bar{\gamma}$. Finally, define the {\em division} of $\bar{\gamma}$ and $M$, as $A^*$, $\mathcal{H}$, $(B_C)_{C \in \mathcal{H}}$, and $(D_C)_{C \in \mathcal{H}}$. Algorithm $\textsf{Partition}$ computes the above $\eps$-nice partition. We give the pseudocode in Algorithm~\ref{Alg:Assign}.

\begin{algorithm}[htb]
	\caption{$\textsf{Partition} (\bar{z}$)}
	\label{Alg:Assign}
	
	Generate $\bar{z}^*$, the integralization of $\bar{z}$
	
	Find a vertex $\bar{\gamma}$ of the $\bar{z}^*$-polytope satisfying \eqref{F4} with equality \label{assign:find_gamma}
	
	Construct the assignment graph $G$ of $\bar{\gamma}$
	
	Find a maximum matching $M$ in $G$
	
	Return the division of $\bar{\gamma}$ and $M$
	
\end{algorithm}

  \begin{claim}
	\label{eqeq6}
	For any $C \in \mathcal{H}$ and $G \in \cG$ the following hold.
	\begin{enumerate}
		\item $D_C \subseteq \textnormal{\textsf{fit}}(C)$. 
		\item $s(D_C) \leq \left(1-s(C)\right) \cdot |B_C|$. 
		\item  $|D_C \cap G| \leq |B_C| \cdot \left(     k(G) - |G \cap C|  \right)$.
	\end{enumerate} 
\end{claim} 

 The properties of the $\bar{z}^*$-polytope are the key to proving Claim~\ref{eqeq6} and consequently Lemma~\ref{FromPolytope}. Recall that $\bar{\gamma}$ is in the $\bar{z}^*$-polytope; thus, the first property of Claim~\ref{eqeq6} follows from constraint \eqref{F1}, the second by constraint \eqref{F2}, and the third by constraint \eqref{F5}. Moreover, using constraint \eqref{F3}, we construct the matching $M$. Finally, we use constraint \eqref{F4} to show that $\{D_C\}_{C \in \mathcal{H}}$ is a partition of items not packed in $A^*$.
\begin{lemma}
	\label{lem:assign}
	The division of $\bar{\gamma}$ and $M$ is an $\eps$-nice partition of size at most $\|\bar{z}\|+\eps^{-22}|\textsf{supp}(\bar{z})|^2$.
\end{lemma}

\noindent{\bf Proof of Lemma~\ref{FromPolytope}}
The correctness of Algorithm~\ref{Alg:Assign} follows from~\ref{lem:assign}. Observe that the vertex~$\bar{\gamma}$ in Step~\ref{assign:find_gamma} can be found via standard linear programing, thus the algorithm runs in polynomial time. 
\qed

\section{Algorithm \textsf{Pack}}
\label{sec:alg_pack}

In this section we give the proof of Lemma~\ref{lem:GREEDY}. 
Let $A$, $\mathcal{H}$, $\left(B_C\right)_{C \in \mathcal{H}}$, and $\left(D_C\right)_{C \in \mathcal{H}}$ be an $\eps$-nice partition of $\mathcal{I}$ of size $m$. For all $C \in \mathcal{H}$, algorithm \textsf{Pack} uses algorithm \textsf{Greedy} (see Lemma~\ref{thm:greedy}) to add the items in $D_C$ to the bins in $B_C$, possibly by using a few extra bins.  This requires a definition of a residual instance with further restrictions so the bins of $B_C$ remain feasible as a packing of $\cI$. Specifically, given $C \in \mathcal{H}$, define $\cI_C = (D_C, \cG_C,s_C,k_C)$ as the {\em residual instance} of $C$ and $\cI$, such that $\mathcal{G}_C = \{G \cap D_C \neq \emptyset ~|~ G \in \cG\}$, and \begin{equation}
	\label{IC}
	\begin{aligned}
		&s_C(\ell) = \frac{s(\ell)}{1-s(C)}&~~~~~~~~~~~&\forall \ell \in D_C\\ 
		&k_C(G \cap D_C) = k(G)-|G \cap C| &&\forall G \in \cG \text{ s.t. } G \cap D_C \neq \emptyset~~
	\end{aligned}
\end{equation} In words, we modify the size function such that a bin $B \in B_C$ may 
receive 
an additional subset of items $S \subseteq D_C$ with the following attributes: the total size of items in $S$ is at most $1-s(C)$, and for each $G \in \cG$ it holds that $|S \cap B| \leq k(G)-|C \cap G|$. Thus, $B \cup S$ is a configuration of $\cI$. Also, observe that if $D_C \neq \emptyset$ then by Definition~\ref{def:partition} $s(C) < 1$; thus, $\cI_C$ is well defined.

\begin{algorithm}[h]
	\caption{$\textsf{Pack}\left(\eps, \cI, A, \mathcal{H}, \left(B_C\right)_{C \in \mathcal{H}}, \left(D_C\right)_{C \in \mathcal{H}}\right)$}
	\label{alg:pack}
	
Initialize $P \leftarrow ()$.

\For{$C \in \mathcal{H}$}{

Compute the residual instance $\cI_C$ of $C$ and $\cI$.

$P_C \leftarrow \textsf{tuple}(B_C) + \textsf{Greedy}(\cI_C,\eps)$.

 $P \leftarrow P \oplus P_C$.\label{step:pPack}

}

Return $P$.

\end{algorithm}

We use the following notation for the algorithm. Given two tuples $T = (T_1, \ldots, T_t)$ and $R = (R_1, \ldots, R_r)$ where $T_i, R_j \subseteq I ~\forall i \in [t], j \in [r]$, recall that $T \oplus R$ is the concatenation of $T,R$. In addition,  if $r> t$ let $T+R = (T_1 \cup R_1, \ldots, T_t \cup R_t, R_{t+1}, \ldots, R_r)$ and if $r \leq t$ let $T+R = (T_1 \cup R_1, \ldots, T_r \cup R_r, T_{r+1}, \ldots, T_t)$; in words, this is an element-wise addition of the items in $R$ to $T$. Finally, given $C \in \mathcal{H}$, let $\textsf{tuple}(B_C)$ be the elements (bins) in $B_C$ ordered arbitrarily as a tuple.  For all $C \in \mathcal{H}$ algorithm \textsf{Pack} calls algorithm \textsf{Greedy} with parameter $\eps$ for the residual instance $\cI_C$ and creates the packing $\textsf{tuple}(B_C) + \textsf{Greedy}(\cI_C,\eps)$; then, algorithm \textsf{Pack} concatenates all of these packings into a packing of $\cI$. We give the 
pseudocode of \textsf{Pack} in Algorithm~\ref{alg:pack}. 
The proof of the next lemma follows from~\eqref{IC}, Definition~\ref{def:partition}, and Lemma~\ref{thm:greedy}. We give the full details and the proof of Lemma~\ref{lem:GREEDY} in Appendix~\ref{sec:PackProofs}. 

\begin{lemma}
	\label{clm:GREEDY}
	For all $C \in \mathcal{H}$, $\textnormal{\textsf{tuple}}(B_C) + \textnormal{\textsf{Greedy}}(\cI_C,\eps)$ is a packing of $D_C \cup \bigcup_{B \in B_C} B$ with respect to $\cI$ in at most $(1+2\eps) \cdot |B_C|+2$ bins. 
\end{lemma} 

\section{Reduction and Reconstruction of the Instance}
\label{sec:reduction}

In this section we present algorithms \textsf{Reduce} and \textsf{Reconstruct}, which handle the structuring of the instance and 
the transformation of the solution for the structured instance into a solution for the original one. For this section, fix $\cj = (I, \mathcal{G}, s,k)$ to be a BPP instance, and let $\eps \in (0, 0.1]$.

Algorithm \textsf{Reduce} generates from a given instance $\cj$ a structured instance $\II$, for which any optimal packing required at most the minimum number of bins required for packing $\cj$
(see Section~\ref{sec:F1}). Given a packing of $\II$, algorithm \textsf{Reconstruct} finds a packing  of the original instance $\cj$, with only a slight increase in the total number of bins used (see Section~\ref{sec:RECONSTRUCTION}). The proofs of the results in this section are given in Appendix~\ref{sec:reductionProofs}. 

\subsection{The Reduction}
\label{sec:F1}

Towards structuring the instance, we first classify the items and groups in $\cj$.
We note that similar classifications were used in prior work (see, e.g., \cite{DW17,Jansen_et_al:2019,DKS21}).  
 For all $i \in  \{2,\ldots, \eps^{-1}+1\}$ let 
 \begin{equation}
	\label{Interval}
	I_i = \big\{\ell \in I~|~ s(\ell) \in \big[\eps^{i+1}, \eps^{i}\big)\big\}
\end{equation} be the $i$-th {\em interval} of $\cj$, containing all items of sizes in the interval $[\eps^{i+1}, \eps^{i})$; we refer to $i$ is a {\em pivot} of $\cj$. Now, for $w \in \{2,\ldots,\eps^{-1}+1\}$, we say that $w$ is a {\em minimal pivot} of $\mathcal{\cj}$ if the $w$-th interval of $\cj$ contains minimal total size of items among all intervals of $\cj$; that is,

\begin{equation}
	\label{eq:k}
	w = \argmin_{i \in \{2, \ldots, \eps^{-1}+1\}} s(I_i).
\end{equation}

The classification of the items depends on the selection of a pivot for $\cj$. Specifically, given $w \in \{2,\ldots, \eps^{-1}+1\}$ we classify item $\ell$ as $w$-{\em  heavy} if $s(\ell) \geq \varepsilon^{w}$, $w$-{\em medium} if $s(\ell) \in [\varepsilon^{w+1},\varepsilon^{w})$ and $w$-{\em light} otherwise. Let $H_w = \{\ell \in I~|~s(\ell) \geq \eps^{w}\}$ be the set of $w$-heavy items in $\cj$. 

For the classification of groups, fix a pivot $w \in \{2, \ldots, \eps^{-1}+1\}$ of $\cj$. We sort the groups in $\mathcal{G}$ in a non-increasing order by the total number of $w$-heavy and $w$-medium items of the group; then, we classify the first up to $ \eps^{-3w-5}$ groups according to the above order as the $w$-{\em large groups}, and the remaining groups are called $w$-{\em small}. More formally, let $g_w(G) = |G \cap (H_w \cup I_w)|~:~\forall G \in \cG$ and $n = |\cG|$. Now, let $G_1, \ldots, G_{n}$ be a 
non-increasing order of $\cG$ by $g_w$, where
groups having the same $g_w$ values are placed in fixed arbitrary order.
Now, let $\kappa_w = \min\left\{\eps^{-3w-5},|\cG| \right\}$ and define the set of $w$-large groups as $\cG_L(w) = \left\{G_i~|~i \in  \left[\kappa_w \right] \right\}$ and the set of $w$-small groups as $\cG \setminus \cG_L$.

\begin{lemma}
	\label{lem:few_large_groups}	
	For all $w \in \{2,\ldots, \eps^{-1}+1\}$ there are at most $\eps^{-3w-5}$ groups $G \in \cG$ such that $$|G \cap (H_w\cup I_w) | \geq \eps^{2w+4} \cdot \OPT(\cj).$$ 
\end{lemma}

The proof of Lemma~\ref{lem:few_large_groups} follows by noting that
the total size of items in a group having at least $\eps^{2w+4} \cdot \OPT(\cj)$ items that are $w$-medium  or $w$-heavy is at least $\eps^{3w+5} \cdot \OPT(\cj)$. 

By Lemma~\ref{lem:few_large_groups}, the number of $w$-heavy and $w$-medium items in a $w$-small group is at most $\eps^{2w+4}\cdot \OPT(\mathcal{J})$. This is useful for the reduction, in which the matroid constraint is slightly relaxed for $w$-small groups, as described below. The $w$-large groups appear in the reduced instance with the original cardinality constraint. In contrast, for each $w$-small group we keep only the $w$-light and $w$-medium items with the original cardinality constraint of the group, whereas all $w$-heavy items from $w$-small groups are placed in a single {\em union group} with unbounded cardinality constraint. Specifically, define the {\em reduced $w$-small groups} containing only $w$-light and $w$-medium items as

\begin{equation}
\label{Is}
\cG_S(w) = \{G \setminus H_w~|~G \in \cG \setminus \cG_L(w)\}.
\end{equation}

Also, let the {\em $w$-union group} containing all $w$-heavy items from $w$-small groups be 

\begin{equation}
\label{Iu}
\Gamma_w = \{\ell \in H_w~|~\ell \in G \text{ s.t. } G \in \cG \setminus \cG_L(w)\}.
\end{equation}

\begin{algorithm}[h]
	\caption{$\textsf{Reduce}(\eps,\cj = (I,\cG,s,k))$}
	\label{alg:reduction}
	
	Find a minimal pivot $w \in \{2,\ldots,\eps^{-1}+1\}$ of $\cj$.\label{step:pivot}
	
	Find $\cG_L(w), \cG_S(w), \Gamma_w$, the $w$-large groups, $w$-reduced small groups, and $w$-union of $\cj$.  
	
	Let $\cG_R = \cG_L(w) \cup \cG_S(w) \cup \{\Gamma_w\}$.\label{step:groups}
	
	Define  the following new cardinality bounds for all $G \in \cG_R$:\label{kr} \[
	k_R(G) =  \label{kr} 
	\begin{cases}
		k(G) & G \in \cG_L(w)\\
		k(G') & G = G' \setminus H_w \text{ s.t. } G' \in \cG \setminus \cG_L(w) \\
		|\Gamma_w|+1 & G = \Gamma_w.\\
	\end{cases}
	\]

	Return the reduced instance $\left(I,\cG_R,s,k_R\right)$. \label{Ir}
	
\end{algorithm}

Algorithm \textsf{Reduce}  constructs the above groups, together with the new cardinality constraints, and preserves the initial set of items and item sizes. The pseudocode of Algorithm \textsf{Reduce} is given in Algorithm~\ref{alg:reduction}. The first claim of Lemma~\ref{lem:reductionReconstruction} holds since the only of groups containing items larger than $\eps^2$ are the $w$-large groups and the $w$-union group. Moreover, a packing of the original instance is also a feasible packing for the reduced instance, thus the optimum of the reduced instance can only be smaller.

\subsection{The Reconstruction}
\label{sec:RECONSTRUCTION}

We now describe how a packing of the reduced instance is transformed to a packing for the original instance.  
For simplicity, assume that objects such as $w$ computed by algorithm $\textsf{Reduce}(\eps,\cj)$ are known and fixed for the reconstruction of $\cj$. In addition, for the remainder of this section, fix a packing $A = (A_1, \ldots, A_m)$ for $\cI = \textsf{Reduce}(\eps,\cj)$. 

Recall that Step~\ref{kr} of algorithm \textsf{Reduce} relaxes the cardinality constraint for the $w$-heavy items in the set $\Gamma_w$. The reconstruction algorithm redefines $A$ for items in $\Gamma_w$ to ensure that the matroid constraint is satisfied for these items; this may require using a few extra bins. Then, $w$-medium and $w$-light items from $w$-small groups are discarded from the packing and packed separately using Algorithm \textsf{Greedy} (see Lemma~\ref{thm:greedy}). This results in a feasible packing of the original instance $\cj$. 
We now describe how our reconstruction algorithm resolves violations among $w$-heavy items. This is done by rearranging the packing of $w$-heavy items in $A$ using a variant of the classic {\em linear shifting} technique of~\cite{fernandez1981bin}. 
Let $\beta = \eps^{-w-2}$ be the {\em shifting parameter}, and $P_1,\ldots, P_q$ the $\mathcal{\cI}$-{\em partition} obtained by applying linear shifting to the items in $\Gamma_w$. More specifically, given the set $\Gamma_w$ in non-decreasing order of item sizes, for all $1\leq i \leq j<q$ and $\ell \in P_i, y \in P_j$ it holds that $s(\ell) \geq s(y)$, $|P_i| = \ceil{\frac{|\Gamma_w|}{\beta}}$, and $P_{q}$ contains the remaining (possibly less than  $\ceil{\frac{|\Gamma_w|}{\beta}}$) items from $\Gamma_w$. It follows that $q \leq \beta$ and let $q(\cI) = q$ be the number of sets in the above partition of $\Gamma_w$. 

\begin{lemma}
\label{lem:Bgreedy}
There is an algorithm \textnormal{\textsf{Fill}} which given a BPP instance $\cj$ and a  packing $A = (A_1, \ldots, A_m)$ of $\cI = \textnormal{\textsf{Reduce}}(\eps,\cj)$ finds in time $\textnormal{poly}(|\cj|, \frac{1}{\eps})$ a partition $(B_1, \ldots, B_m,R)$ of $\Gamma_w$ such that the following hold. 

\begin{enumerate}
\item For all $i \in [m]$ and $j \in \{2, \ldots,q(\cI)\}$ it holds that $|B_i \cap P_j| \leq |A_i \cap P_{j-1}|$ and $|B_i \cap P_1| = 0$.

\item For all $i \in [m]$ and $G \in \cG$ it holds that 
$|B_i \cap G| \leq k(G)$.

\item $|R| \leq \eps \cdot \OPT(\cj)+1$. 
\end{enumerate}   
\end{lemma}

The first condition of the lemma is essentially a shifting argument. The lemma is proved by constructing the following greedy algorithm \textsf{Fill}. We start with empty sets for $B_1, \ldots, B_m$ and with $R = \Gamma_w$. In each iteration, we try to move an item from $R$ to some $B_i, i \in [m]$ without violating the conditions of Lemma~\ref{lem:Bgreedy}.
If no move is possible, we return $(B_1, \ldots, B_m,R)$. The details are given in Appendix~\ref{sec:fill}.

We now combine the output of algorithm $\textsf{Fill}$ with $A$. We first remove from $A$ the items in $\Gamma_w$ and all $w$-medium items from $w$-small groups; then, for all $i \in [m]$ we add to $A_i$ the items in $B_i$. The remaining items in $R$ are packed using extra bins. Formally, let \begin{equation}
	\label{X:X} \Omega_w = \{\ell \in I_w~|~\ell \in G \text{ s.t. } G \in \cG \setminus \cG_L(w)\}
\end{equation} be the $w$-medium items from $w$-small groups. In addition, given two tuples $T_1, T_2$, let $T_1 \oplus T_2$ be the concatenation of $T_1$ by $T_2$. Now, given $(B_1, \ldots, B_m,R) = \textsf{Fill}(A)$, define the tuple 

\begin{equation}
\label{F_A}
F_A = \big((A_i \setminus (\Gamma_w\cup \Omega_w)) \cup B_i~|~ i \in [m]\big) \oplus \big(\{\ell\}~|~\ell \in R\big).	
\end{equation}

In the following, we transform $F_A$ into a feasible packing of $\cj$. Let $F_A = (U_1, \ldots, U_{m'})$. Observe that Lemma~\ref{lem:Bgreedy} does not guarantee that $F_A$
satisfies the cardinality constraints for all the $w$-small groups.
Such violations are resolved by discarding $w$-light items from $U_i, i \in [m']$ to ensure  that the cardinality constraints of all groups in $\cG$ are satisfied.  More specifically, for any bin $i \in [m']$ and group $G \in \cG$ let $F_A(i,G)$ be an arbitrary exclusion-minimal subset of $w$-light items from $U_i$ such that $| (U_i \setminus F_A(i,G) )  \cap G| \leq k(G)$. Note that such a subset exists, since by taking all $w$-light items to $F_A(i,G)$, the above condition is satisfied by Lemma~\ref{lem:Bgreedy} and \eqref{F_A}. Now, the discarded items are accumulated across all bins and groups to form the set $D(A)$, namely

\begin{equation}
\label{D}
D(A) = \bigcup_{i \in [m']} \bigcup_{G \in \cG} F_A(i,G).
\end{equation}

\begin{lemma}
\label{thm:D}

For any packing $A$ of $\cI$,

\begin{enumerate}
\item  $|D(A) \cap G| \leq \eps \cdot \OPT(\cj)$ for all $G \in \cG$. 

\item $s(D(A)) \leq \eps \cdot \OPT(\cj)$.

\end{enumerate}

\end{lemma}

The first condition in the lemma follows from Lemma~\ref{lem:few_large_groups}, since we do not discard items from $w$-large groups. The second condition holds since any $w$-heavy item is larger than any $w$-light item by factor at least $\eps^{-1}$.  These properties are useful for packing the discarded items in only a few extra bins using algorithm \textsf{Greedy}.

Given the partial packing of $\cj$, algorithm \textsf{Reconstruct} proceeds to pack the remaining items, i.e., the $w$-medium items from $w$-small groups and items in $D(A)$. All  of these items are packed in a few extra bins.
This is done by algorithm \textsf{Greedy}, for which we define a residual instance containing only the above remaining items, which preserves the item sizes and group cardinality constraints for the remaining items. Formally, define the {\em discarded instance} of $\cj$ and $A$ as $\mathcal{E} = (E,\cG_E,s_E,k_E)$, where  

\begin{equation}
\label{Ie}
\begin{aligned}
E = \Omega_w \cup D(A), ~~~~~~~~~~~~~~~~~~~~ ~~~~~~
\mathcal{G}_E = \{G \cap E~|~ G \in \cG\}, ~~~~~~~~~~~~~~~~~~ 
\\
s_E(\ell) = s(\ell), ~~\forall \ell \in E, ~~~~~~~~~~~~~~~~~~~
k_E(G \cap E) = k(G), ~~\forall G \in \cG~~~~~~~~~~~~
\end{aligned}
\end{equation}

Given $S \subseteq I$ and a tuple $T =  (T_1, \ldots, T_n)$ such that $T_i \subseteq I$ for all $i \in [n]$, let $T \setminus S = (T_1 \setminus S, \ldots, T_n \setminus S)$  be the tuple induced by removing the items in $S$ from all entries of $T$. The pseudocode of algorithm \textsf{Reconstruct} is given in Algorithm~\ref{Alg:RECON}. By \eqref{F_A} and \eqref{Ie}, the output of algorithm \textsf{Reconstruct} is a packing of $\cj$. Furthermore, this packing is of size at most $m+13\eps \cdot \OPT(\cj)+1$ by Lemmas~\ref{lem:Bgreedy}, \ref{thm:D}, and \ref{lem:few_large_groups}.

\begin{algorithm}[h]
\caption{$\textsf{Reconstruct}(\mathcal{J}, \eps, A)$}
\label{Alg:RECON}

Generate $F_A$ using $\textsf{Fill}(\mathcal{J}, A)$ and \eqref{F_A}.\label{step:1}

Compute the discarded instance $\mathcal{E} = (E,\cG_E,s_E,k_E)$ of $\cj$ and $A$.\label{step:2} 

Return $\left( F_A \setminus E \right) \oplus \textsf{Greedy}(\mathcal{E})$.\label{step:reconl}\label{step:3}

\end{algorithm}

\section{Discussion}
\label{sec:discussion}

In this paper we presented 
an $o(\OPT)$ additive approximation algorithm for Bin Packing with Partition Matroid. While BPP is a natural generalization of Bin Packing variants that have been
studied in the past, to the best of our knowledge it is studied here for the first time.
Our result improves upon the APTAS of \cite{DKS21} for the special case of Group Bin Packing and  generalizes the AFPTAS of \cite{epstein2010afptas} for the special case of Bin Packing with Cardinality Constraints. Our algorithm is based on rounding a solution for the configuration-LP formulation of the problem. The rounding process relies on the key notion of a {\em prototype}, in which items are placeholders for other items, and the use of fractional grouping~\cite{FKS21}. Our algorithm demonstrates the power of this fractional version of linear grouping in solving constrained packing problems; it also 
shows how fractional grouping can be used {\em constructively}.

While our algorithm outputs a solution which uses $\OPT+o(\OPT)$ bins,  the function hidden by the little-o notation is of the form $\frac{x}{(\ln \ln x)^{\frac{1}{17}}}$.  We believe that a tighter analysis may lead to a better additive approximation, 
for example, to an algorithm which returns a solution using $\OPT+O\left(\frac{\OPT}{\ln \OPT} \right)$ bins. We leave the tighter analysis for the full version of this paper. The existence of approximation algorithms for BPP which return a solution using at most $\OPT+O(\OPT^{1-\eps})$ bins, for some constant $\eps>0$, remains open.

The techniques presented in this paper seem to be useful also in other settings.  
Our preliminary study suggests we can apply these techniques 
to obtain a polynomial time approximation scheme for {\em Multiple Knapsack with Partition Matroid}, a generalization of the Multiple Knapsack problem (see, e.g., \cite{CK05, Ja10}) in which the items assigned to each bin form an independent set of a partition matroid. Another  application comes from the design of approximation algorithms for 
Machine Scheduling with Partition Matroid, 
a generalization of the classic Machine Scheduling problem in which the jobs assigned to a machine must be an independent set of a given partition matroid. We note that the problem is a generalization of Machine Scheduling with Bag-Constraints studied in~\cite{DW17,Jansen_et_al:2019}.

The problem of Bin Packing with Partition Matroid is a special case of Bin Packing with Matroid, for which the input is a set of items $I$, a size function $s:I\to \mathbb{R}$ and a matroid~$\cM$. The objective is to partition $I$ into a minimal number of bins $A_1,\ldots, A_m$ such that $A_b$ is an independent set of the matroid $\cM$, and $s(A_b)\leq 1$ for all $b \in [m]$. This problem is a natural generalization of both  Bin Packing and Matroid Partitioning; yet, we were unable to find any published results.
We note that our approach for solving BPP  heavily relies on the structure of the partition matroid, and therefore cannot be easily extended to handle a general matroid.

\appendix

\section{Omitted Proofs of Section~\ref{sec:overview}}
\label{app:omitted}

\newcommand{\LP}{\textnormal{LP}}
\newcommand{\cD}{\mathcal{D}}
\newcommand{\dual}{\textnormal{Dual}}
\newcommand{\ellip}{\textnormal{\textsf{Ellipsoid}}}
\noindent{\bf Proof of Lemma~\ref{configurationLP}:}
The approach presented here is considered as standard, and often the result is mentioned without a proof (see, e.g., Theorem~1.1 in \cite{bansal2010new}). We include the proof 
for completeness. We refer to terms such as separation oracle and  well-described polyhedron as defined in~\cite{grotschel2012geometric}. 

Let $\II=(I,\cG, s, k )$ be a BPP instance and $\eps\in (0,0.1)$. 
For any $\cD\subseteq \cC$ we define the following linear program.
\begin{equation}
	\label{C-LP-relaxed}
	\LP(\cD):~~~~~
	\begin{aligned}
		~~~~~ \min\quad        & ~~~~~\sum_{C \in \cD} \bar{x}_C                                                           \\
		\textsf{s.t.\quad} & ~~~\sum_{~C \in \mathcal{C}[\ell]\cap \cD} \bar{x}_C \geq 1   & \forall \ell \in I~~~~~\\  
		& ~~~~~\bar{x}_C \geq 0 ~~~~~~~~&~~~~~~~~~~~~~~~~~~~~ \forall C \in  \mathcal{C}~~~~
	\end{aligned}
\end{equation}
While $\LP(\cC)$ is not identical to \eqref{C-LP}, solving $\LP(\cC)$ is equivalent to solving \eqref{C-LP} and we will focus on this objective. 

For any $\cD\subseteq \cC$, the dual linear program of $\LP(\cD)$ is the following. 
\begin{equation}
	\label{dual}
	\dual(\cD):~~~~~
	\begin{aligned}
		~~~~~ \max\quad        & ~~~~~\sum_{\ell \in I} \bar{\lambda}_{\ell}                                   \\
		\textsf{s.t.\quad} & ~~~\sum_{\ell \in C} \bar{\lambda}_{\ell} \leq 1   & \forall C\in \cD~~~~~\\  
		& ~~~~~\bar{\lambda}_{\ell} \geq 0 ~~~~~~~~&~~~~~~~~~~~~~~~~~~~~ \forall \ell \in I~~~~
	\end{aligned}
\end{equation}
We note that $\dual(\cC)$ can be solved in polynomial time given a separation oracle (see Theorem 6.3.2 in~\cite{grotschel2012geometric}). However, since no such separation oracle exists, we apply a technique dating back to~\cite{karmarkar1982efficient}  in order to obtain an approximate solution.

Given $\cD\subseteq \cC$ and $v\in \mathbb{R}$, define a polytope
\begin{equation}
	\label{eq:F_def}
	F(\cD, v) = \left\{\bar{\lambda}\in \mathbb{R}^I_{\geq 0} ~~\middle|~~\begin{aligned}
	\sum_{\ell \in I} \bar{\lambda}_{\ell} \geq v \\ 
	\sum_{\ell \in C} \bar{\lambda}_{\ell} \leq 1 &~~~~~~~\forall C\in \cD
\end{aligned}\right\}
\end{equation}
Clearly, $\OPT(\dual(\cD)) \geq v$ if and only if $F(\cD,v) \neq \emptyset$. 

Observe that, for every $\cD\subseteq \cC$ and $v\in \mathbb{R}_{\geq 0}$, each of the inequalities in the definitions $F(\cD, v)$ can be represented using $\varphi= O(|I| + \|v\|)$ bits, where $\|v\|$ is the size of the representation for the number~$v$; thus, $(F(\cD, v), n ,\varphi)  $  is a well-described polyhedron (Definition 6.2.2 in~\cite{grotschel2012geometric}). 
By Theorem 6.4.1 in~\cite{grotschel2012geometric} there is an algorithm $\ellip$ which given a separation oracle for $F(\cD,v)$ determines if $F(\cD,v)\neq \emptyset$ in time $\poly(|I|, \|v\|)$. 

We use $\ellip$ to determine if $F(\cC,v)\neq \emptyset$ with the following (flawed) separation oracle. Given $\bar{\lambda}\in \mathbb{R}^I_{\geq 0}$, the oracle first checks if $\sum_{\ell \in I} \bar{\lambda}_{\ell} <v$. If this is the case, the algorithm returns $\sum_{\ell \in I} \bar{\lambda_{\ell}} <v$ as a separating hyperplane.  Otherwise, the algorithm runs the FPTAS for CMP with the instance~$\II$, the weight function $w(\ell)= \bar{\lambda}_{\ell}$ and $\frac{\eps}{10}$ as the error. If the FPTAS returned $C\in \cC$ such that $\sum_{\ell \in C} \bar{\lambda}_{\ell}>1$, the algorithm returns  $\sum_{\ell \in C} \bar{\lambda}_{\ell}>1$ as a separating hyperplane. If the configuration returned by the FPTAS does not meet this  condition, then  the oracle aborts the execution of $\ellip$.

Observe that the execution of the separation oracle runs in time $\poly(|\cI|, \frac{1}{\eps}, \|v\|)$. Thus, the execution of $\ellip$ with the oracle terminates in polynomial time. The execution can either end with a declaration that $F(\cC,v)=\emptyset$ or be aborted by the separation oracle. Consider each of these two cases:
\begin{itemize}
	\item
The execution terminated by declaring that $F(\cC,v)=\emptyset$. Then, as the separating hyperplanes returned by the oracle are indeed separating hyperplanes, it follows that $F(\cC,v)=\emptyset$  is a correct statement. Let $\cD$ be the set of configurations returned as separating hyperplanes throughout the execution. Then, as all the separating hyperplanes returned are separating hyperplanes for $F(\cD, v)$ we conclude that $F(\cD,v)=\emptyset$ as well. Furthermore, $|\cD|$ is polynomial in $|I|$ and $\|v\|$, as the running time of $\ellip$ is polynomial in these variables. 

\item Otherwise, the execution of $\ellip$ has been aborted. Let $\bar{\lambda}$ be the value given to the separation oracle on its last call (the one which ended up with the abortion). It follows that $\sum_{\ell \in I} \bar{\lambda}_{\ell}\geq v$ and $\sum_{\ell \in C} \bar{\lambda}_{\ell} \leq \frac{1}{1-\frac{\eps}{10}}$ for all $C \in \cC$ (otherwise the FPTAS must return a solution $C$ for which $\sum_{\ell \in C} \bar{\lambda}_{\ell} >1$). Then it holds that $\left(1-\frac{\eps} {10}\right) \bar{\lambda} \in F\left(\cC,\left(1-\frac{\eps} {10}\right)\cdot v \right)$, and consequently $F\left(\cC,\left(1-\frac{\eps} {10}\right)\cdot v \right)\neq \emptyset$.
\end{itemize}

Thus, using a binary search, we can find $v\in [0, n]$ and $\cD\subseteq \cC$ such that $F(\cC,v)=F(\cD, v)=\emptyset$, $F\left(\cC, \left( 1-\frac{\eps}{2}\right) \cdot v\right)\neq \emptyset$, and $|\cD|$ is polynomial in $|I|$.  As $F\left(\cC, \left( 1-\frac{\eps}{2}\right) \cdot v\right)\neq \emptyset$ it follows that $\OPT(\dual(\cC)) \geq\left( 1-\frac{\eps}{2}\right) \cdot  v$, and by strong duality it holds that $\OPT(\LP(\cC)) \geq \left( 1-\frac{\eps}{2}\right) \cdot v$.   Furthermore, as $F(\cC,v)=F(\cD, v)=\emptyset$ it follows that $\OPT(\dual(\cD))\leq v$, and by strong duality $\OPT(\LP(\cD))\leq v$. As $|\cD|$ is polynomial we can solve $\LP(\cD)$ in polynomial time and obtain a solution $\bx$ (which is also a solution for $\LP(\cC)$), such that $$\|\bx\| =\OPT(\LP(\cD)) \leq v \leq \frac{\OPT(\LP(\cC))}{1-\frac{\eps}{2}} \leq (1+\eps)\cdot \OPT(\LP(\cC)).$$

Overall, we obtained a $(1+\eps)$-approximate solution in $\poly(|\cI|, \frac{1}{\eps})$ time, as required. 

 \qed

\noindent{\bf Proof of Lemma~\ref{lem:AFPTAS}:} 
Observe that the optimum of \eqref{C-LP} is at most $\OPT(\II)$, thus by Lemma~\ref{configurationLP} it holds that $\bar{x}$ is a solution to the configuration LP \eqref{C-LP} of $\mathcal{I}$ and  $\|\bx\|\leq(1+\eps)\OPT(\mathcal{I})$. Therefore, by Lemma~\ref{lem:eviction}, Algorithm $\textsf{Evict}(\eps,\cI,\bar{x})$ in Step~\ref{step:evicAFPTAS} returns a prototype $\bar{y}$ with $\bar{\gamma}$ in the $\bar{y}$-polytope such that (i) for all $\ell,j \in I, \ell \neq j$ it holds that $\bar{\gamma}_{\ell,j} = 0$; (ii) for all $C \in \textnormal{\textsf{supp}}(\bar{y})$ it holds that $|C| \leq \eps^{-10}$ and $s(C \setminus L) \leq \eps$; (iii) $\|\bar{y}\| \leq (1+\eps)\|\bar{x}\|$; (iv) for all $\ell \in I$ it holds that $\sum_{C \in \cC[\ell]} \by_C \leq 2$.  Then, in Step~\ref{GetPolytope}, Algorithm $\textsf{Shift}(\eps,\cI,\bar{y})$ returns a good prototype $\bar{z}$ by Lemma~\ref{lem:Cnf} such that

\begin{equation}
	\label{final1}
	\begin{aligned}
		\|\bar{z}\| \leq{} & (1+5\eps)\|\bar{y}\|+Q(\eps) \leq (1+5\eps)\cdot (1+\eps)\cdot \|\bar{x}\|+Q(\eps) \\ \leq{} & (1+9\eps)(1+\eps) \OPT(\mathcal{I})+Q(\eps)\leq (1+19\eps) \OPT(\mathcal{I})+Q(\eps).
	\end{aligned}
\end{equation} The first inequality is by Lemma~\ref{lem:Cnf} . The second inequality is by Lemma~\ref{lem:eviction}. The third inequality is by   Lemma~\ref{configurationLP} and because $0<\eps < 0.1$. Consequently, by Lemma~\ref{FromPolytope}, in Step~\ref{step:BPPnice} it holds that $\mathcal{B}$ is an $\eps$-nice partition with size at most: 

\begin{equation}
	\label{final2}
	\|\bar{z}\| + \eps^{-22}Q^2(\eps) \leq (1+19\eps) \OPT(\mathcal{I})+Q(\eps)+ \eps^{-22}Q^2(\eps) \leq (1+19\eps) \OPT(\mathcal{I})+\eps^{-23}Q^2(\eps)
\end{equation} The first inequality is \eqref{final1}. The second inequality is because $0<\eps<0.1$. Then, by Lemma~\ref{lem:GREEDY}, in Step~\ref{step:greedy1}, a full packing $\Phi$ of $\mathcal{I}$ is constructed. The number of bins in $\Phi$ is bounded by:

\begin{equation*}
	\label{final3}
	\begin{aligned}
		\#\textsf{bins}(\Phi) \leq{} & (1+2\eps) \left( (1+19\eps) \OPT(\mathcal{I})+\eps^{-23}Q^2(\eps) \right) +2\eps^{-22}Q^2(\eps) \leq (1+60\eps) \OPT(\mathcal{I})+3\eps^{-23}Q^2(\eps)\\
		\leq{} &(1+60\eps) \OPT(\mathcal{I})+\eps^{-24}Q^2(\eps) \leq (1+60\eps) \OPT(\mathcal{I})+Q^3(\eps).
	\end{aligned}
\end{equation*} The first inequality is by \eqref{final2} and Lemma~\ref{lem:GREEDY}. The other inequalities hold as $0<\eps<0.1$. \qed

\noindent{\bf Proof of Lemma~\ref{lem:gen_afptas}}

\begin{algorithm}[h]
	\caption{$\genalg({\cj}, \eps)$}
	\label{alg:genalg}
	
	\SetKwInOut{Input}{Input}
	\SetKwInOut{Output}{Output}
	\Input{An  BPP instance $\cj$ and $\eps\in (0,0.1)$ such that  $\eps^{-1}\in \mathbb{N}$}
	\Output{A packing of $\cj$}

	Compute an $\eps$-structured instance $\mathcal{I}$ by $\textsf{Reduce}({\cal J},\eps)$ \label{step:reduction}\;
	
	Compute $\Phi =\textsf{AFPTAS}(\II,\eps)$ (Algorithm~\ref{alg:Fscheme}) \label{step:solve_structured}\;

	Return a packing for ${\cal J}$ by $A = \textsf{Reconstruct}(\cj,\eps,\Phi)$ \label{step:recon1}
\end{algorithm}
The pseudo-code of $\genalg$ is given in Algorithm~\ref{alg:genalg}. 
Consider the execution of $\genalg$ with a BPP instance $\cj$ and $\eps\in(0,0.1)$ such that $\eps^{-1}\in \mathbb{N}$.
By Lemma~\ref{lem:reductionReconstruction} it holds that the instance~$\II$ computed in Step~\ref{step:reduction} is an $\eps$-structured instance and $\OPT(\II) \leq \OPT(\cj)$. Thus by Lemma \ref{lem:AFPTAS} it holds $\Phi$ (calculated in Step~\ref{step:solve_structured}) is a packing of $\II$ which uses at most 
$$ (1+60\eps) \cdot \OPT(\II) + Q^3(\eps)\leq  (1+60\eps) \cdot \OPT(\cj) + Q^3(\eps)$$
bins.  Thus, by Lemma~\ref{lem:reductionReconstruction}, it holds that the packing $A$ returned by the algorithm is a packing of $\cj$ which uses at most 
$$
(1+60\eps) \cdot \OPT(\cj) + Q^3(\eps) +13\eps\cdot  \OPT(\cj)  +1 \leq (1+130) \cdot \OPT(\cj) +3\cdot Q^{3}(\eps) 
$$
bins.

It easily follows from Lemmas~\ref{lem:AFPTAS} and~\ref{lem:reductionReconstruction} that  the overall running time of the Algorithm~\ref{alg:genalg} is $\poly(|\cj|, \frac{1}{\eps})$. 
\qed

\noindent{\bf Proof of Theorem~\ref{thm:main}:}
Let $\mathcal{J} =  (I, \mathcal{G},s,k)$ be a BPP instance. Recall that $V(\cj) = \max_{G \in \mathcal{G}} \ceil{\frac{|G|}{k(G)}}$ is the maximum cardinality of a group in $\mathcal{J}$ divided by the cardinality constraint of the group and let $W =s(I)+ V(\cj)+c$ where $c = \exp\left(\exp\left({100^{17}}\right)\right)$ is a large constant. 
\begin{claim}
	\label{clm:Weps}
$$ 
\OPT(\mathcal{J}) \leq 2\cdot W
$$
\end{claim}
\begin{proof}
	Let $L_{\frac{1}{2}}= \{ \ell \in I~|~s(\ell) > \frac{1}{2}\}$. By Lemma~\ref{thm:greedy} it follows that $\greedy$ finds a packing of the instance $\cj \setminus L_{\frac{1}{2}}$ using 
	at most $$\left(1+2\cdot \frac{1}{2}\right) \cdot\max\left\{s(I\setminus L_{\frac{1}{2}}), V(\cj\setminus L_{\frac{1}{2}}) \right\}+2\leq 2 \cdot\max \left\{s(I\setminus L_{\frac{1}{2}}), V(\cj) \right\}+2$$
	bins, where we define $\cj \setminus L_{\frac{1}{2}} = (I_{L_{\frac{1}{2}}},\cG_{L_{\frac{1}{2}}},s_{L_{\frac{1}{2}}},k_{L_{\frac{1}{2}}})$ as the BPP instance such that \begin{equation*}
		\begin{aligned}
			I_{L_{\frac{1}{2}}} = I \setminus L_{\frac{1}{2}}, ~~~~~~~~~~~~~~~~~~~~ ~~~~~~
			\mathcal{G}_{L_{\frac{1}{2}}} = \{G \cap I_{L_{\frac{1}{2}}} \neq \emptyset~|~ G \in \cG\}, ~~~~~~~~~~~~~~~~~~~~~~~~~~~~~~
			\\
			s_{L_{\frac{1}{2}}}(\ell) = s(\ell), ~~\forall \ell \in I_{L_{\frac{1}{2}}}, ~~~~~~~~~~~
			k_{L_{\frac{1}{2}}}(G \cap I_{L_{\frac{1}{2}}}) = k(G), ~~\forall G \in \cG \text{ s.t. } G \cap I_{L_{\frac{1}{2}}} \neq \emptyset.~~~~~~
		\end{aligned}
	\end{equation*} The items in $L_{\frac{1}{2}}$ can be packed  into $|L_{\frac{1}{2}}|\leq 2 \cdot s \left( L_{\frac{1}{2}}\right) $ bins (that is, partitioned into $|L_{\frac{1}{2}}|$ configurations) with  a single item per bin. Thus, set of items $I$ can be packed into number of bins bounded by
	$$
	2 \cdot\max \left\{s(I\setminus L_{\frac{1}{2}}), V(\cj) \right\}+2 + 2 \cdot s \left( L_{\frac{1}{2}}\right)  \leq 
	2 \left( s(I) + V(\cj)\right)+2 \leq 2 \cdot W
	$$
\end{proof}
 Now, define $\eps = \floor{\left(  \ln \ln W \right)^{\frac{1}{17}}}^{-1}$. Since $W \geq c$, it holds that  $0<\eps <0.1$.
In the following we show that the running time of $\genalg$ on the input $\mathcal{J}$  and $\eps$, as defined above, is polynomial in $|\mathcal{J}|$, while the number of bins in the packing returned by the algorithm is $\OPT(\mathcal{J})+o(\OPT(\mathcal{J}))$.

By Lemma~\ref{lem:gen_afptas}, there are polynomial functions $f,g$ such that the running time of $\genalg \left( \mathcal{J}, \eps \right)$ is at most $f\left(\frac{1}{\eps}\right) \cdot g(|\mathcal{J}|)$. We assume without the loss of generality that $f$ is monotone. Therefore, 
\begin{equation}
	\label{eq:polymain0}
	f\left(\frac{1}{\eps}\right)  \cdot g(|\mathcal{J}|) = f\left(\floor{ \left(  \ln \ln W \right)^{\frac{1}{17}}}\right)  \cdot g(|\mathcal{J}|) \leq  f(W)  \cdot g(|\mathcal{J}|) \leq f(2\cdot |\cj|+c) \cdot g(|\cj|)
\end{equation}
The first equality is by the definition of $\eps$. The first inequality is because $W>c>1$ and for all $x \geq 1$ it holds that $\ln x \leq x$.  The last inequality hold as $V(\cI), s(I)\leq |I|\leq |\cj|$. 
 Since $c$ is a constant it follows from  the \eqref{eq:polymain0} that the running time of $\genalg(\cj,\eps)$ is polynomial in $|\cj|$.

By Lemma~\ref{lem:gen_afptas},  it holds that the packing $A$ returned by $\textsf{Gen-AFPTAS} \left( \mathcal{J}, \eps \right)$ is a packing of $\mathcal{J}$ which uses at most  $ (1+130\eps)OPT(\mathcal{I})+3\cdot Q^3(\eps)$ bins. 
Now,  \begin{equation}
	\label{eq:Q}
	\begin{aligned}
			Q(\eps) ={} & \exp(\eps^{-17}) =  \exp \left( \left(   \floor{\left( \ln \ln W \right)^{\frac{1}{17}}}^{-1} \right)^{-17} \right) \\ \leq{} & \exp \left(   \ln \ln W \right) = \ln  W \leq \ln  \left( 2\cdot\OPT(\mathcal{J}) +c \right). 
	\end{aligned}
\end{equation}

The first equality is by the definition of $Q(\eps)$. The second equality is by the definition of $\eps$.  The last inequality holds by $\OPT(\cj)\geq s(I)$, $\OPT(\cj)\geq V(\cI)$  and the definition of $W$. Hence,  the number of bins used by $A$ is at most
\begin{equation}
	\begin{aligned}
		\label{eq:n1}
	& (1+130\eps)\OPT(\mathcal{J})+3\cdot Q^3(\eps) \\
		\leq& \OPT(\mathcal{J})+  130 \floor{\left(  \ln \ln W \right)^{\frac{1}{17}}}^{-1}\OPT(\mathcal{J})+3\cdot\ln^3  \left( 2\cdot\OPT(\mathcal{J}) +c \right)\\
		\leq & \OPT(\mathcal{J})+  130\cdot \frac{ \OPT(\cj)}{\left(\ln \ln\left( \frac{\OPT(\cj)}{2}\right) \right)^{\frac{1}{17}} -1} +3\cdot \ln^3  \left( 2\cdot\OPT(\mathcal{J}) +c \right) \\  ={} &  \OPT(\mathcal{J})+O\left(\frac{\OPT(\mathcal{J})}{ \left(\ln \ln \OPT(\cj) \right)^{\frac{1}{17}}}\right).
	\end{aligned}
\end{equation}
The  first inequality is by the definition of $\eps$ and \eqref{eq:Q}. The second inequality holds since $W\geq \frac{\OPT(\cj)}{2}$ by Claim~\ref{clm:Weps}.

  Overall, we showed that the algorithm returns a packing of $\cj$ using $ \OPT(\mathcal{J})+O\left(\frac{\OPT(\mathcal{J})}{ \left(\ln \ln \OPT(\cj) \right)^{\frac{1}{17}}}\right)$ bins in polynomial time in $|\cj|$. \qed

\section{Properties of the $\bar{x}$-polytope}
\label{sec:poly}
\newcommand{\bm}{\bar{m}}
\newcommand{\bgam}{\bar{\gamma}}
\newcommand{\bb}{\bar{b}}

 In this section we give the proof of Lemma~\ref{O(1)}. Let $\bar{x}$ be a prototype of $\mathcal{I}$ which satisfy the conditions of Lemma~\ref{O(1)}: there is an integer $k$ such that for all $C \in \textsf{supp}(\bar{x})$ it holds that  $|C| \leq k$ and $\bar{x}_C \in \mathbb{N}$; also, let $\bar{\lambda}$ be a vertex of the $\bar{x}$-polytope for which constraint \eqref{F4} holds with equality.  We say that a constraint is {\em tight} if it holds with equality. Define the {\em active} types of $\bar{x}$ as $$T = \textsf{supp}(\bar{x}) \cup \{j \in I~|~ \exists C \in \textsf{supp}(\bar{x}) \text{ s.t. } j \in C\}.$$

For every $t \in T$ and $G \in \cG$, we define below a value $L_{t,G}$. Let $C \in \cC \cap T$,$j \in I \cap T$, and $G \in \cG$. Define $L_{C,G} = \bar{x}_C \cdot \left(   k(G)-|G \cap C|  \right)$ and $L_{j,G} = \sum_{C' \in \cC[j]} \bar{x}_{C'}$. For every  $(\ell,t) \in I \times (I \cup \cC)$, we say the entry $\bar{\lambda}_{\ell,t}$ {\em corresponds} to $\ell$ and $t$. Note that all entries $\bar{\lambda}_{\ell,t}$ such that $t \in (I \cup \cC) \setminus T$ (i.e., not corresponding to an active type) are required to be zero by the definition of the $\bar{x}$-polytope.

  For the following, fix a group $G \in \cG$. A {\em movement} of $G$ is a vector $\bm \in \mathbb{R}^{I  \times (I \cup \cC)}$ which is used as a relaxation of the constraints corresponding to items in $G$ in the $\bar{x}$-polytope:

 \begin{align}
		\displaystyle \bar{m}_{\ell,t} = 0   ~~~~~~~~~~~~~~~~~~~~~~~~~~~~~~~~~~~~ ~~~~~~	~~~~~~~~~~\forall \ell\in I, t \in T \text{ s.t. } \ell \notin G  \textnormal{ or } \bar{\lambda}_{\ell, t} \in \{0,1\} ~ \label{m0}\\
		\rule{0pt}{1.8em}
		\displaystyle  \sum_{t \in T} \bm_{\ell,t} = 0 ~~~~~~~~~~~~~~~~~~~~~~~~~~~~~~~~~~~~~~~~~~~~~~~\forall \ell \in I ~~~~~~~~~~~~~~~~~~~~~~~~~~~~~~~~~~~~~~~~~~~ \label{ml}\\
		\rule{0pt}{1.8em}
		\displaystyle    \sum_{\ell \in G} \bar{m}_{\ell,t} = 0~~~~~~~~~~~~~~~~~~~~~~~~~~~~  ~~~~~~~~~~	~~~~~~~~~~\forall t \in T \text{ s.t.} \sum_{\ell \in G} \bar{\lambda}_{\ell,t} =L_{t,G}~~~~~~~~~~~~~~~~~~\label{mT}
	\end{align}

 We say that $G$ is {\em fractional} if there are $\ell \in G$ and $t \in T$ such that $\bar{\lambda}_{\ell,t} \in (0,1)$.  

\begin{claim}
	\label{lem:movement}
	For every group $G \in \cG$, If $G$ is fractional then $G$ has a movement $\bm \neq 0$.
\end{claim}
The proof of Claim~\ref{lem:movement} utilizes properties of {\em totally unimodular} matrices. A matrix $A$ is totally unimodular if every square submatrix of $A$ has a determinant $1$, $-1$ or $0$. If $A\in \mathbb{R}^{n\times m}$ is totally unimodular and $b\in \mathbb{Z}^{m}$ is an integral vector, it holds that the vertices of the polytope  $$P_A =\{\bar{v} \in \mathbb{R}_{\geq 0}^n~|~A \bar{v} \leq \bb\}$$ are integral. That is,  if $\bar{v} \in P_A$ is a vertex  of $P_A$ then $\bar{v} \in \mathbb{Z}^n$ \cite{HK56}.  We use the following criteria for total unimodularity, which is a simplified version of a theorem from~\cite{HK56}.
\begin{lemma}
	\label{lem:unimodular}
	Let $A\in \mathbb{R}^{n\times m}$ be a matrix which satisfies the following properties.
	\begin{itemize}
		\item All the entries of $A$ are  in  $\{-1,1,0\}$.
		\item Every column of $A$ has up to two non-zero entries.
		\item If a column of $A$ has two non-zero entries, then these entries have opposite signs.
		\end{itemize}
	Then, $A$ is totally unimodular.
	\end{lemma}

\posA
\noindent{\bf Proof of Claim~\ref{lem:movement}:}
Let $G \in \cG$ be a fractional group. Define $X = \{ (\ell,t) \in G \times T~|~ \ell \notin \textsf{fit}(t)\}$  as all {\em forbidden pairs} of $G$, where for all $(\ell,t) \in X$ it holds that $\bar{\lambda}_{\ell,t} = 0$ by \eqref{F1}.  We show the existence of the movement $\bm \neq 0$ using the polytope $P$ defined as follows. 

	\begin{equation}
		\label{eq:Pj_def}
	P= \left\{ \by\in \mathbb{R}_{\geq 0}^{G \times T} ~\middle |~ \begin{array}{lcc}
		\displaystyle \forall (\ell,t) \in X&:& \displaystyle\by_{\ell,t}\leq 0\\
			\rule{0pt}{1.8em}
	 \displaystyle	\forall \ell\in G&:& \displaystyle  \sum_{t \in T \textnormal{ s.t. } (\ell,t) \in  \notin X} -\by_{\ell,t} \leq -1\\
	 	\rule{0pt}{1.8em}
	 \displaystyle \forall t \in T&:& 
	 \displaystyle  \sum_{\ell\in G \textnormal{ s.t. } (\ell,t)\notin X } \by_{\ell,t} \leq L_{t,G}
	\end{array}\right\}.
	\end{equation}

Define a vector $\bar{\phi} \in \mathbb{R}^{G \times T}$ by $\bar{\phi}_{\ell,t}=\bar{\lambda}_{\ell,t}$ for all $(\ell,t) \in G \times T$. It follows from the definition of the $\bar{x}$-polytope that $\bar{\phi} \in P$. We can represent the inequalities in \eqref{eq:Pj_def} using a matrix notation as  $$P =\left\{\by \in \mathbb{R}^{G \times T}_{\geq 0} ~\middle|~ A\by \leq \bb \right\}.$$ It follows that $A$ contains only entries in $\{-1,0,1\}$ and the entries in $\bb$ are all integral. Furthermore, every column of $A$ contains at most $2$ non-zero entries, and if there are two non-zero entries in a column then they are of a different sign. By Lemma~\ref{lem:unimodular}, it follows that $A$ is totally unimodular. It thus holds that all the vertices of the polytope $P$ are integral. As $\bar{\phi} \in P$ is non-integral (since $G$ is fractional), it follows that $\bar{\phi}$ is not a vertex of $P$. Hence, there is a vector $\bgam\in \mathbb{R}^{G \times T}$, $\bgam \neq 0$  such that $\bar{\phi}+\bgam, \bar{\phi}-\bgam \in P$. 

We define $\bm \in \mathbb{R}^{I \times (I \cup \cC)}$ by $\bm_{\ell,t}=\bgam_{\ell,t}$ for all $(\ell,t)\in G \times  T$ and $\bm_{\ell,t} =0$ otherwise. Clearly, $\bm \neq 0$ as $\bgam \neq 0$. Observe that for $\ell \in I \setminus G$ and $t \in (I \cup \cC)$ it holds that $\bm_{\ell,t}=0$ by definition. For $\ell\in G$ and $t \in (I \cup \cC)$ such that $\bar{\lambda}_{\ell,t} =0$, as $\bar{\phi}_{\ell,t}+\bgam_{\ell,t}, \bar{\phi}_{\ell,t} -\bgam_{\ell,t}\geq 0$ and $\bar{\phi}_{\ell,t}= \bar{\lambda}_{\ell,t}=0$,  it follows that $\bar{m}_{\ell,t} = \bgam_{\ell,t}=0$. 
 For $\ell \in G$ and $t \in I \cup \cC$ such that $\bar{\lambda}_{\ell,t} =1$,  for all $t'\in I \cup \cC \setminus \{t\}$ it holds that $\bar{\phi}_{\ell,t'}= \bar{\lambda}_{\ell,t'} =0$ because constraint \eqref{F4} in the $\bar{x}$-polytope is tight for $\bar{\lambda}$. Thus, by the previous argument, we have $\bgam_{\ell,t'}=0$ for every $t'\in I \cup \cC \setminus \{t\}$ in case that $\bar{\lambda}_{\ell,t} =1$. Therefore,

 \begin{equation}
 	\label{eq:gam_supp_first}
-1 -\bgam_{\ell,t}=- \sum_{t'\in I \cup \cC \textnormal{ s.t } (\ell,t') \notin X} \left(\bar{\phi}_{\ell,t'}  +\bgam_{\ell,t'}\right) \leq -1.
\end{equation} 
The equality is because $\bar{\lambda}$ and also $\bar{\phi}$ hold constraint \eqref{F4} with equality.  The last inequality is due to $ \bar{\phi}+\bgam \in P$. 
Similarly, as $\bar{\phi}-\bgam \in P$, we have
\begin{equation}
	 	\label{eq:gam_supp_second}
	 	-1 +\bgam_{\ell,t}=- \sum_{t'\in I \cup \cC \textnormal{ s.t } (\ell,t') \notin  X} \left(\bar{\phi}_{\ell,t'}  -\bgam_{\ell,t'}\right) \leq -1. 
\end{equation} 
By~\eqref{eq:gam_supp_first} and \eqref{eq:gam_supp_second} we have $\bgam_{\ell,t}=0$. Overall, we have that $\bm$ satisfies \eqref{m0}.

For any $\ell\in I \setminus G$ it holds that $\sum_{t \in T} \bm_{\ell,t} = 0 $ by \eqref{m0}. For $\ell\in G$, since  $\bar{\phi}+\bgam\in P$ we have
\begin{equation}
	\label{eq:t_sum_first}
-1 -\sum_{t \in T} \bgam_{\ell,t} = \sum_{t \in T}-(\bar{\lambda}_{\ell,t} +\bgam_{\ell,t}) =  
 \sum_{t \in T~\textnormal{ s.t. }(\ell,t)\notin X}-(\bar{\phi}_{\ell,t} +\bgam_{\ell,t}) \leq -1,
 \end{equation}

 where the first equality holds since $\bar{\lambda}$ satisfies with equality constraint \eqref{F4} of the $\bar{x}$-polytope and the second equality uses $\bar{\phi}_{\ell,t} +\bgam_{\ell,t}=0$ for $(\ell,t) \in X$ because $\phi+\bar{\gamma} \in P$. Similarly, since $\phi-\bgam\in P$, we have 
 \begin{equation}
 		\label{eq:t_sum_second}
 -1 +\sum_{t \in T} \bgam_{\ell,t} = \sum_{t \in T}-(\bar{\lambda}_{\ell,t} -\bgam_{\ell,t}) =  
 \sum_{t \in T~\textnormal{ s.t. }(\ell,t)\notin X}-(\bar{\phi}_{\ell,t} -\bgam_{\ell,t}) \leq -1.
 \end{equation}

By \eqref{eq:t_sum_first} and \eqref{eq:t_sum_second} we have 
  $\sum_{t \in T} \bm_{\ell,t}=\sum_{t \in T} \bgam_{\ell,t}=0$. Thus, $\bm$ satisfies \eqref{ml}. Finally, let $t \in T$ such that $\sum_{\ell\in G} \bar{\lambda}_{\ell,t} = L_{t,G}$ . As before,
  
  \begin{equation}
  	\label{eq:gam_group_sum_first}
  	L_{t,G} +\sum_{\ell\in G} \bgam_{\ell,t}= \sum_{\ell\in G} \bar{\lambda}_{\ell,t}+\sum_{\ell\in G} \bgam_{\ell,t} =  \sum_{\ell\in G} \left(\bar{\phi}_{\ell,t} +\bgam_{\ell,t}\right)
  	= \sum_{\ell\in G \textnormal{ s.t. } (\ell,t)\notin X} \left(\bar{\phi}_{\ell,t} +\bgam_{\ell,t}\right)  \leq L_{t,G}, 
  	\end{equation}
  where the inequality follows from $\bar{\phi}+\bgam \in P$. Using a similar argument,
  
   \begin{equation}
   		\label{eq:gam_group_sum_second}
  	L_{t,G} -\sum_{\ell\in G} \bgam_{\ell,t}= \sum_{\ell\in G} \bar{\lambda}_{\ell,t}-\sum_{\ell\in G} \bgam_{\ell,t} =  \sum_{\ell\in G} \left(\bar{\phi}_{\ell,t} -\bgam_{\ell,t}\right)
  	= \sum_{\ell\in G \textnormal{ s.t. } (\ell,t)\notin X} \left(\bar{\phi}_{\ell,t} -\bgam_{\ell,t}\right)  \leq L_{t,G}.
  \end{equation}
  
By \eqref{eq:gam_group_sum_first} and \eqref{eq:gam_group_sum_second}, we have $\sum_{\ell\in G} \bm_{\ell,t} = \sum_{\ell\in G} \bgam_{\ell,t} =0$. Thus, $\bm$ satisfies \eqref{mT}. Overall, we show that $\bm \neq 0$ is a movement of $G$.  
\qed

\begin{claim}
	\label{lem:FewFractionalGroups}
	There are at most $|T|$ fractional groups.
\end{claim}

\begin{proof}
	Assume towards a contradiction that there are $|T|+1$ fractional groups. Denote these groups by $G_1, \ldots, G_{|T|+1} \in \cG$. By Lemma~\ref{lem:movement}, for all $j \in [|T|+1]$ it holds that $G_j$ has a movement $\bm^j \neq 0$. Consider the following set of equalities over $a_1, \ldots, a_{|T|+1}$:
	\begin{equation}
		\label{eq:homegeneous} \forall t \in T: ~~~~ \sum_{j=1}^{|T|+1} a_j \sum_{\ell \in I} \bm^j_{\ell,t} \cdot s(\ell) = 0.
	\end{equation}

	These are $|T|$ homogeneous linear equalities in $|T|+1$ variables. Thus, there exist $a_1, \ldots, a_{|T|+1} \in \mathbb{R}$, not all zeros, for which  \eqref{eq:homegeneous} holds.	Define 
	
	$$K_1 = \min_{(\ell,t)\in I \times  T}\min \left\{ \bar{\lambda}_{\ell,t}, 1-\bar{\lambda}_{\ell,t} ~\middle|~0<\bar{\lambda}_{\ell,t}<1\right\}, $$ 
	$$	K_2= \min \left\{
	L_{t,G_j} - \sum_{\ell \in G_j} \bar{\lambda}_{\ell,t}~\middle |~ j \in [ |T|+1], t \in T\text{ s.t. } \sum_{\ell \in G_j}\bar{\lambda}_{\ell,t} < L_{t,G_j}  \right\},
	$$  
	and $K= \min\{K_1,K_2\}$.  Additionally, define 
	$$W=\max\left\{ |\bm^{j}_{\ell,t}|~|~1 \leq j \leq |T|+1,~ (\ell,t)\in I \times  T\right\}.$$

	Observe that $W,K>0$. Using a scaling argument, we may assume that  for any $1\leq j\leq |T|+1$ it holds that $$a_j \leq \frac{K}{|I| \cdot W}.$$ Define 
	
	$$\bar{\phi} = \bar{\lambda} + \sum_{j=1}^{|T|+1} a_j \cdot \bm^j.$$

	In the following we show that $\bar{\phi}$ is in the $\bar{x}$-polytope. For every $\ell\in I\setminus (G_1\cup \ldots \cup G_{|T|+1})$  and $t \in T$ it holds that $\bm^{1}_{\ell,t}= \ldots = \bm^{|T|+1}_{\ell,t} = 0$ due to \eqref{m0}. Also, for all $\ell \in I$ and $t \in I \cup \cC \setminus T$, it holds that $\bm^{1}_{\ell,t}= \ldots = \bm^{|T|+1}_{\ell,t} = 0$ due to \eqref{m0} as $\bar{\lambda}_{\ell,t} = 0$.  Thus,
	$\bar{\phi}_{\ell,t} = \bar{\lambda}_{\ell,t}\in  [0,1]$ in this case. Also, for $\ell \in I$ and $t \in T$ such that $ \bar{\lambda}_{\ell,t}\in  \{0,1\}$ it holds that $\bar{\phi}_{\ell,t} = \bar{\lambda}_{\ell,t} = 0$ by \eqref{m0}. Finally, for $j \in [|T|+1]$, $\ell \in G_j$, and $t \in T$ such that $ \bar{\lambda}_{\ell,t} \in (0,1)$, by \eqref{m0} it holds that $\bar{\phi}_{\ell,t} =  \bar{\lambda}_{\ell,t} + a_j \cdot \bm^j_{\ell,t}$. Following the definitions  of $K$ and $W$, we have 
	$$ 0\leq K -W \cdot \frac{K}{|I|\cdot W} \leq \bar{\lambda}_{\ell,t} + a_j \cdot \bm^j_{\ell,t} \leq 1-K + W \cdot \frac{K}{|I|\cdot W} \leq 1.$$
	Thus, $\bar{\phi} \in [0,1]^{I \times  (I \cup \cC)}$. 	
	
	For every $(\ell,t)\in I \times T$ such that $\ell \notin \textsf{fit}(t)$, it holds that $\bar{\lambda}_{\ell,t}=0$ by \eqref{F1} since $ \bar{\lambda}$ is in the $\bar{x}$-polytope. Thus by \eqref{m0} we have $\bm^{j}_{\ell,t} =0$ for all $1\leq j \leq |T|+1$. Therefore, $\bar{\phi}_{\ell,t} =0$. We conclude that constraint \eqref{F1} is satisfied for $\bar{\phi}$. 
	
	For every $t \in \cC$ it holds that 
	
	$$
	\sum_{\ell \in I} \bar{\phi}_{\ell,t} \cdot s(\ell) = 
	\sum_{\ell\in I} \bar{\lambda}_{\ell,t}\cdot s(\ell)+ \sum_{j=1}^{|T|+1} a_j \sum_{\ell\in I}
	\bm^j_{\ell,t} s(\ell)= \sum_{\ell\in I} \bar{\lambda}_{\ell,t}\cdot s(\ell) \leq (1-s(t)) \bar{x}_t .$$
	
	where the second equality is by \eqref{eq:homegeneous}, and the inequality is due to $\bar{\lambda}$ is in the $\bar{x}$-polytope; hence, the inequality follows by \eqref{F2}.

Let $G \in \mathcal{G}$ and $C \in \cC$. If for all $1 \leq j \leq |T|+1$ it holds that $G \neq G_j$,  then $\bar{\phi}_{\ell,C} = \bar{\lambda}_{\ell,C}$ for every $\ell\in G$ by \eqref{m0}. Thus,  
	$\sum_{\ell\in G } \bar{\phi}_{\ell,C} = \sum_{\ell \in G } \bar{\lambda}_{\ell,C}\leq \bar{x}_C \cdot (k(G)-|G \cap C|)$, as $\bar{\lambda}$ is in the $\bar{x}$-polytope and thus satisfies constraint \eqref{F5}. Otherwise, there is $1\leq j \leq |T|+1$ such that $G = G_j$.  Observe that for every $\ell\in G$ it holds that $\bar{\phi}_{\ell,C} = \bar{\lambda}_{\ell,C}+a_j \bm^j_{\ell,C}$ by \eqref{m0}; 
	consider the following cases.
	
	\begin{itemize}
		\item $\sum_{\ell\in G} \bar{\lambda}_{\ell,C} <L_{C,G}$. Using the definitions of $K$ and $W$, we have
			\begin{equation*}
		    \label{eq:LGM1}
		    	\begin{aligned}
		       \sum_{\ell\in G} \bar{\phi}_{\ell,C}= 
		    \sum_{\ell\in G} \bar{\lambda}_{\ell,C} +a_j \sum_{\ell\in G}  \bm^j_{\ell,C}\leq 
		    L_{C,G} -K +  \frac{K}{|I|\cdot W  } \cdot |I| \cdot W = L_{C,G} = \bar{x}_C \cdot (k(G)-|G \cap C|).
		    \end{aligned}
		\end{equation*}
	
		\item $\sum_{\ell\in G} \bar{\lambda}_{\ell,C} =L_{C,G}$. By \eqref{mT} we have

			\begin{equation}
		    \label{eq:LGM2}
		    \sum_{\ell\in G} \bar{\phi}_{\ell,C}= 
	\sum_{\ell\in G} \bar{\lambda}_{\ell,C} + a_j\sum_{\ell \in G}  \bm^j_{\ell,C}=L_{C,G}  = \bar{x}_C  \cdot (k(G)-|G \cap C|).
		\end{equation}

	\end{itemize}
	
	We showed  that $\sum_{\ell \in G} \bar{\phi}_{\ell,C} \leq \bar{x}_C$ in all cases. We conclude that constraint \eqref{F5} is satisfied for $\bar{\phi}$. 
	
	Let $j \in I$. We split into two cases, as follows.

	\begin{enumerate}
		\item For all $1 \leq i \leq |T|+1$ it holds that $ \textsf{group}(j) \neq G_i$. Then, by \eqref{fitl},  for all $\ell \notin \textsf{group}(j)$ it holds that $\bar{\lambda}_{\ell,j} = 0$; thus, it holds that $\bar{\phi}_{\ell,j} = \bar{\lambda}_{\ell,j}$ for every $\ell \in I$ by \eqref{m0}. By the above,  $\sum_{\ell\in I } \bar{\phi}_{\ell,j} = \sum_{\ell \in I } \bar{\lambda}_{\ell,j} \leq \sum_{C \in \cC[j]} \bar{x}_C$, as $\bar{\lambda}$ is in the $\bar{x}$-polytope and thus satisfies constraint \eqref{F3}.

			\item There is $1\leq i \leq |T|+1$ such that $j \in G_i$. Observe that for every $\ell\in I$ it holds that $\bar{\phi}_{\ell,j} = \bar{\lambda}_{\ell,j}+a_i \bm^i_{\ell,j}$ by \eqref{m0} since for items $\ell \notin G_i$ it holds that $\bar{\lambda}_{\ell,j} = 0$; 
			consider the following cases. 
			
			\begin{itemize}
			\item $\sum_{\ell\in G_i} \bar{\lambda}_{\ell,j} <L_{j,G_i}$. Using the definitions of $K$ and $W$, we have
			\begin{equation*}
				\begin{aligned}
						\sum_{\ell\in I} \bar{\phi}_{\ell,j} ={} & 	\sum_{\ell\in I} \bar{\lambda}_{\ell,j} +a_i \sum_{\ell\in I}  \bm^i_{\ell,j} = 
					\sum_{\ell\in G_i} \bar{\lambda}_{\ell,j} +a_i \sum_{\ell\in G_i}  \bm^i_{\ell,j}\leq 
					L_{j,G_i} -K +  \frac{K}{|I|\cdot W  } \cdot |I| \cdot W \\ ={} & L_{j,G_i} = \sum_{C \in \cC[j]} \bar{x}_C.
				\end{aligned}
			\end{equation*}
			
			\item $\sum_{\ell\in G_i} \bar{\lambda}_{\ell,j} =L_{j,G_i}$. By \eqref{mT} we have

			\begin{equation}
				\label{eq:LGM2}
				\sum_{\ell\in I} \bar{\phi}_{\ell,j} = 	\sum_{\ell\in G_i} \bar{\phi}_{\ell,j} = 
				\sum_{\ell\in G_i} \bar{\lambda}_{\ell,j} + a_i \sum_{\ell \in G}  \bm^i_{\ell,j}=L_{j,G} = \sum_{C \in \cC[j]} \bar{x}_C.
			\end{equation}

		\end{itemize}

	\end{enumerate}

	We showed  that $\sum_{\ell \in I} \bar{\phi}_{\ell,j} \leq L_{j,G}$ in all cases. We conclude that constraint \eqref{F3} is satisfied for $\bar{\phi}$.

		For every $\ell \in I$, we have
	$$\sum_{t \in T} \bar{\phi}_{\ell,t} = 
	\sum_{t \in T} \bar{\lambda}_{\ell,t} + \sum_{j=1}^{|T|+1}
	a_j \sum_{t \in T} \bm^j_{\ell,t} = 
	\sum_{t \in T}  \bar{\lambda}_{\ell,t} = \sum_{t \in I \cup \cC}  \bar{\lambda}_{\ell,t}\geq 1,$$
	where the second equality is by \eqref{ml}, and the last inequality follows since $\bar{\lambda}$ is in the $\bar{x}$-polytope; thus, it satisfies constraint \eqref{F4}. Overall, we have that $\bar{\phi}$ is in the $\bar{x}$-polytope.
	
	We can also define $\bar{\phi}' =  \bar{\lambda}-\sum_{j=1}^{|T|+1} a_j\cdot \bm^j$. By a symmetric argument we can show that $\bar{\phi}'$ is in the $\bar{x}$-polytope as well. It also holds that $\bar{\phi},\bar{\phi}' \neq \bar{\lambda}$ as there is $1\leq j^* \leq |T|+1$ such that $a_{j^*}                                                                                                                                                                                                                                                                                                                                                                                                                                                                                                                                                                                                                                                                                                                                                                                                                                                                                                                                                                                                                                                                                                                                                                                                                                                                                                          \neq 0$; also, there are $\ell^* \in G_{j^*}$  and $t^*\in T$ such that $\bm^{j^*}_{\ell^*,t^*}\neq 0$. Thus, $\bar{\phi}_{\ell^*,j^*}=  \bar{\lambda}_{\ell^*,t^*} +a_{j^*} \bm^{j^*}_{\ell^*,j^*}\neq  \bar{\lambda}_{\ell^*,t^*}$ and $\bar{\phi}'_{\ell^*,j^*}=  \bar{\lambda}_{\ell^*,t^*} -a_{j^*} \bm^{j^*}_{\ell^*,j^*}\neq  \bar{\lambda}_{\ell^*,t^*}$. 
	Furthermore, $ \bar{\lambda} =\frac{1}{2} \cdot \bar{\phi}+\frac{1}{2}\cdot \bar{\phi}'$, and we conclude that $ \bar{\lambda}$ is not a vertex of the $\bar{x}$-polytope. Contradiction. 
\end{proof}

\begin{lemma}
	\label{lem:slot-types}

If a group $G \in \cG$ is fractional, then $$|\{\bar{\lambda}_{\ell,t} \in (0,1) ~|~ \ell \in G, t \in T\}| \leq 2|T|.$$ 

\end{lemma}

\begin{proof}
	Let $X = \{(\ell,t) \in G \times T~|~ \ell \notin \textsf{fit}(t)\}$. Define $R = G \times T \setminus X$ as the set of all pairs of items where the first fits in the place of the second. 
	Let  $Q \subseteq \mathbb{R}^{R}$ be the set (polytope) of  all the vectors $\bar{\rho} \in \mathbb{R}^{R}$ which satisfy the following constraints:

	\begin{align}
			\displaystyle 	\bar{\rho}_{\ell,t} \geq 0   ~~~~~~~~~~~~~~~~~~~~ ~~~~~~~~~~~~~~~	~~~~~~~~~~~\forall (\ell,t) \in R~~~~~~~	\label{eq:rho1}\\
			\rule{0pt}{1.8em}
			\displaystyle 	\sum_{t \in T \text{ s.t. } (\ell,t) \in R} \bar{\rho}_{\ell,t} \geq 1 ~~~~~~~~~~~~~~~~~~~~~~~~~~~~~~	~~~~~~\forall \ell \in  G ~~~~~~~~~~~~	\label{eq:rhoB}\\
			\rule{0pt}{1.8em}
			\displaystyle 	\sum_{\ell \in G \text{ s.t. } (\ell,t) \in R} \bar{\rho}_{\ell,C} \cdot s(\ell)+\sum_{\ell \in I \setminus G} \bar{\lambda}_{\ell,C} \cdot s(\ell) \leq (1-s(C)) \bar{x}_C ~~~~~~~~~~~~~	~~~~~~\forall C \in \cC \cap T ~~~~~~	\label{eq:rhoD}\\
			\rule{0pt}{1.8em}
			\displaystyle 	\sum_{\ell \in G \text{ s.t. } (\ell,t) \in R} \bar{\rho}_{\ell,t} \leq L_{t,G} ~~~~~~~~~~~~~~~~~~~~~~~~~~	~~~~~~\forall t \in T ~~~~~~~~~~~~	\label{eq:rhoS}
		\end{align}

		Now, we define $\bar{\phi} \in  \mathbb{R}^{R}$ by $\bar{\phi}_{\ell,t} = \bar{\lambda}_{\ell,t}$ for all $(\ell,t)\in R$. As $\bar{\lambda}$ belongs to the $\bar{x}$-polytope, 
		it easily follows that $\bar{\phi} \in Q$ since $\bar{\lambda}_{\ell,t}=0$ for all $(\ell,t) \in (I \times I \cup \cC) \setminus R$ by \eqref{F1} (i.e., where $\ell \notin \textsf{fit}(t)$).

		\begin{claim}
			$\bar{\phi}$ is  a vertex of $Q$.
		\end{claim}
		
		\begin{proof}
			Assume towards contradiction that $\bar{\phi}$ is not a vertex of $Q$. Thus, there is a vector $\bgam \in \mathbb{R}^{R}$, $\bgam \neq 0$  such that $\bar{\phi}+\bgam, \bar{\phi}-\bgam \in Q$.  Thus, by \eqref{eq:rhoB}, for all $\ell\in G$ we get that  
			
			$$
		  1 \leq 	\sum_{ t\in T \text{ s.t. } (\ell,t) \in R} (\bar{\phi}_{\ell,t}+\bar{\gamma}_{\ell,t}) = 	\sum_{t\in T\text{ s.t. } (\ell,t) \in R} \bar{\phi}_{\ell,t} +	\sum_{t\in T \text{ s.t. } (\ell,t) \in R} \bar{\gamma}_{\ell,t} = 1+\sum_{t\in T \text{ s.t. } (\ell,t) \in R} \bar{\gamma}_{\ell,t} 
			$$
			
			The last equality is because $\bar{\lambda}$ satisfies constraint \eqref{F4} of the $\bar{x}$-polytope with equality and by the definition of $\phi$. In addition,  
			$$
		1 \leq	\sum_{t\in T \text{ s.t. } (\ell,t) \in R} (\bar{\phi}_{\ell,t}-\bar{\gamma}_{\ell,t}) = \sum_{t\in T \text{ s.t. } (\ell,t) \in R} \bar{\phi}_{\ell,t} -	\sum_{t\in T \text{ s.t. } (\ell,t) \in R} \bar{\gamma}_{\ell,t} = 1-\sum_{t\in T \text{ s.t. } (\ell,t) \in R} \bar{\gamma}_{\ell,t} 
			$$
				The last equality is because $\bar{\lambda}$ satisfies constraint \eqref{F4} of the $\bar{x}$-polytope with equality and by the definition of $\phi$. Hence, by the above it follows that 
			\begin{equation}
				\sum_{t\in T\text{ s.t. } (\ell,t) \in R} \bar{\gamma}_{\ell,t}=0.
				\label{eq:gamzero}
			\end{equation}

			Furthermore, for all $t \in T$  the following hold by \eqref{eq:rhoS}:
			
			$$
	 L_{t,G} \geq 	\sum_{ \ell \in G \text{ s.t. } (\ell,t) \in R} (\bar{\phi}_{\ell,t}+\bar{\gamma}_{\ell,t})= 	\sum_{ \ell \in G \text{ s.t. } (\ell,t) \in R} \bar{\phi}_{\ell,t} +	\sum_{ \ell \in G \text{ s.t. } (\ell,t) \in R} \bar{\gamma}_{\ell,t}
			$$
			
			and 
			$$
	L_{t,G} \geq 	\sum_{ \ell \in G \text{ s.t. } (\ell,t) \in R} (\bar{\phi}_{\ell,t}-\bar{\gamma}_{\ell,t}) = 	\sum_{ \ell \in G \text{ s.t. } (\ell,t) \in R} \bar{\phi}_{\ell,t}-	\sum_{ \ell \in G \text{ s.t. } (\ell,t) \in R} \bar{\gamma}_{\ell,t}
			$$

			Hence, by the above, because for all $(\ell,t) \in R$ it holds that $\bar{\lambda}_{\ell,t} = \bar{\phi}_{\ell,t}$, and that by \eqref{F1} for all $(\ell,t) \notin R$ it holds that $\bar{\lambda}_{\ell,t} = 0$,  for all $t \in T$ it follows that:
			\begin{equation}
			\sum_{ \ell \in G} \bar{\lambda}_{\ell,t}-L_{t,G}	 \leq \sum_{ \ell \in G \text{ s.t. } (\ell,t) \in R} \bar{\gamma}_{\ell,t} \leq L_{t,G}- \sum_{ \ell \in G} \bar{\lambda}_{\ell,t}.
				\label{eq:TTT}
			\end{equation}
			
			Define $\bar{\lambda}^+,\bar{\lambda}^- \in \mathbb{R}^{I \times (I \cup \mathcal{C})}$ by $\bar{\lambda}^+_{\ell,t} = \bar{\lambda}^-_{\ell,t} = \bar{\lambda}_{\ell,t}$ for $(\ell,t) \in \left( I \times (I \cup \mathcal{C}) \right) \setminus R$, and $\bar{\lambda}^+_{\ell,t} = \bar{\lambda}_{\ell,t}+\bgam_{\ell,t},~\bar{\lambda}^-_{\ell,t} = \bar{\lambda}_{\ell,t}-\bgam_{\ell,t}$ for $(\ell,t)\in R$.  Since $\bgam\neq 0$ it follows that $\bar{\lambda}^+ \neq \bar{\lambda}^-$. 
			
			We now show that $\bar{\lambda}^+, \bar{\lambda}^-$ are in the $\bar{x}$-polytope. Let $\ell \in I , t \in I \cup \mathcal{C}$ such that  $\ell \notin \textsf{fit}(t)$. Therefore, $(\ell,t) \notin R$; thus, $\bar{\lambda}^+_{\ell,t} =  \bar{\lambda}_{\ell,t} = 0$ because $\bar{\lambda}$ is in the $\bar{x}$-polytope. We conclude that \eqref{F1} is satisfied for $\bar{\lambda}^+$. Let $C \in \mathcal{C}$. Then, 
			
			\begin{equation}
				\sum_{\ell \in I} \bar{\lambda}^+_{\ell,C} \cdot  s(\ell) =  \sum_{\ell \in G \text{ s.t. } (\ell,t) \in R} (\bar{\phi}_{\ell,C}+\bgam_{\ell,C}) \cdot s(\ell)+\sum_{\ell \in I \setminus G} \bar{\lambda}_{\ell,C} \cdot s(\ell)  \leq \left(1-s(C)\right) \cdot \bar{x}_C.
				\label{eq:+1}
			\end{equation}
			
		The inequality is because $\bar{\phi}+\bar{\gamma} \in Q$. By \eqref{eq:+1} we conclude that \eqref{F2} is satisfied for $\bar{\lambda}^+$.  Let $G' \in \mathcal{G}$ and $C \in \mathcal{C}$. Then, if $G' \neq G$ it holds that:

			\begin{equation}
				\sum_{\ell \in G'}  \bar{\lambda}^+_{\ell,C} =   \sum_{\ell \in G}  \bar{\lambda}_{\ell,C} \leq \bar{x}_C  \cdot (k(G)-|G \cap C|) =  L_{C,G}. 
				\label{aaeq:+2}
			\end{equation}
			
			The inequality is because $\bar{\lambda}$ is in the $\bar{x}$-polytope and satisfies \eqref{F5}. Otherwise, $G' = G$. 
			
				\begin{equation}
					\begin{aligned}
							\sum_{\ell \in G}  \bar{\lambda}^+_{\ell,C} ={} & 	\sum_{\ell \in G} \bar{\lambda}_{\ell,C}+ \sum_{\ell \in G} \bgam_{\ell,C} =   \sum_{\ell \in G} \bar{\lambda}_{\ell,C}+ \sum_{\ell \in G \text{ s.t. } (\ell,C) \in R} \bgam_{\ell,C}\\  \leq{} &
							 \sum_{\ell \in G} \bar{\lambda}_{\ell,C} + \left(L_{C,G}-\sum_{\ell \in G} \bar{\lambda}_{\ell,C} \right) = L_{C,G}= \bar{x}_C  \cdot (k(G)-|G \cap C|).
					\end{aligned}
				\label{aaeq:+2a}
			\end{equation}
			
			The inequality is by \eqref{eq:TTT}.  By \eqref{aaeq:+2} and \eqref{aaeq:+2a} we conclude that \eqref{F5} is satisfied for $\bar{\lambda}^+$.

			Let $j \in I$. We split into two cases.

			\begin{enumerate}
				\item $j \notin G$. Then, by \eqref{fitl}  and \eqref{F1} for all $\ell \notin \textsf{group}(j)$ it holds that $\bar{\lambda}^+_{\ell,j} = 0$ and for all $\ell' \in I$ it holds that  $\bar{\lambda}^+_{\ell',j} = \bar{\lambda}_{\ell',j}$. Therefore, 		\begin{equation*}
					\sum_{\ell \in I}  \bar{\lambda}^+_{\ell,j} =   \sum_{\ell \in \textsf{group}(j)}  \bar{\lambda}_{\ell,C}^+ =  \sum_{\ell \in \textsf{group}(j)}  \bar{\lambda}_{\ell,C} \leq  \sum_{C \in \cC[j]}  \bar{x}_C. 
				\end{equation*}

	\item $j \in G$. Then,  by \eqref{fitl}  and \eqref{F1} for all $\ell \notin G$ it holds that $\bar{\lambda}^+_{\ell,j} = 0$. Therefore,

		\begin{equation*}
		\begin{aligned}
				\sum_{\ell \in I}  \bar{\lambda}^+_{\ell,j} ={} & 	\sum_{\ell \in G}  \bar{\lambda}^+_{\ell,j} =   \sum_{\ell \in G} \bar{\lambda}_{\ell,j}+ \sum_{\ell \in G \text{ s.t. } (\ell,j) \in R} \bgam_{\ell,j} \\  \leq{} &
		 \sum_{\ell \in G} \bar{\lambda}_{\ell,j}+  \left(L_{j,G}- \sum_{\ell \in G} \bar{\lambda}_{\ell,j} \right) = L_{j,G} =  \sum_{C \in \cC[j]} \bar{x}_C. 
		\end{aligned}
	\end{equation*} The inequality is by \eqref{eq:TTT}.

			\end{enumerate}
				By the above, we conclude that \eqref{F3} is satisfied for $\bar{\lambda}^+$.

			Let $\ell \in I$. We split into two cases.

			\begin{enumerate}
				\item $\ell \notin G$. 
				$$\sum_{t \in I \cup \cC} \bar{\lambda}^+_{\ell,t} = \sum_{t \in I \cup \cC} \bar{\lambda}_{\ell,t} \geq 1.$$ The last inequality is because $ \bar{\lambda}$ satisfies constraint \eqref{F4} of the $\bar{x}$-polytope. 
				
				\item $\ell \in G$.  $$\sum_{t \in I \cup \cC} \bar{\lambda}^+_{\ell,t} \geq \sum_{t \in T} \left(\bar{\lambda}_{\ell,t}+ \bar{\gamma}_{\ell,t}\right) = \sum_{t \in T} \bar{\lambda}_{\ell,t}+\sum_{t \in T \text{ s.t. } (\ell,t) \in R} \bar{\gamma}_{\ell,t} = \sum_{t \in T} \bar{\lambda}_{\ell,t} \geq 1.$$ The second equality is by \eqref{eq:gamzero}. 
				
			\end{enumerate} By the above,  constraint \eqref{F4} is satisfied for $\bar{\lambda}^+$. Overall, we conclude that $\bar{\lambda}^+$ is in the $\bar{x}$-polytope.

			Using a symmetric argument it follows that $\bar{\lambda}^-$ is also in the $\bar{x}$-polytope.  Since $\bar{\lambda}= \frac{\bar{\lambda}^+ + \bar{\lambda}^-}{2}$, it follows that $\bar{\lambda}$ is not a vertex of the $\bar{x}$-polytope. Contradiction. Thus, $\bar{\phi}$ is a vertex of $Q$. 
		\end{proof}

		Let $F= \{\ell \in G~|~\exists t \in T \text{ s.t. } \bar{\phi}_{\ell,t} \in (0,1)\}$. By the definition of $\bar{\phi}$, it suffices to show that $|F| \leq 2|T|$.
		For any $\ell\in I$ define $\delta_{\ell} = |\{t \in T~|~(\ell,t) \in R\}|$ to be all of the corresponding pairs of $\ell$ in $R$. It follows that $|R| = \sum_{\ell\in I} \delta_{\ell}$. For $\ell \in F$, up to $\delta_{\ell}-1$ of the corresponding inequalities of $\ell$ in \eqref{eq:rho1} and \eqref{eq:rhoB} are tight: to satisfy \eqref{eq:rhoB} for $\ell$ at least two inequalities corresponding to $\ell$ in \eqref{eq:rho1} are not tight. In addition, for $\ell \in G \setminus F$, up to $\delta_{\ell}$ corresponding constraints in  \eqref{eq:rho1} and \eqref{eq:rhoB} for $\ell$ are tight: to satisfy \eqref{eq:rhoB} for $\ell$ at least one inequality corresponding to $\ell$ in \eqref{eq:rho1} is not tight. Moreover, there are at most $2|T|$ constraints in \eqref{eq:rhoD} and \eqref{eq:rhoS} that can be tight. Hence, the number of tight constraints is at most

		\begin{equation}
			\label{eq:tight}
			\sum_{\ell\in F} (\delta_{\ell} -1) + \sum_{\ell \in G \setminus F} \delta_\ell +2|T| \leq |R|-|F|+2|T|.
		\end{equation}

		As $\bar{\phi}$ is a vertex of $Q$, there are at least $|R|$ tight inequalities. Thus, $|F|\leq 2|T|$. Note that for all $t \in I \cup \cC \setminus T$ and $\ell \in I$ it holds that $\bar{\lambda}_{\ell,t} = 0$ by the constraints of the $\bar{x}$-polytope; thus, such entries cannot be fractional, and the number of fractional entries corresponding to  items in $G$ is at most $2|T|$. 
	\end{proof}

	\noindent{\bf Proof of Lemma~\ref{O(1)}:}
	
	Note that for all $t \in I \cup \cC \setminus T$ and $\ell \in I$ it holds that $\bar{\lambda}_{\ell,t} = 0$ by the constraints of the $\bar{x}$-polytope; thus, such entries cannot be non-integral. In addition, by Lemma~\ref{lem:FewFractionalGroups} and Lemma~\ref{lem:slot-types} the number of non-integral entries is at most $2|T| \cdot |T|$. Therefore, $$2|T| \cdot |T| \leq 2\left(|\textsf{supp}(\bar{x})|+k\cdot |\textsf{supp}(\bar{x})|\right)^2 \leq 2\left(2k \cdot |\textsf{supp}(\bar{x})|\right)^2  = 8 k^2 |\textsf{supp}(\bar{x})|^2.$$ The first inequality is because $|I \cap T|$ is upper bounded by the size of $\textsf{supp}(\bar{x})$ times the number of items in each configuration in $\textsf{supp}(\bar{x})$. \qed

\section{Omitted Proofs of Section~\ref{sec:eviction}}
\label{app:omittedEvict}

\noindent{\bf Proof of Lemma~\ref{lem:evictionHelp}:}
	Let $\ell_1, \ldots, \ell_{r}$ be the items in $C \setminus L$ in decreasing order such that for all $a,b \in [r], a<b$ it holds that $\ell_a < \ell_b$. In addition, let $U_C~=~\{\ell_1, \ldots, \ell_h\}$ such that $h = \alpha$. 	For the simplicity of the notations, for any $j \in [h]$ let $T_j =s \left( C \cap L \cup \{\ell_1, \ldots, \ell_{j-1}\} \right)$ be the total size of the large items in $C$ and the first $j-1$ small items in $C$ according to the above sorted order of items.\footnote{With a slight abuse of notation, assume that $\{\ell_1, \ldots, \ell_{j-1}\} = \emptyset$ for $j=1$.} We use the following auxiliary claims.

	\begin{claim}
		\label{claim:C1}
		 For all $i,k \in [h]$ which satisfy that $k>i+\eps^{-4}$, it holds that $s(\ell_i) > \frac{s(\ell_k)}{\eps^2}$.
	\end{claim}
	
	\begin{proof}
		 Assume towards a contradiction that there are $i,k \in [h]$ which satisfy that $k>i+\eps^{-4}$ and $s(\ell_i) \leq \frac{s(\ell_k)}{\eps^{2}}$. Thus, 
		
		\begin{equation}
			\label{eq:sik}
			s(\{\ell_i, \ell_{i+1}, \ldots, \ell_k\}) \geq (k-i) \cdot s(\ell_k) \geq \eps^{-4} \cdot s(\ell_k) \geq \eps^{-4}  \cdot \eps^{2} \cdot  s(\ell_i) =  \eps^{-2}  s(\ell_i) 
		\end{equation} The first inequality is because the items are in decreasing order that induces a non increasing order of item sizes. The second inequality is because $k>i+\eps^{-4}$. The third inequality is by the assumption that $s(\ell_i) \leq \frac{s(\ell_k)}{\eps^{2}}$. 
		
	 Note that it holds that $T_i<1$ since $T_i$ is the total size of items in a proper subset of $C$ and the sizes of the items are strictly larger than zero. Now, because $\ell_i \in U_C$ and $i < \alpha$, we conclude by \eqref{fitS}  and \eqref{h} that \begin{equation}
			\label{eq:llli}
			s(\ell_i) > \eps \left( 1-T_i  \right).
		\end{equation}  Therefore, 
		\begin{equation}
			\label{eq:contr}
		s(C) \geq T_i  +s(\{\ell_i, \ldots, \ell_k\}) \geq  T_i  +\eps^{-1} \left( 1-T_i  \right) \geq T_i  +\left( 1-T_i  \right)+(\eps^{-1}-1)  ( 1-T_i) > 1.
		\end{equation}
		The first inequality is because $C \cap L \cup \{\ell_1, \ldots, \ell_{i-1}\}$ is disjoint to $\{\ell_i, \ldots, \ell_k\}$. The second inequality is by \eqref{eq:sik} and \eqref{eq:llli}. The last inequality follows since $0<\eps<0.1$ and $0\leq T_i<1$. By \eqref{eq:contr}, we reach a contradiction to the fact that $s(C) \leq 1$ because $C$ is a configuration. 
	\end{proof}

		\begin{claim}
		\label{claim:C1*}
		For any $i \in [h]$ such that $i+\eps^{-4} < h-1$ it holds that $\ell_i \in \mathcal{L}_C$. 
	\end{claim}
	
	\begin{proof}
	\begin{equation}
		\label{eq:LC}
		s(\ell_i)> \frac{s(\ell_{h-1})}{\eps^{2}} > \eps^{-2} \cdot \eps (1-T_{h-1}) = \eps^{-1} (1-T_{h-1}) \geq \eps^{-1} (1-s(R(C))).
	\end{equation}
	
	The first inequality is by Claim~\ref{claim:C1}. The second inequality is because  $h-1 < h \leq \alpha$; thus, similarly to \eqref{eq:llli}, it holds that $s(\ell_{h-1}) > \eps \left( 1-T_{h-1}  \right)$. The last inequality is because $T_{h-1}$ is the total size in a proper subset of $R(C)$; thus, $T_{h-1} \leq s(R(C))$. By \eqref{eq:LC}, we conclude that for any $i \in [h]$ such that $i +\eps^{-4}< h-1$ it holds that $\ell_i \in \mathcal{L}_C$.
	\end{proof}

	 Now, using the above claim: 
	
	$$|\mathcal{L}_C| \geq h-2-\eps^{-4} \geq \eps^{-5}- 2\eps^{-4} = \eps^{-4} (\eps^{-1}-2)> \eps^{-4}.$$
	
	The first inequality is since by Claim~\ref{claim:C1*} it holds that $\ell_1, \ldots, \ell_{t} \in \mathcal{L}_C$ for $t = h-2-\eps^{-4}$. The second inequality is because $h =  \alpha$. \qed 
	
\noindent{\bf Proof of Lemma~\ref{lem:wProperties}:}
		Let $C \in \mathcal{C}$. We split the proof into the following claims which together form the proof of the lemma.

		\begin{claim}
			\label{claim:evic1}
			 Condition~1 of Lemma~\ref{lem:wProperties} holds for $|U_C| = \alpha$.
		\end{claim}
	
	\begin{proof}

		$$\|\bar{w}^C\| = \sum_{\ell \in \mathcal{L}_C} \bar{w}^C_{R(C) \setminus \{\ell\}} =   \sum_{\ell \in \mathcal{L}_C}  \frac{1}{|\mathcal{L}_C|-1}  = \frac{|\mathcal{L}_C|}{|\mathcal{L}_C|-1} \leq \frac{\eps^{-4}}{\eps^{-4}-1} \leq 1+2\eps^4. $$
		
		The first equality is because by the definition $\bar{w}^C$, for any other entry $C'$ that is not in the second summation it holds that $\bar{w}^C_{C'} = 0$. The second equality is by the definition of $\bar{w}^C$ in case that $|U_C| = \alpha$. The first inequality is because $|U_C| = \alpha$; therefore, by Lemma~\ref{lem:evictionHelp} it holds that $|\mathcal{L}_C| \geq \eps^{-4}$. The last inequality is because $0<\eps<0.1$. 
	\end{proof}

			\begin{claim}
			\label{claim:evic2}
				 Condition~2 of Lemma~\ref{lem:wProperties} holds for $|U_C| = \alpha$.
		\end{claim}
		
		\begin{proof}
			  Let $\ell \in R(C)$. Then, 
			
			$$\sum_{C' \in \mathcal{C}[\ell]} \bar{w}^C_{C'} \geq \sum_{j \in \mathcal{L}_C \setminus \{\ell\}} \bar{w}^C_{R(C) \setminus \{j\}} = \sum_{j \in \mathcal{L}_C \setminus \{\ell\}} \frac{1}{|\mathcal{L}_C|-1} \geq \frac{|\mathcal{L}_C|-1}{|\mathcal{L}_C|-1} = 1.$$ 
			
			The first inequality is since for any $j \in \mathcal{L}_C \setminus \{\ell\}$ it holds that  $R(C) \setminus \{j\} \in \mathcal{C}[\ell]$ and that all entries in $\bar{w}^C$ are non-negative. The first equality is by the definition of $\bar{w}^C$ in case that $|U_C| = \alpha$. The last inequality is because $|\mathcal{L}_C \setminus \{\ell\}| \geq |\mathcal{L}_C|-1$. 
		\end{proof}
		
			\begin{claim}
			\label{claim:evic3}
				 Condition~3 of Lemma~\ref{lem:wProperties} holds for $|U_C| = \alpha$.
		\end{claim}
		
		\begin{proof}
			  For the third condition of the lemma, let $C ' \in \textsf{supp}(\bar{w}^C)$. By the definition of $\bar{w}^C$ in case that $|U_C| = \alpha$, there is $\ell \in \mathcal{L}_C$ such that $C' = R(C) \setminus \{\ell\}$. Now,

			\begin{equation}
				\label{eq:1-RC}
				1-s(R(C)) = 1-s(R(C)\setminus \{\ell\})-s(\ell) \leq 1-s(R(C)\setminus \{\ell\}) -\frac{1-s(R(C))}{\eps}.
			\end{equation}
			The first equality is because $\ell \in \mathcal{L}_C$ and $\mathcal{L}_C \subseteq R(C)$. The first inequality is by the definition of $\mathcal{L}_C$. Therefore, 
			
			\begin{equation}
				\label{eq:1-RC2}
				1-s(R(C)) \leq \left(\frac{1}{1+\frac{1}{\eps}}\right)(1-s(R(C)\setminus \{\ell\})) \leq \eps(1-s(R(C)\setminus \{\ell\})).
			\end{equation}
			The first inequality is by \eqref{eq:1-RC}. Let $j \in C \setminus R(C)$. Now, 
			
			\begin{equation}
				\label{eq:fit1}
				s(j) \leq 1-s(R(C)) \leq  \eps(1-s(R(C)\setminus \{\ell\})).
			\end{equation}
			
			The first inequality is because $j \notin R(C)$ and $s(C) \leq 1$. The second inequality is by \eqref{eq:1-RC2}. Therefore, by \eqref{eq:fit1} we get that $j \in  \textsf{fit}(C')$ by \eqref{fitS} (because $j \in C \setminus R(C)$ it holds that $j \notin L$). Because there are no restrictions over $j$ besides that $j \in C \setminus R(C)$, we conclude that $C \setminus R(C) \subseteq \textsf{fit}(C')$. In addition,

			$$s(C \setminus R(C)) \leq 1-s(R(C)) \leq 1-s(R(C) \setminus \{\ell\}) = 1-s(C').$$
			
			The first inequality is because $R(C) \subseteq C$ and $s(C) \leq 1$. The second  inequality is because $s(\ell)\geq 0$. The first equality is by the definition of $C'$. 
		\end{proof}
		
			\begin{claim}
			\label{claim:evic4}
				 Condition~4 of Lemma~\ref{lem:wProperties} holds.
		\end{claim}
		
		\begin{proof}
			let $C ' \in \textsf{supp}(\bar{w}^C)$. Assume towards a contradiction that $s(C' \setminus L)>\eps$. Therefore, 
			\begin{equation}
				\label{eq:w4}
				\eps<s(C' \setminus L) \leq 1-s(C' \cap L).
			\end{equation}
			The second inequality is because $C' \in \cC$ is a configuration. Recall $\ell_1$ from \eqref{h} for $C'$; it also holds that 
			\begin{equation}
				\label{eq:w4b}
				s(\ell_1) \leq \eps^2 = \eps \cdot \eps < \eps \cdot \left(1-s(C' \cap L)\right).
			\end{equation}
			The first inequality is because $\ell_1 \in I \setminus L$. The last inequality is by \eqref{eq:w4}. Therefore, by~ \eqref{eq:w4b} we get that $\ell_1 \in  \textsf{fit}(C' \cap L)$ by \eqref{fitS}. By the definition of the relaxation $\bar{w}^{C}$ it holds that $C'\cap L=C\cap L$, thus $\ell_1 \in \textsf{fit}(C\cap L)$. Then, by \eqref{h} and the definition of $U_C$ we conclude that $U_C  = \emptyset$. Hence,  by the definition of $\bar{w}^C$ and that $C' \in \textsf{supp}(\bar{w}^C)$, it holds that $C' = C \cap L$; thus, $s(C' \setminus L) = 0$. For any $\eps>0$, this is a contradiction that $s(C' \setminus L) > \eps$. 
		\end{proof}

			\begin{claim}
			\label{claim:evic6}
			Condition~5 of Lemma~\ref{lem:wProperties} holds.
		\end{claim}
		
		\begin{proof}
			Let $C' \in \textsf{supp}(\bar{w}^C)$ and $G \in \cG$. Now,  $$\left|G \cap \left(C \setminus R(C)\right)\right| \leq |G \cap C|-|G \cap R(C)| \leq  k(G) -|G \cap R(C)| \leq k(G) - |G \cap C'|.$$ The last inequality is because by the definition of $\bar{w}^C$ it holds that $C' \subseteq R(C)$
 		\end{proof}

			\begin{claim}
			\label{claim:evic5}
			Conditions 1-3 of Lemma~\ref{lem:wProperties} hold for $|U_C| < \alpha$. 
		\end{claim}
		
		\begin{proof}
			For the first condition of the lemma: $$\|\bar{w}^C\| = \bar{w}^C_{R(C)} = 1 \leq 1+2\eps^4. $$ The first equality is because by the definition $\bar{w}^C$, for any other configuration $C' \in \mathcal{C} \setminus \{R(C)\}$ it holds that $\bar{w}^C_{C'} = 0$. The second equality is by the definition of $\bar{w}^C$ in case that $|U_C| < \alpha$.  
			
			For the second condition of the lemma, let $\ell \in R(C)$. Then, $$\sum_{C' \in \mathcal{C}[\ell]} \bar{w}^C_{C'} = \bar{w}^C_{R(C)} = 1.$$ The first equality is because by the definition $\bar{w}^C$, for any other configuration $C' \in \mathcal{C} \setminus \{R(C)\}$ it holds that $\bar{w}^C_{C'} = 0$. The second equality is bythe definition of $\bar{w}^C$ in case that $|U_C| < \alpha$.
			
			For the third condition of the lemma, let $C ' \in \textsf{supp}(\bar{w}^C)$. By the definition of $\bar{w}^C$ in case that $|U_C| < \alpha$, it holds that $C' = R(C)$. Let $j \in C \setminus R(C)$. It holds that $j \in \textsf{fit}(C')$ by \eqref{h}, since $|U_C|<\alpha$. Because there are no restrictions over $j$ besides that $j \in C \setminus R(C)$, we conclude that $C \setminus R(C) \subseteq \textsf{fit}(C')$. In addition, 
			
			$$s(C \setminus R(C)) \leq 1-s(R(C)) = 1-s(C'). $$ The first inequality is because $R(C) \subseteq C$ and $s(C) \leq 1$. The first equality is by the definition of $C'$. 
		\end{proof} The proof of Lemma~\ref{lem:wProperties} follows by Claim~\ref{claim:evic1}, Claim~\ref{claim:evic2}, Claim~\ref{claim:evic3}, Claim~\ref{claim:evic4}, Claim~\ref{claim:evic5}, and Claim~\ref{claim:evic6}.

	\noindent{\bf Proof of Lemma~\ref{lem:eviction}:}
We use in the proof of Lemma~\ref{lem:eviction} the following claims. 

\begin{claim}
	\label{lem:evictionPolynomial}
	Algorithm~\ref{Alg:eviction} is polynomial.
\end{claim}

\begin{proof}
	The number of elements over which the loop in Step~\ref{step:forsupp} iterates is polynomial, since  $|\textsf{supp}(\bar{x})|$ is at most the encoding size of the input (in a sparse representation). Let $C \in \textsf{supp}(\bar{x})$ be a configuration in the support of $\bar{x}$. It follows by \eqref{h} that $U_C$ can be easily computed in polynomial time by iteratively adding items from $C \setminus L$ in decreasing order of the items. Therefore, given $U_C$, the relaxation of $C$ that is $\bar{w}^C$ can be computed in polynomial time: by the definition of $\bar{w}^C$, the number of nonzero entries in $\bar{w}^C$ is at most $|\mathcal{L}_C|$, which is polynomial in the size of the instance. This is since $\mathcal{L}_C \subseteq I$ is a subset of the items. Finally, Step~\ref{step:return} is also polynomial as $\bar{y}$ is a linear combination of a polynomial number (i.e., $|\textsf{supp}(\bar{x})|$) of vectors, each with a polynomial number of nonzero entries. By the above, the claim follows. 
\end{proof}

	\begin{claim}
		\label{lem:eviction1}
		For any $0<\eps<0.1$, an $\eps$-structured BPP instance $\mathcal{I}$, and a solution $\bar{x}$ to the configuration LP in \eqref{C-LP}, Algorithm~\ref{Alg:eviction} returns a prototype $\bar{y}$  of $\mathcal{I}$ with $\bar{\lambda}$ in the $\bar{y}$-polytope such that for all $\ell,j \in I, \ell \neq j$ it holds that $\bar{\lambda}_{\ell,j} = 0$. 
	\end{claim}
	
	\begin{proof}
		We define below a point $\bar{\lambda} \in \mathbb{R}^{I \times (I \cup \mathcal{C})}_{\geq 0}$ in the $\bar{y}$-polytope. Thus, it follows that the $\bar{y}$-polytope is non-empty. For any $C \in \mathcal{C}$, define the {\em point of $C$} as a vector $\bar{\gamma}^C \in \mathbb{R}^{I \times (I \cup \mathcal{C})}_{\geq 0}$ as follows. For any $\ell \in R(C)$, define 
		
		\begin{equation}
			\label{gammal}
			\bar{\gamma}^C_{\ell,\ell} = 1.
		\end{equation}
	
	Moreover, for any $\ell \in C \setminus R(C)$ and $C' \in \mathcal{C}$ define 
		\begin{equation}
		\label{gammaC}
		\bar{\gamma}^C_{\ell,C'} = \bar{w}^C_{C'}.
	\end{equation}

For any other $\ell \in I$ and $t \in (I \cup \mathcal{C})$ define

	\begin{equation}
		\label{gammae}
	\bar{\gamma}^C_{\ell,t} = 0.
\end{equation}

	Finally, define 

\begin{equation}
	\label{lambda}
	\bar{\lambda} = \sum_{C \in \mathcal{C}} \bar{x}_C \cdot \bar{\gamma}^C.
\end{equation}

We show below that all constraints of the $\bar{y}$-polytope are satisfied for $\bar{\lambda}$. Let  $\ell \in I$  and $t \in I \cup \mathcal{C}$ such that  $\ell \notin \textsf{fit}\left(t \right)$. Then, 

\begin{equation}
	\label{eq:ev1}
	\bar{\lambda}_{\ell,t} = \sum_{C \in \mathcal{C}} \bar{y}_C \cdot \bar{\gamma}^C_{\ell,t} = \sum_{C \in \mathcal{C}} \bar{y}_C \cdot 0 = 0.
\end{equation}

The first equality is by \eqref{lambda}. Because $\ell \notin \textsf{fit}\left(t \right)$, then $\ell \neq t$ by \eqref{fitl}; moreover, for all $C \in \mathcal{C}$ such that $\ell \in C \setminus R(C)$ and $C ' \in \textsf{supp}(\bar{w}^C)$,  it holds that $t \neq C'$ because $\ell \in \textsf{fit}(C')$ by Lemma~\ref{lem:wProperties}. Therefore, by \eqref{gammae} the second equality follows. By \eqref{eq:ev1} we conclude that constraint \eqref{F1} is satisfied for $\bar{\lambda}$. 

Let $C \in \mathcal{C}$. Then,

\begin{equation}
		\label{eq:ev2}
	\begin{aligned}
	\sum_{\ell \in I} \bar{\lambda}_{\ell,C} \cdot  s(\ell) &= \sum_{\ell \in I~} \sum_{C' \in \mathcal{C}} \bar{x}_{C'} \cdot \bar{\gamma}^{C'}_{\ell,C} \cdot  s(\ell) =  \sum_{C' \in \mathcal{C}~} \sum_{\ell \in I} \bar{x}_{C'} \cdot \bar{\gamma}^{C'}_{\ell,C} \cdot  s(\ell)\\
	={} & \sum_{C' \in \mathcal{C}~} \sum_{\ell \in C' \setminus R(C')} \bar{x}_{C'} \cdot \bar{w}^{C'}_{C} \cdot  s(\ell) 
	=   \sum_{C' \in \mathcal{C}}  \bar{x}_{C'} \cdot \bar{w}^{C'}_{C} \cdot s\left(C' \setminus R(C')\right)\\
	\leq{} & \sum_{C' \in \mathcal{C}}  \bar{x}_{C'} \cdot \bar{w}^{C'}_{C} \cdot \left(1-s(C)\right)   =(1-s(C)) \cdot  \sum_{C' \in \textsf{supp}(\bar{x})}  \bar{x}_{C'} \cdot \bar{w}^{C'}_{C} = (1-s(C)) \bar{y}_{C}. 
	\end{aligned}
\end{equation}

The first equality is by \eqref{lambda}. The second equality is by changing the order of summation. The third equality is by \eqref{gammal}, \eqref{gammaC} and \eqref{gammae}.  The inequality is because for all $C' \in \mathcal{C}$, if $C \in \textsf{supp}({\bar{w}^{C'}})$ then  $s(C' \setminus R(C')) \leq 1-s(C)$ by Lemma~\ref{lem:wProperties}. The last equality is by Step~\ref{step:return} of Algorithm~\ref{Alg:eviction}. Therefore, by \eqref{eq:ev2} we conclude that constraint \eqref{F2} is satisfied for $\bar{\lambda}$.

Let $G \in \mathcal{G}$ and $C \in \mathcal{C}$. Then,

\begin{equation}
	\label{eq:ev3}
	\begin{aligned}
	 \sum_{\ell \in G} \bar{\lambda}_{\ell,C}  ={} &   \sum_{\ell \in G~}  \sum_{C' \in \mathcal{C}} \bar{x}_{C'} \cdot \bar{\gamma}^{C'}_{\ell,C}  = \sum_{C' \in \mathcal{C}~}  \sum_{\ell \in G}  \bar{x}_{C'} \cdot \bar{\gamma}^{C'}_{\ell,C} = \sum_{C' \in \mathcal{C}~~}  \sum_{\ell \in G \cap \left(C' \setminus R(C') \right)}  \bar{x}_{C'} \cdot \bar{w}^{C'}_{C}\\
	   ={} & \sum_{C' \in \mathcal{C} \text{ s.t. } C \in \textsf{supp}(\bar{w}^{C'})~~}   \bar{x}_{C'} \cdot \bar{w}^{C'}_{C}  \sum_{\ell \in G \cap \left(C' \setminus R(C') \right)} 1 
	   \leq  \sum_{C' \in \mathcal{C}} \bar{x}_{C'} \cdot \bar{w}^{C'}_{C} \cdot \left(k(G)-|G \cap C|\right) \\ ={} & \left(k(G)-|G \cap C|\right) \sum_{C' \in \textsf{supp}(\bar{x})~}  \bar{x}_{C'} \cdot \bar{w}^{C'}_{C}  =\bar{y}_C \cdot \left(k(G)-|G \cap C|\right)
	\end{aligned}
\end{equation}

The first equality is by \eqref{lambda}. The second equality is by changing the order of summation. The third equality is by \eqref{gammal}, \eqref{gammaC} and \eqref{gammae}.  The inequality is by Condition~5 of Lemma~\ref{lem:wProperties}. The last equality is by Step~\ref{step:return} of Algorithm~\ref{Alg:eviction}. Therefore, by \eqref{eq:ev3} we conclude that constraint \eqref{F5} is satisfied for $\bar{\lambda}$.

Let $j \in I$. Then, 

\begin{equation}
	\label{eq:ev4}
	\begin{aligned}
	 \sum_{\ell \in I} \bar{\lambda}_{\ell,j} ={} &  \sum_{\ell \in I~}  \sum_{C \in \mathcal{C}} \bar{x}_{C} \cdot \bar{\gamma}^{C}_{\ell,j} =   \sum_{C \in \mathcal{C}~} \sum_{\ell \in I} \bar{x}_{C} \cdot \bar{\gamma}^{C}_{\ell,j} =  \sum_{C \in \mathcal{C}[j] \text{ s.t. } j \in R(C)} \bar{x}_{C}\\ \leq{} &  \sum_{C \in \mathcal{C}[j] \text{ s.t. } j \in R(C)} \bar{x}_{C} \cdot \left(  \sum_{C' \in \mathcal{C}[j]} \bar{w}^{C}_{C'} \right) 
	\leq  \sum_{C \in \mathcal{C}[j]} \sum_{C' \in \mathcal{C}[j]}  \bar{x}_{C} \cdot  \bar{w}^{C}_{C'}\\   ={} &   \sum_{C' \in \mathcal{C}[j]} \sum_{C \in \mathcal{C}[j]}  \bar{x}_{C} \cdot  \bar{w}^{C}_{C'}
	 \leq \sum_{C' \in \mathcal{C}[j]} \sum_{C \in \textsf{supp}(\bar{x})}  \bar{x}_{C} \cdot  \bar{w}^{C}_{C'}  =  \sum_{\substack{C' \in \mathcal{C}[j]}}  \bar{y}_{C'} = \sum_{\substack{C \in \mathcal{C}[j]}}  \bar{y}_{C}.
	\end{aligned}
\end{equation}

The first equality is by \eqref{lambda}. The second equality is by changing the order of summation. The third equality is by \eqref{gammal}, \eqref{gammaC} and \eqref{gammae}.  The first  inequality is because for all $C \in \cC$ and $j \in R(C)$ it holds that  $ \sum_{C' \in \mathcal{C}[j]} \bar{w}^{C}_{C'}  \geq 1$ by Lemma~\ref{lem:wProperties}. The second inequality is because for all $C,C' \in \cC$ it holds that $ \bar{x}_{C} \cdot  \bar{w}^{C}_{C'} \geq 0$.  The fourth equality is is by changing the order of summation. The third inequality is because for all $C \in \mathcal{C}[j] \setminus \textsf{supp}(\bar{x})$ it holds that $\bar{x}_{C} \cdot  \bar{w}^{C}_{C'} = 0$. The fifth equality is by Step~\ref{step:return} of Algorithm~\ref{Alg:eviction}. The sixth equality is by a change in the notation of a variable. Therefore, by \eqref{eq:ev4} we conclude that constraint \eqref{F3} is satisfied for $\bar{\lambda}$. 

Let $\ell \in I$. Then, we use the following equations. 

\begin{equation}
	\label{eq:ev5a}
	\sum_{t \in I} \bar{\lambda}_{\ell,t} =   \sum_{t \in I~} \sum_{C \in \mathcal{C}} \bar{x}_{C} \cdot \bar{\gamma}^{C}_{\ell,t}  =  \sum_{C \in \mathcal{C}~} \sum_{t \in I}  \bar{x}_{C} \cdot \bar{\gamma}^{C}_{\ell,t} =  \sum_{C \in \mathcal{C}[\ell] \text{ s.t. } \ell \in R(C)} \bar{x}_{C}.
\end{equation}

The first equality is by \eqref{lambda}. The second equality is by changing the order of summation. The third equality is by \eqref{gammal} and \eqref{gammae}. In addition,

\begin{equation}
	\label{eq:ev5b}
	\begin{aligned}
		\sum_{t \in \mathcal{C}} \bar{\lambda}_{\ell,t} ={} &   \sum_{t \in \mathcal{C}~} \sum_{C \in \mathcal{C}} \bar{x}_{C} \cdot \bar{\gamma}^{C}_{\ell,t}  =  \sum_{C \in \mathcal{C}~} \sum_{t \in \mathcal{C}}  \bar{x}_{C} \cdot \bar{\gamma}^{C}_{\ell,t} =  \sum_{C \in \mathcal{C}[\ell] \text{ s.t. } \ell \in C \setminus R(C)~} \sum_{t \in \mathcal{C}} \bar{x}_{C} \cdot \bar{w}^{C}_{t}\\   ={} &  \sum_{C \in \mathcal{C}[\ell] \text{ s.t. } \ell \in C \setminus R(C)~} \bar{x}_{C} \sum_{t \in \mathcal{C}}   \bar{w}^{C}_{t} \geq \sum_{C \in \mathcal{C}[\ell] \text{ s.t. } \ell \in C \setminus R(C)~} \bar{x}_{C} \sum_{t \in \mathcal{C}[\ell]}   \bar{w}^{C}_{t} \geq   \sum_{C \in \mathcal{C}[\ell] \text{ s.t. } \ell \in C \setminus R(C)~} \bar{x}_{C}.
	\end{aligned}
\end{equation}

The first equality is by \eqref{lambda}. The second equality is by changing the order of summation. The third equality is by \eqref{gammaC} and \eqref{gammae}. The first inequality is because $\mathcal{C}[\ell] \subseteq \mathcal{C}$. The second inequality is because $ \sum_{t \in \mathcal{C}[\ell]}   \bar{w}^{C}_{t} \geq 1$ by Lemma~\ref{lem:wProperties}. Thus,

\begin{equation}
	\label{eq:ev5d}
\sum_{t \in I \cup \mathcal{C}} \bar{\lambda}_{\ell,t} = 	\sum_{t \in I} \bar{\lambda}_{\ell,t}+ 	\sum_{t \in \mathcal{C}} \bar{\lambda}_{\ell,t}  \geq \sum_{C \in \mathcal{C}[\ell] \text{ s.t. } \ell \in R(C)} \bar{x}_{C} +\sum_{C \in \mathcal{C}[\ell] \text{ s.t. } \ell \in C \setminus R(C)~} \bar{x}_{C} = \sum_{C \in \mathcal{C}[\ell]} \bar{x}_{C} \geq 1.  
\end{equation}

The first equality is because $I \cap \mathcal{C} = \emptyset$. The first inequality is by \eqref{eq:ev5a} and \eqref{eq:ev5b}. The second equality is because for all $C \in \mathcal{C}[\ell]$ either $\ell \in R(C)$ or $\ell \in C\setminus R(C)$. The last inequality is because $\bar{x}$ is a solution for \eqref{C-LP}. Therefore, by  \eqref{eq:ev5d}, constraint \eqref{F4} is satisfied for $\bar{\lambda}$. In summary, we conclude that $\bar{\lambda}$ is a point in the $\bar{y}$-polytope; hence, the $\bar{y}$-polytope is non-empty. 

	\end{proof}

			\begin{claim}
		\label{lem:eviction2}
		For any $0<\eps<0.1$, an $\eps$-structured BPP instance $\mathcal{I}$, and a prototype $\bar{x}$ of $\mathcal{I}$, Algorithm~\ref{Alg:eviction} returns a prototype $\bar{y}$ such that $||\bar{y}|| \leq (1+\eps)||\bar{x}||$.
	\end{claim}

			\begin{proof}

				\begin{align*}
				\|\bar{y}\| ={} & \left\|  \sum_{C \in \textsf{supp}(\bar{x})} \bar{x}_{C} \cdot \bar{w}^{C} \right\| \leq  \sum_{C \in \textsf{supp}(\bar{x})} \bar{x}_{C} \cdot \| \bar{w}^{C} \|  \leq  \sum_{C \in \textsf{supp}(\bar{x})} \bar{x}_{C} \cdot (1+2\eps^4) = (1+2\eps^4)\|\bar{x}\| \leq (1+\eps)\|\bar{x}\|.
			\end{align*}
				
		The first equality is by Step~\ref{step:return} of Algorithm~\ref{Alg:eviction}. The first inequality is by the triangle inequality. The second inequality is because $\|\bar{w}^{C}\|  \leq (1+2\eps^4)$ by Lemma~\ref{lem:wProperties} for all $C \in \textsf{supp}(\bar{x})$. The last inequality is because $0<\eps< 0.1$.

			\end{proof}
			
			\begin{claim}
				\label{lem:eviction3}
				For any $0<\eps<0.1$, an $\eps$-structured BPP instance $\mathcal{I}$, and a prototype $\bar{x}$ of $\mathcal{I}$,  Algorithm~\ref{Alg:eviction} returns a prototype $\bar{y}$ such that for all $C \in \textsf{supp}(y)$ it holds that $|C| \leq \eps^{-10}$.
			\end{claim}
			
			\begin{proof}
				Let $C \in \textsf{supp}(\bar{y})$ be a configuration. Because $C \in \textsf{supp}(\bar{y})$, then by Step~\ref{step:return} of Algorithm~\ref{Alg:eviction}, there is $C' \in \textsf{supp}(\bar{x})$ such that $\bar{w}^{C'}_C \neq 0$.  Therefore,  $$|C| \leq |R(C')| = |C' \cap L \cup U_{C'}| \leq |C' \cap L|+|U_{C'}| \leq \eps^{-2}+\alpha \leq 2\eps^{-6} \leq \eps^{-10}.$$ 
				
				The first inequality is because the maximal number of items in $C$ is bounded by $|R(C')|$ by the definition of the relaxation vector of a configuration. The first equality is by the definition of $R(C')$. The second inequality is by the union bound. The third inequality is since the size of each large item in $L$ is at least $\eps^2$ and there can be at most $\eps^{-2}$ such items in a configuration without exceeding the maximal total size of items in a configuration which is $1$; moreover, by \eqref{h} it holds that $|U_C| \leq \alpha$. The fourth inequality is because $\alpha = \lceil \eps^{-5} \rceil$. The last inequality is because $0<\eps<0.1$. 
			\end{proof}
			
				\begin{claim}
				\label{lem:eviction4}
				For any $0<\eps<0.1$, an $\eps$-structured BPP instance $\mathcal{I}$, and a prototype $\bar{x}$ of $\mathcal{I}$,  Algorithm~\ref{Alg:eviction} returns a prototype $\bar{y}$ such that for all $C \in \textsf{supp}(y)$ it holds that $s(C \setminus L) \leq \eps$. 
			\end{claim}
			
			\begin{proof}
				Let $C \in \textsf{supp}(\bar{y})$ be a configuration. Because $C \in \textsf{supp}(\bar{y})$, then by Step~\ref{step:return} of Algorithm~\ref{Alg:eviction}, there is $C' \in \textsf{supp}(\bar{x})$ such that $\bar{w}^{C'}_C \neq 0$.  Therefore,  $C \in \textsf{supp}(\bar{w}^{C'})$ and it follows that $s(C \setminus L) \leq \eps$ by Lemma~\ref{lem:wProperties}. 
			\end{proof}

				\begin{claim}
				\label{lem:evictionF4}
			for every $\ell \in I$ it holds that $\sum_{C\in \cC[\ell]} \by_C\leq 2$. 
			\end{claim}
			
			\begin{proof}
				\begin{align*}
		\sum_{C\in \cC[\ell]} \by_C ={} & \sum_{C\in \cC[\ell]~} \sum_{C' \in \textsf{supp}(\bar{x})} \bar{x}_{C'} \cdot \bar{w}^{C'}_C =  \sum_{C' \in \textsf{supp}(\bar{x})~} \sum_{C\in \cC[\ell]} \bar{x}_{C'} \cdot \bar{w}^{C'}_C \leq  \sum_{C' \in \textsf{supp}(\bar{x})~} \bar{x}_{C'} \cdot \sum_{C\in \cC}   \bar{w}^{C'}_C \\ ={} &\sum_{C' \in \textsf{supp}(\bar{x})~} \bar{x}_{C'} \cdot \|\bar{w}^{C'}\| \leq \sum_{C' \in \textsf{supp}(\bar{x})~} \bar{x}_{C'} \cdot (1+2\eps^{4}) \leq \sum_{C' \in \textsf{supp}(\bar{x})~} \bar{x}_{C'} \cdot 2 = 2 \cdot \|\bar{x}\| = 2. 
			\end{align*}
			The first equality is by Step~\ref{step:return} of Algorithm~\ref{Alg:eviction}. The second inequality is because $\|\bar{w}^{C'}\|  \leq (1+2\eps^4)$ by Lemma~\ref{lem:wProperties}. The third inequality is because $0<\eps< 0.1$. The last equality is by \eqref{LP}. 
			\end{proof}

			The proof of Lemma~\ref{lem:eviction} follows by Claim~\ref{lem:evictionPolynomial}, Claim~\ref{lem:eviction1}, Claim~\ref{lem:eviction2}, Claim~\ref{lem:eviction3}, Claim~\ref{lem:eviction4}, and Claim~\ref{lem:evictionF4}. \qed

\section{Omitted Proofs of Section~\ref{sec:AlgPrototype}}
\label{app:omittedShifting}

\noindent{\bf Proof of Lemma~\ref{lem:significantDiscard}:}
We prove the two properties of the lemma below.

\begin{enumerate}
	\item 	\begin{equation*}
		\label{eq:ev5b}
		\begin{aligned}
			\sum_{\ell \in I \setminus L~} f_{\bar{y}}(\ell) \cdot s(\ell) ={} &   \sum_{\ell \in I \setminus L~}  \sum_{C \in \mathcal{C}[\ell]} \bar{y}_C \cdot s(\ell)  =   \sum_{C \in \mathcal{C}~}  \sum_{\ell \in C \setminus L} \bar{y}_C \cdot s(\ell) =   \sum_{C \in \mathcal{C} } \bar{y}_C  \sum_{\ell \in C \setminus L}  s(\ell)\\   ={} &  \sum_{C \in \mathcal{C} } \bar{y}_C  \cdot s(C \setminus L) \leq \sum_{C \in \mathcal{C}} \bar{y}_C  \cdot \eps  = \eps  \sum_{C \in \mathcal{C}} \bar{y}_C = \eps \|\bar{y}\|.
		\end{aligned}
	\end{equation*} The first equality is by the definition of frequency. The second equality is by changing the order of summation. The inequality follows since $\bar{y}$ holds the conditions of Lemma~\ref{lem:Cnf}. 
	
	\item Assume towards a contradiction that there is $G \in \mathcal{G} \setminus  \mathcal{G}_{\eta}$ such that $f_{\bar{y}}(G \setminus L) > \eps \cdot \|\bar{y}\|$. Therefore, by the definition of $\cG_{\eta}$ as the set of groups with maximal frequncies of small items, it holds that
	\begin{equation}
		\label{eq:G'cap}
		f_{\bar{y}}(G_{\eta} \setminus L)  > \eps \cdot \|\bar{y}\|.
	\end{equation} We reach a contradiction by the following. 
	
	\begin{equation}
		\label{eq:f1}
		\sum_{\ell \in I} f_{\bar{y}}(\ell) = \sum_{\ell \in I ~}  \sum_{C \in \mathcal{C}[\ell]} \bar{y}_C =  \sum_{C \in \mathcal{C}~}  \sum_{\ell \in C} \bar{y}_C \leq  \sum_{C \in \mathcal{C}~}  \eps^{-10} \cdot \bar{y}_C =  \eps^{-10} \|\bar{y}\|.
	\end{equation} The first equality is by the definition of the frequency. The second equality is by changing the order of summation. The inequality follows since $\bar{y}$ holds the conditions of Lemma~\ref{lem:Cnf}; thus, for all $C \in \textsf{supp}(\bar{y})$ it holds that $|C| \leq \eps^{-10}$. In addition, 
	
	\begin{equation}
		\label{eq:f2}
		\sum_{\ell \in I} f_{\bar{y}}(\ell) \geq \sum_{i \in [\eta]~} 	\sum_{\ell \in G_i \setminus L} f_{\bar{y}}(\ell) \geq \eta \cdot f_{\bar{y}}(G_{\eta} \setminus L) > \eta \cdot  \eps \cdot \|\bar{y}\| >  \eps^{-10} \|\bar{y}\|.
	\end{equation} The first inequality is because for all $\ell \in I$ it holds that $f_{\bar{y}}(\ell) \geq 0$. The second inequality is because for all $G' \in \cG_{\eta}$ it holds that $f_{\bar{y}}(G' \setminus L) \geq f_{\bar{y}}(G_{\eta} \setminus L)$. The third inequality is by \eqref{eq:G'cap}. The last inequality is since $\eta =  \eps^{-12} $ or $\cG_{\eta}  = \cG$. By \eqref{eq:f1} and \eqref{eq:f2} we reach a contradiction. \qed
\end{enumerate}

We use the following auxiliary claim. \begin{claim}
	\label{AETA}
	$|\cA \cup \cG_{\eta}| \leq 2K(\eps)$. 
\end{claim}

\begin{proof}
	$$|\cA \cup \cG_{\eta}| \leq |\cA|+|\cG_{\eta}| \leq K(\eps)+\eta \leq 2K(\eps).$$
	The second inequality is because $\mathcal{I}$ is $\eps$-structured, there can be at most $K(\eps)$ massive groups by \eqref{Agroups} and Definition~\ref{def:structuring}. The third inequality is because $K(\eps) = \eps^{-\eps^{-2}}$, $\eta \leq \eps^{-12}$, and $\eps < 0.1$. 
\end{proof}

\noindent{\bf Proof of Lemma~\ref{lem:P(C)}:}
Let $C \in \textsf{supp}(\bar{y})$. We show the properties of the lemma below.

\begin{enumerate}
	\item

		\begin{equation}
		\begin{aligned}
			\label{pcHelpL}
			s\left(P(C)\right) \leq \sum_{\Phi \in \mathcal{Q}} s\left(\textsf{last}_{|C \cap \Phi|}(\Phi)\right) \leq {} &   \sum_{\Phi \in \mathcal{Q}} s\left(C \cap \Phi\right) \leq s(C). 
		\end{aligned}
	\end{equation}

The second inequality is because for all $\Phi \in \mathcal{Q}$ it holds that $\textsf{last}_{|C \cap \Phi|}(\Phi)$ contains the $|C \cap \Phi|$ items in $\Phi$ with minimal size. The last inequality is because each item belongs to at most one class in $\mathcal{Q}$.

	\item Let $\ell \in \textsf{fit}(C)$. First, by \eqref{fitS} it holds that $s(\ell) \leq \eps^{2}$. Second, we use the following inequalities.	\begin{equation}
		\label{Ffit}
	s(\ell) \leq \eps \cdot (1-s(C)) \leq \eps \cdot (1-s(P(C))).
	\end{equation} The first inequality is because $\ell \in \textsf{fit}(C)$. The second inequality is by \eqref{pcHelpL}. Therefore, by \eqref{Ffit} and that $s(\ell) \leq \eps^{2}$ we conclude that $\ell \in \textsf{fit}(P(C))$ by \eqref{fitS}. 

	\item 
	
	$$|P(C)| = \left|\bigcup_{\Phi \in \mathcal{Q}} \textsf{last}_{|C \cap \Phi|}(\Phi)\right| \leq  \sum_{\Phi \in \mathcal{Q}} |\textsf{last}_{|C \cap \Phi|}(\Phi)| = \sum_{\Phi \in \mathcal{Q}} |C \cap \Phi| \leq |C| \leq \eps^{-10}.$$ The first equality is by \eqref{eq:mapS,eq:mapM}. The second equality is by the definition of \textsf{last}. The second inequality is because each item in $C$ belongs to at most one class in $\mathcal{Q}$. The last inequality is because $C \in \textsf{supp}(\bar{y})$; thus, since $\bar{y}$ satisfies the conditions of Lemma~\ref{lem:Cnf} it holds that $|C| \leq \eps^{-10}$.

	\item First, there can be at most $\eps^{-\upsilon-10}$ classes for each important group $G \in \cG_{\eta} \cup \cA$ by Observation~\ref{ob:FG}. Thus, the number of classes is bounded by $|\cG_{\eta} \cup \cA| \cdot \eps^{-\upsilon-10}$. Second, by Claim~\ref{AETA} we have a bound of $2K(\eps) $ on the number  for  important groups. By the above,

		\begin{equation}
		\begin{aligned}
			\label{bound1111}
			|\mathcal{Q}| \leq{} & |\cG_{\eta} \cup \cA| \cdot \eps^{-\upsilon-10} \leq 2K(\eps)  \cdot \eps^{-\upsilon-10} = 2 \eps^{-\eps^{-2}} \cdot \eps^{-3\eps^{-2}-10} \leq \eps^{-6\eps^{-2}} 
		\end{aligned}
	\end{equation}

The first inequality is by Claim~\ref{AETA}. Therefore,

		\begin{equation*}
		\begin{aligned}
			\label{bound1}
		|\{P(C')~|~C' \in \textsf{supp}(\bar{y})\}| \leq{} & \left( \eps^{-10} |\mathcal{Q}| \right)^{\eps^{-10}} \leq  \left( \eps^{-\eps^{-1}} \cdot \eps^{-6\eps^{-2}} \right)^{\eps^{-10}} \leq    \left( \eps^{-\eps^{-3}} \right)^{\eps^{-10}} \leq \left( \exp(\eps^{-1})^{\eps^{-3}}\right)^{\eps^{-10}}\\ \leq{} &  \left( \exp(\eps^{-4}) \right)^{\eps^{-10}} \leq  \exp(\eps^{-14}) \leq  \exp(\eps^{-14}) +3K(\eps)-3K(\eps) \\
		 \leq{} & 2 \cdot \exp(\eps^{-14})-3K(\eps) \leq Q(\eps)-3K(\eps).
		\end{aligned}
	\end{equation*} 	Recall that for each $C' \in \textsf{supp}(\bar{y})$ it holds that $|C'| \leq \eps^{-10}$; in addition, by \eqref{eq:mapS,eq:mapM}, for each item $\ell \in P(C')$ there is a single $\Phi \in \mathcal{Q}$ such that $\ell \in \textsf{last}_{|C' \cap \Phi|} (\Phi)$ and it follows that the number of distinct options for choosing $\ell$ is bounded by $|C' \cap \Phi| \cdot |\mathcal{Q}| \leq |C'| \cdot |\mathcal{Q}| \leq \eps^{-10} \cdot |\mathcal{Q}|$. Thus, by choosing at most $\eps^{-10}$ items as described above, the first inequality follows. The second inequality is by \eqref{bound1111}. The last inequalities follow since $0<\eps<0.1$, $K(\eps) = \eps^{-\eps^{-2}}$ and $Q(\eps) =   \exp \left(\eps^{-17}\right)$.

\item Let $G \in \cG$. Now, 
	\begin{equation*}
	\begin{aligned}
	|G \cap P(C)| \leq{} &  \left|   \bigcup_{\Phi \in \mathcal{Q} \text{ s.t. } \Phi \subseteq G} \textsf{last}_{|C \cap \Phi|} (\Phi)       \right| = \sum_{\Phi \in \mathcal{Q} \text{ s.t. } \Phi \subseteq G} \left|        \textsf{last}_{|C \cap \Phi|} (\Phi)       \right| \\ ={} & \sum_{\Phi \in \mathcal{Q} \text{ s.t. } \Phi \subseteq G} |C \cap \Phi|  \leq  |G \cap C|
	\end{aligned}
\end{equation*}

	The first inequality is by \eqref{eq:mapS,eq:mapM}. The second equality is because each item belongs to at most one class. The third equality is because for all $\Phi \in \mathcal{Q}$ the number of items in $\textsf{first}_{|C \cap \Phi|} (\Phi)$ is the same as in $|C \cap \Phi|$ by the definition of $\textsf{first}$. The last inequality is because by  \eqref{eq:mapS,eq:mapM} every item in $C \cap G$ belongs to at most one class and each class is a subset of some group. 
\end{enumerate}
\qed

In the following we prove Lemma~\ref{lem:Cnf}. The proof relies on the following claims and definitions and is given at the end of this section. 

\begin{claim}
	\label{lem:shiftingPolynomial}
	Algorithm~\ref{Alg:shifting} is runs in time $\textnormal{poly}(\frac{1}{\eps}, |\cI|)$.
\end{claim}

\begin{proof}
Finding the significant groups can be done in polynomial time by sorting the groups according to the small item frequencies, where the the small item frequency of a group can be computed in polynomial time as $|\textsf{supp}(\bar{y})|$ is polynomial. In addition, finding the massive groups, all groups with large items can be computed in linear time by iterating over all groups. Thus, the construction of the important groups takes polynomial time and therefore the constructing the classes in Step~\ref{step:classesQ} can be also achieved in polynomial time by \eqref{eq:classS}. Furthermore, since $|\textsf{supp}(\bar{y})|$ is polynomial, computing the projections of all $C \in \textsf{supp}(\bar{y})$ is polynomial. Finally, the linear combinations in Step~\ref{step:u} and Step~\ref{step:shiftreturn} are of a polynomial number of elements; thus, they can be computed in polynomial time. 
\end{proof}

	For convenience, we repeat the construction of $\bar{\lambda}$ as given in Section~\ref{sec:AlgPrototype}. Recall that $\bar{\gamma}$ in the $\bar{y}$-polytope such that for all $\ell,j \in I$, $\ell \neq j$ it holds that $\bar{\gamma}_{\ell,j} = 0$. We define below a point $\bar{\lambda} \in \mathbb{R}^{I \times (I \cup \mathcal{C})}_{\geq 0}$ and show that $\bar{\lambda}$ is in the $\bar{z}$-polytope. Thus, it follows that the $\bar{z}$-polytope is non-empty. For all $C \in \cC$, let $P^{-1}(C) = \{C' \in \cC~|~ P(C') = C\}$. We start with constructing  a point $\bar{\psi} \in \mathbb{R}^{I \times (I \cup \mathcal{C})}_{\geq 0}$, which is used in the construction of $\bar{\lambda}$. For all $\ell \in I$ and $C \in \cC$ define
	
	\begin{equation}
		\label{lambda'lC}
		\bar{\psi}_{\ell,C} = \sum_{C' \in P^{-1}(C)} \bar{\gamma}_{\ell,C'}.
	\end{equation}

	 Moreover, for any $G \in \cG_{\eta} \cup \cA$, $i \in [\tau(G)-1]$, $\ell \in \Delta_{i+1}(G)$ and $j \in \Delta_{i}(G)$ define

	\begin{equation}
		\label{lambda'lDelta}
	\bar{\psi}_{\ell,j} =   \frac{f_{\bar{u}}(j)}{f_{\bar{u}}(\Delta_{i}(G))} \cdot \bar{\gamma}_{\ell,\ell}.
	\end{equation}

For any other $\ell,j \in I$ define $	\bar{\psi}_{\ell,j} = 0$. We use the following auxiliary claims. 
		
		\begin{claim}
			\label{psi5}
		$\bar{\psi}$ satisfies Constraint \eqref{F1} of the $\bar{u}$-polytope. 
	\end{claim}

\begin{proof}
Let $\ell \in I$ and $t \in I \cup \cC$ such that $\bar{\psi}_{\ell,t}>0$. We split into two cases as follows. 

\begin{enumerate}
	\item  $t \in \cC$. By \eqref{lambda'lC}, there exists $C \in P^{-1}(t)$ such that $\bar{\gamma}_{\ell,C}>0$.  Therefore, $\ell \in \textsf{fit}(C)$ by \eqref{F1} since $\bar{\gamma}$ is in the $\bar{y}$-polytope. Consequently, by Lemma~\ref{lem:P(C)} it follows that $\ell \in  \textsf{fit}(t)$. 
	
	\item There are $G \in \cG_{\eta} \cup \cA$ and $i \in [\tau(G)-1]$ such that $t \in \Delta_{i}(G)$. By \eqref{lambda'lDelta} it follows that $\ell \in \Delta_{i+1}(G)$. Hence, $\ell \in  \textsf{fit}(t)$ by Observation~\ref{ob:FG}. 
	
\end{enumerate} By the definition of $\bar{\psi}$, for any $\ell \in I$, the above cases cover all possibilities for all $t \in I \cup \cC$ such that $\bar{\psi}_{\ell,t}>0$. 
\end{proof}

		\begin{claim}
			\label{psi6}
		$\bar{\psi}$ satisfies constraint \eqref{F2} of the $\bar{u}$-polytope. 
	\end{claim}

	\begin{proof}
			Let $C \in \mathcal{C}$. 
		
			\begin{equation*}
			\label{C!E}
			\begin{aligned}
				\sum_{\ell \in I} \bar{\psi}_{\ell,C} \cdot  s(\ell) ={} &  \sum_{\ell \in I~} \sum_{C' \in P^{-1}(C)} \bar{\gamma}_{\ell,C'} \cdot  s(\ell)= \sum_{C' \in P^{-1}(C)~}  \sum_{\ell \in I}  \bar{\gamma}_{\ell,C'} \cdot  s(\ell)  \\
				\leq{} & \sum_{C' \in P^{-1}(C)~} \left(1-s(C')\right) \bar{y}_{C'} \leq  \sum_{C' \in P^{-1}(C)~}  \left(1-s(C)\right) \bar{y}_{C'} \leq   \left(1-s(C)\right) \bar{u}_C. 
			\end{aligned}
		\end{equation*}

		The  first inequality is because $\bar{\gamma}$ is in the $\bar{y}$-polytope; thus, the inequality follows by \eqref{F2}. The second inequality is because for all $C' \in \textsf{supp}(\bar{y})$ it holds that $s(P(C')) \leq s(C')$ by Lemma~\ref{lem:P(C)}. The last inequality is by Step~\ref{step:u} of Algorithm~\ref{Alg:shifting}.

	\end{proof}

		\begin{claim}
			\label{psi7}
		$\bar{\psi}$ satisfies constraint \eqref{F5} of the $\bar{u}$-polytope. 
	\end{claim}

	\begin{proof}

		Let $G \in \mathcal{G}$ and $C \in \mathcal{C}$.

			\begin{equation*}
			\begin{aligned}
			\sum_{\ell \in G} \bar{\psi}_{\ell,C} ={} & \sum_{\ell \in G~} \sum_{C' \in P^{-1}(C)} \bar{\gamma}_{\ell,C'} = \sum_{C' \in P^{-1}(C)~} \sum_{\ell \in G} \bar{\gamma}_{\ell,C'} \leq  \sum_{C' \in P^{-1}(C)} \bar{y}_{C'} \cdot \left(k(G)-|C' \cap G|\right) \\ \leq{} &  \left(k(G)-|C \cap G|\right) \sum_{C' \in P^{-1}(C)} \bar{y}_{C'}  \leq \bar{u}_C \cdot \left(k(G)-|C \cap G|\right) .
			\end{aligned}
		\end{equation*}
		
		 The first inequality is because $\bar{\gamma}$ is in the $\bar{y}$-polytope; thus, the inequality follows by \eqref{F5}. The second inequality is by Property~5 of Lemma~\ref{lem:P(C)}.

	\end{proof}

		We use the following auxiliary claims.
		
				\begin{claim}
				\label{aux:1}
				For all $\Phi \in \mathcal{Q}$ it holds that $f_{\bar{y}}(\Phi) \leq \frac{f_{\bar{u}}(\Phi)}{\left(1+\frac{2\eps^{-\upsilon}}{\|\bar{y}\|}\right)} $. 
			\end{claim}
			
			\begin{proof}
				\begin{equation}
					\label{a1jL}
					\begin{aligned}
						f_{\bar{y}}(\Phi) ={} & \sum_{\ell \in \Phi} f_{\bar{y}}(\ell) = \sum_{\ell \in \Phi~} \sum_{C \in \cC[\ell]} \bar{y}_C = \sum_{C \in \cC~} \sum_{\ell \in \Phi \cap C}  \bar{y}_C = \sum_{C \in \cC~} |C \cap \Phi| \cdot \bar{y}_C\\ ={} & 
						\sum_{C \in \cC~}  \sum_{C' \in P^{-1}(C)~} |C' \cap \Phi| \cdot \bar{y}_{C'} =  \sum_{C \in \cC~} |C \cap \Phi|  \sum_{C' \in P^{-1}(C)~}  \bar{y}_{C'} \\ \leq{} &  \sum_{C \in \cC~} |C \cap \Phi|  \cdot \frac{\bar{u}_C}{\left(1+\frac{2\eps^{-\upsilon}}{\|\bar{y}\|}\right)} = \frac{1}{\left(1+\frac{2\eps^{-\upsilon}}{\|\bar{y}\|}\right)} \cdot \sum_{C \in \cC~} |C \cap \Phi|  \cdot \bar{u}_C.
					\end{aligned}
				\end{equation} The sixth equality is because for all $C \in \cC$ it holds that $|C \cap \Phi| = |P(C) \cap \Phi|$ by \eqref{eq:mapS,eq:mapM}. The last inequality is by Step~\ref{step:u} of Algorithm~\ref{Alg:shifting}. Additionally, 	\begin{equation}
					\label{a2jL}
					\begin{aligned}
						f_{\bar{u}}(\Phi) ={} & \sum_{\ell \in \Phi} f_{\bar{u}}(\ell) = \sum_{\ell \in \Phi~} \sum_{C \in \cC[\ell]} \bar{u}_C = \sum_{C \in \cC~} \sum_{\ell \in \Phi \cap C}  \bar{u}_C = \sum_{C \in \cC~} |C \cap \Phi| \cdot \bar{u}_C
					\end{aligned}
				\end{equation} The claim follows by \eqref{a1jL} and \eqref{a2jL}. 
			\end{proof}

		\begin{claim}
	\label{jS}
	Let $G \in \cG_{\eta} \cup \cA$, $i \in [\tau(G)-1]$ and let $j \in \Delta_{i}(G)$. It holds that $\sum_{\ell \in I} \bar{\psi}_{\ell,j} \leq f_{\bar{u}}(j)$.
\end{claim}

\begin{proof}
		\begin{equation*}
		\begin{aligned}
			\sum_{\ell \in I} \bar{\psi}_{\ell,j} ={} & \sum_{\ell \in  \Delta_{i+1}(G)} \frac{ f_{\bar{u}}(j) }{f_{\bar{u}}( \Delta_{i}(G))} \cdot \bar{\gamma}_{\ell,\ell} \leq  \sum_{\ell \in  \Delta_{i+1}(G)} \frac{ f_{\bar{u}}(j) }{f_{\bar{u}}( \Delta_{i}(G))} \cdot f_{\bar{y}}(\ell)  \cdot  = \frac{ f_{\bar{u}}(j) }{f_{\bar{u}}( \Delta_{i}(G))} \cdot f_{\bar{y}}(\Delta_{i+1}(G)) \\ \leq{} &  \frac{ f_{\bar{u}}(j) }{f_{\bar{y}}( \Delta_{i}(G))\cdot \left(1+\frac{2\eps^{-\upsilon}}{\|\bar{y}\|}\right)} \cdot   { f_{\bar{y}}(\Delta_{i+1}(G))}{} 
			\leq 
			 \frac{ f_{\bar{u}}(j) }{\eps^{\upsilon}\cdot\|\bar{y}\| \cdot \left(1+\frac{2\eps^{-\upsilon}}{\|\bar{y}\|}\right)} \cdot   {\left(\eps^{\upsilon}\cdot\|\bar{y}\|+2\right)}{}  \\
			  =& f_{\bar{u}}(j) = \sum_{C \in \cC[j]} \bar{u}_C.
		\end{aligned}
	\end{equation*}
	 The first equality is by \eqref{lambda'lDelta}. The second inequality is by Claim~\ref{aux:1}. The third inequality is by Observation~\ref{ob:FG}.

\end{proof}

The following observation follows from Claim~\ref{jS}, that cover the only case where for some $\ell,j \in I$ it holds that $\bar{\psi}_{\ell,j}>0$ by \eqref{lambda'lDelta}. 
	\begin{observation}
		\label{psi8}
	$\bar{\psi}$ satisfies constraint \eqref{F3} of the $\bar{u}$-polytope. 
\end{observation}

\noindent{\bf Proof of Claim~\ref{ob:psi}:}
Claim~\ref{ob:psi} follows from Claim~\ref{psi5}, Claim~\ref{psi6}, Claim~\ref{psi7}, and Observation~\ref{psi8}. \qed

We show below that $\psi$ satisfies constraint \eqref{F4} for a specific subset of the items, using the following auxiliary claims.

	\begin{claim}
	\label{aux:65}
	For all $\ell \in I$ it holds that $\sum_{C \in \cC} \bar{\psi}_{\ell,C} = \sum_{C \in \cC} \bar{\gamma}_{\ell,C}$. 
\end{claim}

\begin{proof}
	\begin{equation*}
		\label{1jL}
		\begin{aligned}
		\sum_{C \in \cC} \bar{\psi}_{\ell,C} ={} & \sum_{C \in \cC~} \sum_{C' \in P^{-1}(C)} \bar{\gamma}_{\ell,C'} = \sum_{C' \in \cC~~}  \sum_{C \in \cC \text{ s.t. } P(C') = C}    \bar{\gamma}_{\ell,C'} = \sum_{C \in \cC}    \bar{\gamma}_{\ell,C}.
		\end{aligned}
	\end{equation*} 
\end{proof}

	\begin{claim}
	\label{psi9S}
	For all $G \in \cG_{\eta} \cup \cA$ and $\ell \in G \setminus \Delta_{1}(G)$ it holds that $\sum_{t \in I \cup \cC} \bar{\psi}_{\ell,t} \geq 1$. 
\end{claim}

\begin{proof}

 by \eqref{lambda'lDelta}, there is $i \in \{2,3, \ldots, \tau(G)\}$ such that $\ell \in \Delta_{i}(G)$. Therefore,
 
 	\begin{equation*}
 	\begin{aligned}
 		\sum_{t \in I \cup \cC} \bar{\psi}_{\ell,t} ={} & \sum_{j \in \Delta_{i-1}(G)} \frac{f_{\bar{u}}(j)}{f_{\bar{u}}(\Delta_{i-1}(G))} \cdot \bar{\gamma}_{\ell,\ell} +\sum_{C \in \cC} \bar{\psi}_{\ell,C} =  \frac{\bar{\gamma}_{\ell,\ell} }{{f_{\bar{u}}(\Delta_{i-1}(G))}} \sum_{j \in \Delta_{i-1}(G)} f_{\bar{u}}(j) +\sum_{C \in \cC} \bar{\psi}_{\ell,C} \\ 
 		={} & \frac{\bar{\gamma}_{\ell,\ell} }{{f_{\bar{u}}(\Delta_{i-1}(G))}} \cdot f_{\bar{u}}(\Delta_{i-1}(G)) +\sum_{C \in \cC} \bar{\gamma}_{\ell,C} = \sum_{t \in I \cup \cC} \bar{\gamma}_{\ell,t} \geq 1.  
 	\end{aligned}
  \end{equation*} 
 
The first equality is by \eqref{lambda'lDelta}. The third equality is by Claim~\ref{aux:65}. The fourth equality is because for all $\ell \neq j \in I$ it holds that $\bar{\gamma}_{\ell,j} = 0$. The last inequality is because $\bar{\gamma}$ is in the $\bar{y}$-polytope and thus satisfies property \eqref{F4}.

\end{proof}

We repeat the definition of $\bar{\mu}$ from Section~\ref{sec:AlgPrototype}. We define a point $\bar{\mu} \in \mathbb{R}^{I}_{\geq 0}$ as follows. For all $\ell \in I$, define 

\begin{equation}
	\label{mu}
	\bar{\mu}_{\ell} = \max \left\{0,1-\sum_{t \in I \cup \cC} \bar{\psi}_{\ell,t}\right\}.
\end{equation}

	We prove below some properties of $\bar{\mu}$.

		\begin{claim}
		\label{mu:fit}
		$\textnormal{\textsf{supp}}(\bar{\mu}) \setminus L \subseteq \textnormal{\textsf{fit}}(\emptyset)$.
	\end{claim}
	
	\begin{proof}
		
		Let $\ell \in \textsf{supp}(\bar{\mu})  \setminus L$. Thus, it follows that $\ell \in I \setminus L$. Hence, $\ell \in \textsf{fit}(\emptyset)$ by \eqref{fitS}. 
	\end{proof}
	
		\begin{claim}
		\label{mu:fit2}
		For all $\ell \in \textnormal{\textsf{supp}}(\bar{\mu})$ it holds that $\ell \in  \textnormal{\textsf{fit}} \left(r\left(\textnormal{\textsf{group}}(\ell)\right)\right)$. 
	\end{claim}
	
	\begin{proof}
		
		Let $\ell \in \textsf{supp}(\bar{\mu})$. Thus, by the definition of the maximal slot $r\left(\textnormal{\textsf{group}}(\ell)\right)$ it holds that $s(\ell) \leq s\left(r\left(\textnormal{\textsf{group}}(\ell)\right)\right)$ and that $\textnormal{\textsf{group}}(\ell) = r\left(\textnormal{\textsf{group}}(\ell)\right)$. Hence, the claim follows by \eqref{fitl}. 
	\end{proof}
	
	We use the following auxiliary claim. 
		\begin{claim}
		\label{aux:mugroup}
	For all $S \subseteq I$ it holds that $\sum_{\ell \in S} \bar{\mu}_{\ell} \leq f_{\bar{y}}(S)$. 
	\end{claim}
	
	\begin{proof}
		
			\begin{equation*}
			\begin{aligned}
			\sum_{\ell \in S} \bar{\mu}_{\ell} ={} &  \sum_{\ell \in S} \max \left\{0,1-\sum_{t \in I \cup \cC} \bar{\psi}_{\ell,t}\right\} \leq  \sum_{\ell \in S} \max \left\{0,1-\sum_{C \in \cC} \bar{\psi}_{\ell,C}\right\} = \sum_{\ell \in S} \max \left\{0,1-\sum_{C \in  \cC} \bar{\gamma}_{\ell,C}\right\}\\
			 \leq{} & \sum_{\ell \in S} \max \left\{0,\sum_{j \in I} \bar{\gamma}_{\ell,j}\right\} = \sum_{\ell \in S}  \bar{\gamma}_{\ell,\ell} \leq  \sum_{\ell \in S} f_{\bar{y}}(\ell) =  f_{\bar{y}}(S). 
			\end{aligned}
		\end{equation*}  The second equality is by Claim~\ref{aux:65}. The second inequality is because $\bar{\gamma}$ is in the $\bar{y}$-polytope; thus, the inequality follows by \eqref{F4}. 
	\end{proof}

	\begin{claim}
	\label{mu:group:ups}
	For all $G \in \cG_{\eta} \cup \cA$ it holds that $\sum_{\ell \in G} \bar{\mu}_{\ell} \leq \eps^{\upsilon} \cdot \|\bar{y}\|+2$.
\end{claim}

\begin{proof}
		\begin{equation*}
		\begin{aligned}
			\sum_{\ell \in G} \bar{\mu}_{\ell} = \sum_{\ell \in \Delta_{1}(G)} \bar{\mu}_{\ell} \leq{} & f_{\bar{y}} \left( \Delta_{1}(G) \right) \leq  \eps^{\upsilon}\cdot \|\bar{y}\|+2.
		\end{aligned}
	\end{equation*}  The first equality is by Claim~\ref{psi9S} and \eqref{mu}. The first inequality is by Claim~\ref{aux:mugroup}. The second inequality is by Observation~\ref{ob:FG}.
\end{proof}

		\begin{claim}
		\label{mu:group}
		For all $G \in \cG$ it holds that $\sum_{\ell \in G} \bar{\mu}_{\ell} \leq \eps \cdot \|\bar{y}\|+2$.
	\end{claim}
	
	\begin{proof}
	We split the proof into two cases. \begin{enumerate}
		\item $G \in \cG_{\eta} \cup \cA$. The claim follows by Claim~\ref{mu:group:ups}. 
		
		\item  $G \in \cG \setminus (\cG_{\eta} \cup \cA)$. Then, 	\begin{equation*}
			\begin{aligned}
				\sum_{\ell \in G} \bar{\mu}_{\ell} ={} &  \sum_{\ell \in G \setminus L} \bar{\mu}_{\ell} \leq f_{\bar{y}} \left( G \setminus L \right) \leq \eps\cdot \|\bar{y}\| \leq  \eps  \cdot \|\bar{y}\|+2.
			\end{aligned}
		\end{equation*}  The first equality is because $G \notin \cA$; thus, it follows that $L \cap G = \emptyset$ by \eqref{Agroups}.  The first inequality is by Claim~\ref{aux:mugroup}. The second inequality is by Lemma~\ref{lem:significantDiscard}. 
	\end{enumerate}
	\end{proof}

		\begin{claim}
		\label{mu:size}
	$\sum_{\ell \in I \setminus L} \bar{\mu}_{\ell}  \cdot s(\ell) \leq \eps\cdot \|\bar{y}\|$.
	\end{claim}

	\begin{proof}
		We use the following inequalities. First, 
			\begin{equation}
				\label{aux:size1}
			\begin{aligned}
		 	\sum_{\ell \in I \setminus L} \bar{\mu}_{\ell}  \cdot s(\ell)  ={} & \sum_{\ell \in I \setminus L}  \max \left\{0,1-\sum_{t \in I \cup \cC} \bar{\psi}_{\ell,t}\right\} \cdot s(\ell)\\ \leq{} & \sum_{\ell \in I \setminus L} \max \left\{0,1-\sum_{C \in \cC} \bar{\psi}_{\ell,C}\right\} \cdot s(\ell) = \sum_{\ell \in I \setminus L} \max \left\{0,1-\sum_{C \in  \cC} \bar{\gamma}_{\ell,C}\right\} \cdot s(\ell) 
		\end{aligned}
		\end{equation} The last equality is by Claim~\ref{aux:65}. 	Second, 
	
		\begin{equation}
			\label{aux:size2}
		\begin{aligned}
			\sum_{\ell \in I \setminus L} \max \left\{0,1-\sum_{C \in  \cC} \bar{\gamma}_{\ell,C}\right\} \cdot s(\ell) 
		\leq{} & \sum_{\ell \in I \setminus L} \max \left\{0,\sum_{j \in I} \bar{\gamma}_{\ell,j}\right\} \cdot s(\ell) = \sum_{\ell \in I \setminus L}  \bar{\gamma}_{\ell,\ell} \cdot s(\ell)\\ \leq{} &  \sum_{\ell \in I \setminus L} f_{\bar{y}}(\ell) \cdot s(\ell) \leq \eps\|\bar{y}\|. 
		\end{aligned}
	\end{equation} The first inequality is because $\bar{\gamma}$ is in the $\bar{y}$-polytope; thus, the inequality follows by \eqref{F4}. The last inequality is by Lemma~\ref{lem:significantDiscard}. The claim follows by \eqref{aux:size1} and \eqref{aux:size2}.

	\end{proof}

	\noindent{\bf Proof of Claim~\ref{lem:mu}:}
	The proof follows by Claim~\ref{mu:fit}, Claim~\ref{mu:group}, and Claim~\ref{mu:size}. \qed
	
	We Define another point $\bar{\rho} \in \mathbb{R}^{I \times (I \cup \mathcal{C})}_{\geq 0}$ based on $\bar{\mu}$. For all $\ell \in I \setminus L$ define $\bar{\rho}_{\ell,\emptyset} = \bar{\mu}_{\ell}$; in addition, for all $G \in \cG, \ell \in G$ define $\bar{\rho}_{\ell,r(G)} = \bar{\mu}_{\ell}$; for any other entry $(\ell,t)  \in I \times \left(I \cup \cC\right)$ define $\bar{\rho}_{\ell,t} = 0$. Finally, using the definitions of $\bar{\rho}$ and $\bar{\psi}$ we define:

	 \begin{equation}
	 	\label{LA}
\bar{\lambda} = \bar{\rho}+\bar{\psi}.
	 \end{equation}

	 	\begin{claim}
	 	\label{LA5}
	 $\bar{\lambda}$ satisfies constraint \eqref{F1} of the $\bar{z}$-polytope. 
	 \end{claim}
	 \begin{proof}
	 	 Let $\ell \in I$ and $t \in I \cup \cC$ such that $\bar{\lambda}_{\ell,t}>0$. We split into three cases.
	 	 
	 	 \begin{enumerate}
	 	 	\item $\ell \in I \setminus L$ and $\bar{\rho}_{\ell,t}>0$.  By the definition of $\bar{\rho}$ it follows that $t = \emptyset$. In addition, $\bar{\rho}_{\ell,\emptyset} = \bar{\mu}_{\ell}$ and it follows that $\ell \in \textsf{supp}(\bar{\mu})$. Therefore, by Lemma~\ref{mu:fit} it holds that $\ell \in \textsf{fit}(t)$. 
	 	 	
	 	 		\item $\ell \in L$ and $\bar{\rho}_{\ell,t}>0$.  By the definition of $\bar{\rho}$ it follows that $t = r\left(\textsf{group}(\ell)\right)$. In addition, $\bar{\rho}_{\ell,r\left(\textsf{group}(\ell)\right)} = \bar{\mu}_{\ell}$ and it follows that $\ell \in \textsf{supp}(\bar{\mu})$. Therefore, by Lemma~\ref{mu:fit2} it holds that $\ell \in \textsf{fit}(t)$. 
	 	 	
	 	 		\item $\bar{\psi}_{\ell,t}>0$.  It follows that $\ell \in \textsf{fit}(t)$ by Claim~\ref{ob:psi}. 
	 	 \end{enumerate} By \eqref{LA}, the above proves the claim. 
	 \end{proof}
	
	 	\begin{claim}
	 	\label{LA6}
	 	$\bar{\lambda}$ satisfies constraint \eqref{F2} of the $\bar{z}$-polytope. 
	 \end{claim}
	 \begin{proof}
	 Let $C \in \cC$. We split into two cases.
	 
	 \begin{enumerate}
	 	\item $C \neq \emptyset$. 	\begin{equation*}
	 		\label{LA6aux}
	 		\begin{aligned}
	 			\sum_{\ell \in I} \bar{\lambda}_{\ell,C} \cdot s(\ell) ={} & \sum_{\ell \in I} \left(\bar{\rho}_{\ell,C}+\bar{\psi}_{\ell,C}\right) \cdot s(\ell) = \sum_{\ell \in I} \bar{\psi}_{\ell,C} \cdot s(\ell) \leq (1-s(C))\bar{u}_C \leq (1-s(C))\bar{z}_C. 
	 		\end{aligned}
	 	\end{equation*} The second equality is because for all $\ell \in I$ and $C' \in \cC \setminus  \{\emptyset\}$ it holds that $\bar{\rho}_{\ell,C'} = 0$. The first inequality is because $\bar{\psi}$ satisfies constraint  \eqref{F2} of the $\bar{u}$-polytope by Claim~\ref{ob:psi}. The second inequality is by Step~\ref{step:shiftreturn} in Algorithm~\ref{Alg:shifting}. 
	 	
	 	\item $C = \emptyset$. 	\begin{equation*}
	 		\begin{aligned}
	 			\sum_{\ell \in I} \bar{\lambda}_{\ell, \emptyset} \cdot s(\ell) ={}
	 			 & \sum_{\ell \in I} \left(\bar{\rho}_{\ell, \emptyset}+\bar{\psi}_{\ell, \emptyset}\right) \cdot s(\ell)  =  
	 			 \sum_{\ell \in I\setminus L} \bar{\mu}_{\ell} \cdot s(\ell)+\sum_{\ell \in I} \bar{\psi}_{\ell, \emptyset} \cdot s(\ell)
	 			 \\ \leq{} & \eps\|\bar{y}\|+(1-s( \emptyset))\bar{u}_{\emptyset} = \eps\|\bar{y}\|+\bar{u}_{\emptyset} \leq \bar{z}_{\emptyset} = (1-s( \emptyset)) \cdot \bar{z}_{\emptyset}.
	 		\end{aligned}
	 	\end{equation*}  The first inequality is by Claim~\ref{mu:size} and because $\bar{\psi}$ satisfies constraint  \eqref{F2} of the $\bar{u}$-polytope by Claim~\ref{ob:psi}. The second inequality is by Step~\ref{step:shiftreturn} in Algorithm~\ref{Alg:shifting}. 
	 \end{enumerate}
	 
	 	 \end{proof}

	 	\begin{claim}
	 	\label{LA7}
	 	$\bar{\lambda}$ satisfies constraint \eqref{F5} of the $\bar{z}$-polytope. 
	 \end{claim}
	 
	 \begin{proof}
	 	Let $G \in \cG$ and $C \in \cC$. We split into two cases.
	 	
	 	\begin{enumerate}
	 		\item $C \neq \emptyset$. 	\begin{equation*}
	 			\label{LA6aux}
	 			\begin{aligned}
	 			\sum_{\ell \in G} \bar{\lambda}_{\ell,C} = \sum_{\ell \in G} \left(\bar{\rho}_{\ell,C}+\bar{\psi}_{\ell,C}\right) =   \sum_{\ell \in G} \bar{\psi}_{\ell,C} \leq \bar{u}_C \cdot \left(   k(G) - |G \cap C|  \right) \leq \bar{z}_C \cdot \left(   k(G) - |G \cap C|  \right).
	 			\end{aligned}
	 		\end{equation*} The second equality is because for all $\ell \in I$ and $C' \neq \emptyset$ it holds that $\bar{\rho}_{\ell,C'} = 0$. The first inequality is because $\bar{\psi}$ satisfies constraint  \eqref{F5} of the $\bar{u}$-polytope by Claim~\ref{ob:psi}. The second inequality is by Step~\ref{step:shiftreturn} in Algorithm~\ref{Alg:shifting}. 
	 		
	 		\item $C = \emptyset$. 	\begin{equation*}
	 			\begin{aligned}
	 				\sum_{\ell \in G} \bar{\lambda}_{\ell,\emptyset} ={} & \sum_{\ell \in G} \left(\bar{\rho}_{\ell,\emptyset}+\bar{\psi}_{\ell,\emptyset}\right)   =  \sum_{\ell \in G\setminus L} \bar{\mu}_{\ell}+ \sum_{\ell \in G} \bar{\psi}_{\ell,\emptyset}  \leq \eps  \cdot \|\bar{y}\|+2+\bar{u}_{\emptyset} \cdot k(G)  \\ \leq{} & \left( \eps\|\bar{y}\|+2 \right) \cdot k(G)+\bar{u}_{\emptyset} \cdot k(G) \leq \bar{z}_{\emptyset} \cdot k(G).
	 			\end{aligned}
	 		\end{equation*}  The first inequality is by Claim~\ref{mu:group} and because $\bar{\psi}$ satisfies constraint  \eqref{F5} of the $\bar{u}$-polytope by Claim~\ref{ob:psi}. The last inequality is by Step~\ref{step:shiftreturn} in Algorithm~\ref{Alg:shifting}. 
	 	\end{enumerate}

	 \end{proof}

	 	\begin{claim}
	 	\label{LA8}
	 	$\bar{\lambda}$ satisfies constraint \eqref{F3} of the $\bar{z}$-polytope. 
	 \end{claim}
 
 \begin{proof}
 	Let $j \in I$. We split into two cases, as follows. \begin{enumerate}
 		\item If $\textsf{group}(j) \in \cG_{\eta}\cup \cA$ and $j = r\left(   \textsf{group}(j) \right)$. Then, 		\begin{equation*}
 			\begin{aligned}
 			\sum_{\ell \in I} \bar{\lambda}_{\ell,j} = \sum_{\ell \in I} \left(\bar{\rho}_{\ell,j}+\bar{\psi}_{\ell,j}\right) =  \sum_{\ell \in \textsf{group}(j)} \bar{\mu}_{\ell,j} + \sum_{\ell \in I} \bar{\psi}_{\ell,j} \leq \eps^{\upsilon}\cdot\|\bar{y}\|+2+\sum_{C \in \cC[j]} \bar{u}_C \leq \sum_{C \in \cC[j]} \bar{z}_C.
 			\end{aligned}
 		\end{equation*}  The first inequality is by Claim~\ref{mu:group:ups} and because $\bar{\psi}$ satisfies constraint  \eqref{F3} of the $\bar{u}$-polytope by Claim~\ref{ob:psi}. The second inequality is by Step~\ref{step:shiftreturn} in Algorithm~\ref{Alg:shifting}. 
 		
 		\item Otherwise, 	$$\sum_{\ell \in I} \bar{\lambda}_{\ell,j} = \sum_{\ell \in I} \left(\bar{\rho}_{\ell,j}+\bar{\psi}_{\ell,j}\right) = \sum_{\ell \in I} \bar{\psi}_{\ell,j} \leq \sum_{C \in \cC[j]} \bar{u}_C \leq \sum_{C \in \cC[j]} \bar{z}_C.$$ The first inequality is because $\bar{\psi}$ satisfies constraint  \eqref{F3} of the $\bar{u}$-polytope by Claim~\ref{ob:psi}. The second inequality is by Step~\ref{step:shiftreturn} in Algorithm~\ref{Alg:shifting}. 
 	\end{enumerate}

 \end{proof}

	 \begin{claim}
	 	\label{LA9}
	 	$\bar{\lambda}$ satisfies constraint \eqref{F4} of the $\bar{z}$-polytope. 
	 \end{claim}
	 
	 \begin{proof}
	 	Let $\ell \in I$. 
	 	$$\sum_{t \in I \cup \cC} \bar{\lambda}_{\ell,t} = \sum_{t \in I \cup \cC} \left(\bar{\rho}_{\ell,j}+\bar{\psi}_{\ell,t}\right) =  \bar{\mu}_{\ell}+\sum_{t \in I \cup \cC} \bar{\psi}_{\ell,t} = \max \left\{0,1-\sum_{t \in I \cup \cC} \bar{\psi}_{\ell,t}\right\}+\sum_{t \in I \cup \cC} \bar{\psi}_{\ell,t} \geq 1.$$ 
	 	
	 	 \end{proof}
 	 
 	 The following observation follows by Claim~\ref{LA5}, Claim~\ref{LA6}, Claim~\ref{LA7}, Claim~\ref{LA8}, and Claim~\ref{LA9}.

	\noindent{\bf Proof of Claim~\ref{lem:shifting1}:}
The proof follows by Claim~\ref{LA5}, Claim~\ref{LA6}, Claim~\ref{LA7}, Claim~\ref{LA8}, and Claim~\ref{LA9}. \qed

\begin{claim}
	\label{lem:shifting2a}
	Algorithm~\ref{Alg:shifting} returns a prototype $\bar{z}$  of $\mathcal{I}$ such that $|\textnormal{\textsf{supp}}(\bar{z})| \leq Q(\eps)$.
\end{claim}

\begin{proof}
		\begin{equation*}
		\begin{aligned}
			|\textsf{supp}(\bar{z})| \leq{} & |\{P(C)~|~ C \in \textsf{supp}(\bar{y})\}|+1+|\cG_{\eta}\cup \cA| \\ \leq{} & Q(\eps)-3K(\eps)+ |\cG_{\eta} \cup \cA|+1 \leq Q(\eps)-3K(\eps)+2K(\eps)+1 \leq Q(\eps). 
		\end{aligned}
	\end{equation*} 
 The first inequality is because for all $C \in \textsf{supp}(\bar{z})$ either $C  = \emptyset$, there is $G \in \cG_{\eta} \cup \cA$ such that $C = r(G)$, or there is $C' \in \textsf{supp}(\bar{y})$ such that $P(C') = C$ by Step~\ref{step:u} and Step~\ref{step:shiftreturn} of Algorithm~\ref{Alg:shifting}. The second inequality is by Lemma~\ref{lem:P(C)}. The third inequality is by Claim~\ref{AETA}. 
\end{proof}

\begin{claim}
	\label{lem:shifting2}
 Algorithm~\ref{Alg:shifting} returns a prototype $\bar{z}$  of $\mathcal{I}$ such that $\|\bar{z}\| \leq (1+5\eps) \|\bar{y}\|+Q(\eps)$. 
\end{claim}

		\begin{proof}
We use the following inequalities.

\begin{equation}
	\label{imp:1}
\begin{aligned}
\left\|     \sum_{G \in \cG_{\eta} \cup \cA} \left( \eps^{\upsilon} \cdot \|\bar{y}\|+2 \right) \cdot \mathbb{I}^{r(G)}     \right\| \leq {} & \sum_{G \in \cG_{\eta} \cup \cA} \left( \eps^{\upsilon} \cdot \|\bar{y}\|+2 \right) \cdot \|\mathbb{I}^{r(G)}\| \\ ={} &  |\cG_{\eta} \cup \cA| \cdot (\eps^{\upsilon}\cdot \|\bar{y}\|+2) \leq 2K(\eps) \cdot (\eps^{\upsilon} \cdot \|\bar{y}\|+2) \\ \leq{} & \eps^{-2\eps^{-2}} \cdot \eps^{3\eps^{-2}} \cdot \|\bar{y}\|+ 4K(\eps) \\ \leq{} & \eps \cdot \|\bar{y}\|+ 4K(\eps). 
\end{aligned}
\end{equation} The first inequality is by the triangle inequality. The second inequality is by Claim~\ref{AETA}.

\begin{equation}
	\label{imp:2}
	\begin{aligned}
    {} & \left\| \left(4\eps \|\bar{y}\|+\eps^{-3} \right)\cdot \mathbb{I}^{\emptyset} +\left(1+\frac{2\eps^{-\upsilon}}{\|\bar{y}\|}\right) \sum_{C \in \textsf{supp}(\bar{y})} \bar{y}_C \cdot \mathbb{I}^{P(C)} \right\| \\
     \leq{}  &  \left\|\left(4\eps \|\bar{y}\|+\eps^{-3} \right)\ \cdot \mathbb{I}^{\emptyset}\right\|+ \left(1+\frac{2\eps^{-\upsilon}}{\|\bar{y}\|}\right)  \left\| \sum_{C \in \textsf{supp}(\bar{y})} \bar{y}_C \cdot \mathbb{I}^{P(C)} \right\| \\ \leq{} & 
     \left( 4\eps \|\bar{y}\| +\eps^{-3}\right)\cdot  \| \mathbb{I}^{\emptyset}\|+  \left(1+\frac{2\eps^{-\upsilon}}{\|\bar{y}\|}\right)  \sum_{C \in \textsf{supp}(\bar{y})} \bar{y}_C \cdot \| \mathbb{I}^{P(C)} \|  \\ 
     ={} & 4\eps \|\bar{y}\| +\eps^{-3}+ \left(1+\frac{2\eps^{-\upsilon}}{\|\bar{y}\|}\right)  \|\bar{y}\| = (1+4\eps)\cdot \|\bar{y}\|  +2\eps^{-\upsilon} +\eps^{-3}
	\end{aligned}
\end{equation} 	The first inequality is by the triangle inequality. The third equality is since for all $C \in \cC$ it holds that $\|\mathbb{I}^C\| = 1$ by the definition of the indicator of $C$. Now, 

			\begin{align*}
				\|\bar{z}\| \leq{} & \left\|  \sum_{G \in \cG_{\eta} \cup \cA} \left( \eps^{\upsilon} \cdot \|\bar{y}\|+2 \right) \cdot \mathbb{I}^{r(G)}  +
				\left(4\eps \|\bar{y}\| +\eps^{-3}\right)\cdot \mathbb{I}^{\emptyset} +\left(1+\frac{2\cdot \eps^{-\upsilon}}{\|\bar{y}\|}\right) \sum_{C \in \textsf{supp}(\bar{y})} \bar{y}_C \cdot \mathbb{I}^{P(C)} \right\| \\ 
				\leq{} &  \eps \cdot \|\bar{y}\|+ 4K(\eps)+(1+4\eps)\cdot \|\bar{y}\|  +2\eps^{-v} +\eps^{-3}\leq (1+5\eps) \cdot \|\bar{y}\|+7\eps^{-3\eps^{-2}} \\ \leq{} & (1+5\eps) \cdot \|\bar{y}\|+Q(\eps).  
			\end{align*} The first equality is by Step~\ref{step:u} and Step~\ref{step:shiftreturn} of Algorithm~\ref{Alg:shifting}. The second inequality is by the triangle inequality and by \eqref{imp:1} and \eqref{imp:2}. The last inequalities are because $\eps <0.1$, $K(\eps)  = \eps^{-\eps^{-2}}$ and $Q(\eps) = \exp(\eps^{-17})$.

		\end{proof}
		
		\begin{claim}
			\label{lem:shifting5}
		Algorithm~\ref{Alg:shifting} returns a prototype $\bar{z}$  of $\mathcal{I}$ such that for all $C \in \textnormal{
		\textsf{supp}}(\bar{z})$ it holds that $|C| \leq \eps^{-10}$. 
		\end{claim}
		
		\begin{proof}
			Let $C \in \textsf{supp}(\bar{z})$ be a configuration. If $C = \emptyset$ or there is $G \in \cG_{\eta} \cup \cA$ such that $C = r(G)$ that contains a single item, then the claim is trivial. Otherwise, because $C \in \textsf{supp}(\bar{z})$, then by Step~\ref{step:u} of Algorithm~\ref{Alg:shifting}, there is $C' \in \textsf{supp}(\bar{y})$ such that $P(C') = C$.  Therefore,  $|C| = |P(C')| \leq \eps^{-10}$ by Lemma~\ref{lem:P(C)}. 
		\end{proof}

		\noindent{\bf Proof of Lemma~\ref{lem:Cnf}:} The proof follows by Claim~\ref{lem:shiftingPolynomial}, Claim~\ref{lem:shifting1}, Claim~\ref{lem:shifting2a}, Claim~\ref{lem:shifting2}, and Claim~\ref{lem:shifting5}. \qed

\section{Omitted Proofs of Section~\ref{sec:FromPolytope}}
\label{app:omittedFromPolytope}

\noindent{\bf Proof of Lemma~\ref{lem:matchingU}:}
For all $\ell \in U$, let $d(\ell) = |\{v \in V~|~ (\ell,v) \in E\}|$ be the degree of $\ell$ in $G$. For all $\ell \in U$ there is exactly one $j \in I$ such that $\bar{\gamma}_{\ell,j} = 1$ because $\bar{\gamma}$ satisfies constraint \eqref{F4} of the $\bar{z}^*$-polytope with equality; let $\ell^* = j$. Therefore, for every $\ell \in U$ it holds:

\begin{equation}
	\label{welldegree}
	d(\ell) =  |\{ (C,\ell^*,k) \in V~|~ (\ell,(C,\ell^*,k)) \in E\}| = \sum_{C \in \cC[\ell^*]~} \sum_{k \in [\bar{z}^*_C]} 1 =  \sum_{C \in \cC[\ell^*]} \bar{z}^*_C \geq \sum_{\ell' \in I} \bar{\gamma}_{\ell',\ell^*} \geq \bar{\gamma}_{\ell,\ell^*} = 1. 
\end{equation} The first inequality is because $\bar{\gamma}$ is in the $\bar{z}^*$-polytope. Define a fractional matching as a vector $\bar{M} \in [0,1]^E$ such that for all $(\ell,v) \in E$ define $\bar{M}_{(\ell,v)} = \frac{1}{d(\ell)}$; note that $\bar{M}$ is well-defined (i.e., for all $\ell \in U$ it holds that $d(\ell)>0$) by \eqref{welldegree}. We show below that $\bar{M}$ is a fractional matching in $G$ of size~$|U|$. 

\begin{enumerate}
	\item $\bar{M}$ is a feasible fractional matching in $G$. Let $\ell' \in U$. It holds that:
	
	$$\sum_{(\ell',v) \in E} \bar{M}_{(\ell',v)} = \sum_{(\ell',v) \in E}  \frac{1}{d(\ell)}  = d(\ell) \cdot \frac{1}{d(\ell)}=1.$$ In addition, let $(C,j,k) \in V$ and let $v = (C,j,k)$. It holds that:

	\begin{equation*}
		\begin{aligned}
			\sum_{(\ell,v) \in E} \bar{M}_{(\ell,v)} ={} & \sum_{(\ell,v) \in E}  \frac{1}{d(\ell)} = \sum_{\ell \in U \text{ s.t. } \bar{\gamma}_{\ell,j} = 1}  \frac{1}{d(\ell)} \leq \sum_{\ell \in U \text{ s.t. } \bar{\gamma}_{\ell,j} = 1} \bar{\gamma}_{\ell,j} \cdot \frac{1}{d(\ell)}\\ \leq{} &
			\sum_{\ell \in U \text{ s.t. } \bar{\gamma}_{\ell,j} = 1}  \bar{\gamma}_{\ell,j} \cdot \frac{1}{\sum_{\ell' \in I} \bar{\gamma}_{\ell',j}} \leq  \frac{1}{\sum_{\ell' \in I} \bar{\gamma}_{\ell',j}} \cdot \sum_{\ell \in I} \bar{\gamma}_{\ell,j} = 1.
		\end{aligned}
	\end{equation*} The second inequality is by \eqref{welldegree}. By the above, we conclude that $\bar{M}$ is a feasible fractional matching in $G$.
	
	\item $|\bar{M}| = |U|$. $$|\bar{M}| =  \sum_{e \in E}\bar{M}_{e} =  \sum_{\ell \in U} \sum_{(\ell,v) \in E} \bar{M}_{e} =  \sum_{\ell \in U} \sum_{(\ell,v) \in E} \frac{1}{d(\ell)}  = \sum_{\ell \in U} \frac{d(\ell)}{d(\ell)}=  |U|.$$
\end{enumerate} By the above, $\bar{M}$ is a fractional matching in $G$ with size $|U|$. Since $G$ is a bipartite graph, there is an integral matching with size $|U|$ in $G$ \cite{schrijver2003combinatorial}.  \qed

\noindent{\bf Proof of Lemma~\ref{lem:allowed}:}
For all $\ell \in A_{C,k}$ there is exactly one $j \in C$ such that $(\ell,(C,j,k)) \in M$ by the definition of $A_{C,k}$ and because $M$ is a matching; we define $\ell^* = j$ to be the {\em slot of } $\ell$. We construct an injective function $\sigma: A_{C,k} \rightarrow C$ such that  for all $\ell \in A_{C,k}$ define $\sigma(\ell) = \ell^*$. By the above, for all $\ell \in A_{C,k}$ it holds that $\ell^*$ is well defined and it follows that $\sigma$ is a function. Let $\ell_1,\ell_2 \in A_{C,k}$ such that $\ell_1 \neq \ell_2$. Since $M$ is a matching, it follows that $\ell^*_1 \neq \ell^*_2$; hence, $\sigma(\ell_1) \neq \sigma(\ell_2)$ and we conclude that $\sigma$ is injective. Finally, let $\ell \in A_{C,k}$. Since $(\ell,(C,\ell^*,k)) \in M$, by the definition of $E$ it follows that $\bar{\gamma}_{\ell,\ell^*} = 1$; thus, $\ell \in \textsf{fit}(\ell^*)$ by \eqref{F1} since $\bar{\gamma}$ is in the $\bar{y}$-polytope. \qed

\noindent{\bf Proof of Claim~\ref{eqeq6}:}
We prove the three properties of the claim below. 

\begin{enumerate}
	\item $D_C \subseteq \textsf{fit}(C)$: Let $\ell \in D_C$. By the definition of $D_C$ it holds that $\bar{\gamma}_{\ell,C} = 1$, and in particular $\bar{\gamma}_{\ell,C} >0$; thus, since $\bar{\gamma}$ is in the $\bar{z}^*$-polytope it follows by \eqref{F1} that $\ell \in \textsf{fit}(C)$. 
	
	\item $s(D_C) \leq \left(1-s(C)\right) \cdot |B_C|$:
	\begin{equation*}
		s(D_C) = \sum_{\ell \in I \text{ s.t. } \bar{\gamma}_{\ell,C} = 1} s(\ell) =    \sum_{\ell \in I \text{ s.t. } \bar{\gamma}_{\ell,C} = 1} \bar{\gamma}_{\ell,C} \cdot s(\ell) \leq   \sum_{\ell \in I} \bar{\gamma}_{\ell,C} \cdot s(\ell)  
		\leq (1-s(C)) \bar{z}^*_C \leq  \left(1-s(C)\right) \cdot |B_C|. 
	\end{equation*} The second inequality is because $\bar{\gamma}$ is in the $\bar{z}^*$-polytope; thus, the inequality follows by \eqref{F2}. The last inequality is by \eqref{BC}.

	\item  $|D_C \cap G| \leq |B_C| \cdot \left(  k(G) - |G \cap C|  \right)$:
	\begin{equation*}
		\begin{aligned}
			|D_C \cap G| ={} & \left|\bigcup_{\ell \in I \text{ s.t. } \bar{\gamma}_{\ell,C} = 1} \{\ell\} \cap G \right| \leq \sum_{\ell \in I \text{ s.t. } \bar{\gamma}_{\ell,C} = 1} \left| \{\ell\} \cap G \right| = \sum_{\ell \in G \text{ s.t. } \bar{\gamma}_{\ell,C} = 1} \left| \{\ell\} \right|\\
			={} & \sum_{\ell \in G \text{ s.t. } \bar{\gamma}_{\ell,C} = 1} \bar{\gamma}_{\ell,C} \leq \sum_{\ell \in G} \bar{\gamma}_{\ell,C} \leq \bar{z}^*_C \cdot \left(  k(G) - |G \cap C|  \right) \leq |B_C| \cdot \left(  k(G) - |G \cap C|  \right).  
		\end{aligned} 
	\end{equation*} The third inequality is because $\bar{\gamma}$ is in the $\bar{z}^*$-polytope; thus, the inequality follows by \eqref{F5}. The last inequality is by \eqref{BC}. \qed
	
\end{enumerate}

\noindent{\bf Proof of Lemma~\ref{lem:assign}:} 
We divide the proof of Lemma~\ref{lem:assign} to the following claims. For simplicity, let $A^* = (A_1, \ldots, A_m)$.

\begin{claim}
	\label{eqeq1}
	$A^*$ is a packing.
\end{claim} 

\begin{proof}
	Let $i \in [m]$. By \eqref{Astar}, we split into two cases.  
	
	\begin{enumerate}
		\item $A_i \in F$.  Then, by the definition of $F$ it holds that $A_i$ is a singleton, a set containing a single item. Thus, $s\left(A_i\right) \leq 1$ and for all $G \in \cG$ it holds that $|G \cap A^*_b| \leq 1 \leq k(G)$.
		
		\item  There are $C \in C^*$ and $k \in [\bar{z}^*_{C^*}]$ such that $A_i = A_{C,k}$. Therefore, by Lemma~\ref{lem:allowed} it holds that $A^*_b$ is allowed in $C$. Hence, there is an injective function $\sigma: A_i \rightarrow C$ such that for all $\ell \in A_i$ it holds that $\ell \in \textsf{fit}\left(\sigma(\ell)\right)$. Therefore, 
		
		\begin{enumerate}
			\item $s(A_i) \leq 1$. 	$$s(A_i) = \sum_{\ell \in A_i} s(\ell) \leq  \sum_{\ell \in A_i} s\left(\sigma(\ell)\right) \leq \sum_{\ell \in C} s(\ell) = s(C) \leq 1.$$ The first inequality is because for all $\ell \in A_i$ it holds that $\ell \in \textsf{fit}\left(\sigma(\ell)\right)$; thus, the inequality follows by \eqref{fitl}. The second inequality is because $\sigma$ is injective. The last inequality is because $C$ is a configuration.

			\item For all $G \in \cG$ it holds that $|A_i \cap G| \leq k(G)$. Let $G \in \cG$. Now, $$|A_i \cap G| = \left|\bigcup_{\ell \in A_i} \{\ell\} \cap G\right| = \bigcup_{\ell \in A_i} \left|\{\ell\} \cap G\right| \leq \bigcup_{\ell \in A_i} \left|\{\sigma(\ell)\} \cap G\right|  \leq    \left|\bigcup_{\ell \in C} \{\ell\} \cap G\right| = |C \cap G| \leq k(G).$$ The first inequality is because $\sigma$ is injective and that $G \cap M = \emptyset$; thus, by \eqref{fitl} for all $\ell \in A_i \cap G$ it holds that $\sigma(\ell) \in A_i \cap G$. The second inequality is because $\sigma$ is injective. The last inequality is because $C$ is a configuration. 
		\end{enumerate}

	\end{enumerate}
\end{proof}

\begin{claim}
	\label{eqeq2}
	$|\mathcal{H}| \leq \eps^{-22}Q^2(\eps)$ 
\end{claim} 

\begin{proof}
	$$|\mathcal{H}| = \left|\textsf{supp}(\bar{z}^*) \cup F\right| \leq  \left|\textsf{supp}(\bar{z}^*)\right| + \left| F\right| \leq |\textsf{supp}(\bar{z}^*)|+\eps^{-21}|\textsf{supp}(\bar{z}^*)|^2 \leq \eps^{-22}|\textsf{supp}(\bar{z})|^2 \leq  \eps^{-22}Q^2(\eps).$$
	
	The second inequality is because $\bar{z}^*$ holds the conditions of Lemma~\ref{FromPolytope} by Observation~\ref{ob:y}; thus, the inequality follows by Lemma~\ref{O(1)}. The third inequality is because $0<\eps<0.1$ and by Observation~\ref{ob:y}. The last inequality is because $\bar{z}$ holds the conditions of Lemma~\ref{FromPolytope}, that is, $\bar{z}$ is a good prototype.
\end{proof}

\begin{claim}
	\label{eqeq3}
	The size of $\mathcal{B}$ is at most $\|\bar{z}\| + \eps^{-22}|\textnormal{\textsf{supp}}(\bar{z})|^2$.
\end{claim} 

\begin{proof}
	The size of $\mathcal{B}$ is defined as the number of entries in $A^*$. Therefore, by \eqref{Astar} the size of $\mathcal{B}$ is at most:
	\begin{equation*}
		\begin{aligned}
			|F|+\sum_{C \in \textsf{supp}(\bar{z}^*)} \sum_{k \in [\bar{z}^*_C]}  ={} & |F|+\sum_{C \in \textsf{supp}(\bar{z}^*)} \bar{z}^*_C  \leq \|\bar{z}^*\| + \eps^{-21}|\textsf{supp}(\bar{z}^*)|^2\\ \leq{} & \|\bar{z}\| +|\textsf{supp}(\bar{z})|+  \eps^{-21}|\textsf{supp}(\bar{z})|^2 \leq  \|\bar{z}\| + \eps^{-22}|\textsf{supp}(\bar{z})|^2 \leq \|\bar{z}\| + \eps^{-22}Q^2(\eps).
		\end{aligned}
	\end{equation*} The first inequality is because $\bar{z}^*$ holds the conditions of Lemma~\ref{FromPolytope} by Observation~\ref{ob:y}; thus, the inequality follows by Lemma~\ref{O(1)}. The second inequality is by Observation~\ref{ob:y}. The last inequality is because $\bar{z}$ is a good prototype.

\end{proof}

\begin{claim}
	\label{partition:BC}
	$\{B_C\}_{C \in \mathcal{H}}$ is a partition of $\{A_i~|~i \in [m]\}$
\end{claim} 

\begin{proof}
	
	Let $i \in [m]$. By \eqref{Astar}, we split into two cases.  
	
	\begin{enumerate}
		\item $A_i \in F$.  Then, by the definition of $\mathcal{H}$ it holds that $A_i \in \mathcal{H}$; thus, by \eqref{BC} it follows that $A_i \in B_{A_i}$ and for all $C \in \mathcal{H} \setminus \{A_i\}$ it holds that $A_i \notin B_{C}$.
		
		\item  There are $C \in \mathcal{H}$ and $k \in [\bar{z}^*_{C}]$ such that $A_i = A_{C,k}$. Therefore, by \eqref{BC} it follows that $A_i \in B_{C}$ and for all $C' \in \mathcal{H} \setminus \{C\}$ it holds that $A_i \notin B_{C'}$.
	\end{enumerate} 
	
\end{proof}

Let $f = \{\ell~|~\{\ell\} \in F\}$ be all fractional items.  We use the following auxiliary claim.

\begin{claim}
	\label{fU}
	$A^*$ is a packing of $f \cup U$. 
\end{claim} 

\begin{proof}
	Let $\ell \in U \cup f$. We split into two cases.
	
	\begin{enumerate}
		\item $\ell \in f$. By \eqref{Astar}, there is $i \in [m]$ such that $A_i = \{\ell\}$.
		
		\item  $\ell \in U$. Since $M$ is a matching of size $|U|$ in $G$, there is $(C,j,k) \in V$ such that $(\ell,(C,j,k)) \in M$. Therefore, by the definition of the bin of $C$ and $k$ it holds that $\ell \in A_{C,k}$. Hence, by \eqref{Astar} there is $i \in [m]$ such that $A_i = A_{C,k}$ and the claim follows.
	\end{enumerate} Moreover, by \eqref{Astar} it follows that for any $\ell \in I \setminus (U \cup f)$ and $i \in [m]$ it holds that $\ell \notin A_i$.  
\end{proof}

\begin{claim}
	\label{eqeq4}
	$\{D_C\}_{C \in \mathcal{H}}$ is a partition of $I \setminus \bigcup_{i \in [m]} A_i$.
\end{claim} 

\begin{proof}
	
	Let $\ell \in I \setminus \bigcup_{i \in [m]} A_i$. By Claim~\ref{fU}, it follows that $I \setminus \bigcup_{i \in [m]} A_i = I \setminus (U \cup f)$ and we get $\ell \in I \setminus (U \cup f)$. Thus, as $\ell$ is not fractional and is not fully assigned to a slot type as it is not in $U$,  it follows that there is $C \in \textsf{supp}(\bar{z}^*)$ such that $\bar{\gamma}_{\ell,C} = 1$ by \eqref{F4} since $\bar{\gamma}$ is in the $\bar{z}^*$-polytope. In addition, since \eqref{F4} holds with equality, we get that there is exactly one such $C$. Since $\textsf{supp}(\bar{z}^*) \subseteq \mathcal{H}$, it follows that $\ell \in D_C$ and for all  $C' \in \mathcal{H} \setminus \{C\}$ it holds that $\ell \notin D_{C'}$. Hence, the claim follows. 
	
\end{proof}

\begin{claim}
	\label{eqeq5}
	For any $C \in \mathcal{H}$ and $A \in B_C$ it holds that $A$ is allowed in $C$
\end{claim} 

\begin{proof}
	
	By Claim~\ref{partition:BC}, there is $i \in [m]$ such that $A = A_i$. By \eqref{Astar}, we split into two cases.  
	
	\begin{enumerate}
		\item $A_i \in F$.  Then, by \eqref{BC} it holds that $A_i = C$ and $C$ is a singleton. Then, it trivially follows that $A_i$ is allowed in $A_i$, since for each item $\ell \in A_i$ it holds that $\ell \in \textsf{fit}(\ell)$ by \eqref{fitl}.  
		
		\item  There are $C \in \mathcal{H}$ and $k \in [\bar{z}^*_{C}]$ such that $A_i = A_{C,k}$. Therefore, by Lemma~\ref{lem:allowed} the claim follows. 
	\end{enumerate} 
	
\end{proof}

The proof  of Lemma~\ref{lem:assign} follows by Claim~\ref{eqeq1}, Claim~\ref{eqeq2}, Claim~\ref{eqeq3}, Claim~\ref{partition:BC}, Claim~\ref{eqeq4}, Claim~\ref{eqeq5}, and Claim~\ref{eqeq6}. \qed

\section{Omitted Proofs of Section~\ref{sec:alg_pack}}
\label{sec:PackProofs}

We use the following auxiliary claim. \begin{claim}
	\label{clm:fitifit}
	For all $C \in \mathcal{H}$ and $\ell \in D_C$ it holds that $s_C(\ell) < 0.1$.
\end{claim}

\begin{proof}
	$$s_C(\ell) = \frac{s(\ell)}{1-s(C)} \leq \frac{\eps \cdot (1-s(C))}{1-s(C)} = \eps < 0.1.$$
	 The first equality is by \eqref{IC}. The first inequality is because $\ell \in \textsf{fit}(C)$ by Definition~\ref{def:partition}; hence, the inequality follows by \eqref{fitS}. 
\end{proof}

\noindent{\bf Proof of Lemma~\ref{clm:GREEDY}:} Let $X = \textsf{tuple}(B_C),  Y = \textsf{Greedy}(\cI_C,\eps)$, $X = (X_1, \ldots, X_r)$, $Y = (Y_1, \ldots, Y_t)$, $S = D_C \cup \bigcup_{B \in B_C} B$, $Z = X+Y$, and $Z = (Z_1, \ldots, Z_p)$. We prove the following.

\begin{enumerate}
	\item $Z$ is a partition of $S$. Let $\ell \in S$. If $\ell \in D_C$, then there is $i \in [r]$ such that $\ell \in X_i$; therefore, by the definition of $Z$ it holds that $\ell \in Z_i$. Otherwise, there is $i \in [t]$ such that $\ell \in Y_i$; therefore, by the definition of $Z$ it holds that $\ell \in Z_i$.
	
	\item For all $i \in [p]$ it holds that $s(Z_i) \leq 1$. If $D_C = \emptyset$ then the claim is satisfied because $s(Z_i) = s(X_i) \leq s(C) \leq 1$; The first inequality is because $X_i \in B_C$ and thus $X_i$ is allowed in $C$; The second inequality is because $C$ is a configuration. Otherwise, If $i \leq t$ and $i \leq r$: $$s(Z_i) = s(X_i)+s(Y_i) \leq s(C)+s(Y_i) \leq s(C)+(1-s(C)) \cdot s_C(Y_i) \leq s(C)+1-s(C) = 1.$$ The first inequality is because by Definition~\ref{def:partition} it holds that $X_i$ is allowed in $C$. The second inequality is by \eqref{IC}. The third inequality is by Lemma~\ref{thm:greedy}; also, the conditions of Lemma~\ref{thm:greedy} are indeed satisfied by Claim~\ref{clm:fitifit}. Using similar arguments, we can show the claim for $i > t$ or $i >r$ by the definition of $Z$.
	
	\item For all $i \in [p]$ and $G \in \cG$ it holds that $|G \cap Z_i| \leq k(G)$. If $|G\cap D_C| = 0$ then the claim is trivially satisfied because $X_i$ is allowed in the configuration $C$. Otherwise, if $i \leq t$ and $i \leq r$: 	\begin{equation*}
		\begin{aligned}
		|G \cap Z_i|  \leq{} & |G \cap X_i| +|G \cap Y_i|  \leq |G \cap C|+|G \cap Y_i| \leq |G \cap C|+k_C(G \cap D_C) \\ \leq{} & |G \cap C|+\left(k(G)-|G \cap C|\right) = k(C).
		\end{aligned}
	\end{equation*}  The second inequality is because by Definition~\ref{def:partition} it holds that $X_i$ is allowed in $C$. The third inequality is by Lemma~\ref{thm:greedy}; also, the conditions of Lemma~\ref{thm:greedy} are indeed satisfied by Claim~\ref{clm:fitifit}. The fourth inequality is by \eqref{IC}. Using similar arguments, we can show the claim for $i > t$ or $i >r$ by the definition of $Z$.

\item It  holds that $p \leq (1+2\eps) \cdot |B_C|+2$. 	\begin{equation*}
	\begin{aligned}
	p ={} & \max\{r, t\} \leq \max\left\{|B_C|, (1+2\eps)\cdot   \max\left\{  s_C(D_C) , \max_{G \in \cG_{C}} \frac{|G|}{k_C(G)}   \right\} +2     \right\} \\ \\ \leq{} &  \max\left\{|B_C|, (1+2\eps)\cdot   \max\left\{  \frac{s(D_C)}{1-s(C)} , \max_{G \in \cG \text{ s.t. } G \cap D_C \neq \emptyset} \frac{|G\cap D_C|}{k(G)-|G \cap C|}   \right\} +2     \right\} \\ \\ \leq{} &  \max\left\{|B_C|, (1+2\eps)\cdot   \max\left\{  \frac{(1-s(C)) \cdot |B_C|}{1-s(C)} , \max_{G \in \cG \text{ s.t. } G \cap D_C \neq \emptyset} \frac{|B_C| \cdot (k(G) - |C \cap G|)}{k(G)-|G \cap C|}   \right\} +2     \right\} \\={} & (1+2\eps)\cdot |B_C|+2.
	\end{aligned}
\end{equation*} The first inequality is by Lemma~\ref{thm:greedy}; also, the conditions of Lemma~\ref{thm:greedy} are indeed satisfied by Claim~\ref{clm:fitifit}. The second inequality is by \eqref{IC}. The third inequality is by Definition~\ref{def:partition}. 

 \qed
\end{enumerate}

\noindent{\bf Proof of Lemma~\ref{lem:GREEDY}:} For all $C \in \mathcal{H}$ let $X^C = \textsf{tuple}(B_C),  Y^C = \textsf{Greedy}(\cI_C,\eps)$, $X^C = (X^C_1, \ldots, X^C_{r(C)})$, $Y^C = (Y^C_1, \ldots, Y^C_{t(C)})$, $S^C = D_C \cup \bigcup_{B \in B_C} B$, $Z^C = X^C+Y^C$, and $P_C = (Z^C_1, \ldots, Z^C_{p(C)})$. In addition, let $P = (P_1, \ldots, P_a)$ be the tuple returned by Algorithm~\ref{alg:pack}. We prove the following.

\begin{enumerate}
	\item $P$ is a partition of $I$: Let $\ell \in I$. Then, by Definition~\ref{def:partition} there is exactly one $C \in \mathcal{H}$ such that $\ell \in S^C$; then, by Lemma~\ref{clm:GREEDY} there is exactly one $i \in [p(C)]$ such that $\ell \in Z^C_i$; therefore, by Step~\ref{step:pPack} of Algorithm~\ref{alg:pack} there is exactly one  $j \in [a]$ such that $\ell \in P_j$ and the claim follows. 
	
	\item For all $i \in [a]$ it holds that $s(P_i) \leq 1$: By Step~\ref{step:pPack} of Algorithm~\ref{alg:pack} there is $C \in \mathcal{H}$ and $j \in [p(C)]$ such that $P_j = Z^C_j$. Therefore, $s(P_i) = s(Z^C_j) \leq 1$ by Lemma~\ref{clm:GREEDY}. 
	
	\item For all $i \in [a]$ and $G \in \cG$ it holds that $|G \cap P_i| \leq k(G)$:  By Step~\ref{step:pPack} of Algorithm~\ref{alg:pack} there is $C \in \mathcal{H}$ and $j \in [p(C)]$ such that $P_i = Z^C_j$. Therefore, $|P_i \cap G| = |Z^C_j \cap G| \leq k(G)$ by Lemma~\ref{clm:GREEDY}. 
	
	\item The size of $P$ that is $a$ is at most $(1+2\eps) \cdot m+  2 \cdot \eps^{-22} \cdot Q^2(\eps)$:
	$$a = \sum_{C \in \mathcal{H}} p(C) \leq \sum_{C \in \mathcal{H}} \left((1+2\eps)\cdot |B_C|+2\right) \leq (1+2\eps) \cdot m+  2 \cdot \eps^{-22} \cdot Q^2(\eps)$$ The first equality is by Step~\ref{step:pPack} of Algorithm~\ref{alg:pack}. The first inequality is by Lemma~\ref{clm:GREEDY}. The second inequality is because $\{B_C\}_{C \in \mathcal{H}}$ is a partition of $\{A_i~|~i \in [m]\}$ and that $|\mathcal{H}| \leq  \eps^{-22} \cdot Q^2(\eps)$ by Definition~\ref{def:partition}.

	 \qed
\end{enumerate}

\section{Omitted Proofs of Section~\ref{sec:reduction}}
\label{sec:reductionProofs}

In this section we refer to a BPP instance $\cj$ fixed in Section~\ref{sec:reduction} and a given parameter $\eps \in (0,0.1]$; in addition, we use $\cI = \textsf{Reduce}(\eps,\cj)$ as in Section~\ref{sec:reduction}. 
\subsection{Proofs From the Reduction}

\begin{lemma}
	\label{lem:k_val}
	For every minimal pivot $w\in \{2,\ldots, \eps^{-1}+1\}$ of $\mathcal{J}$ it holds that $s(I_w) \leq \eps \cdot \OPT(\mathcal{J})$. 
\end{lemma}

\begin{proof}

	

	$$ s(I_w) \leq \eps \cdot s\left( \bigcup_{i \in \{2,\ldots, \eps^{-1}+1\}} I_i \right) \leq \eps \cdot s(I) \leq \eps \cdot \OPT(\cj).$$

	The first inequality is because $I_2, \ldots, I_{\eps^{-1}+1}$ are $\eps^{-1}$ disjoint subsets of items and by \eqref{eq:k} it holds that $I_w$ is an interval of minimum total size. The last inequality is because each bin may contain total size at most $1$; thus, the total size of items is a lower bound to the optimum. 

\end{proof}

\noindent{\bf Proof of Lemma~\ref{lem:few_large_groups}:} 
	
	If $\OPT(\cj) = 0$ there are no items and the claim trivially follows. Otherwise, let $$\cB = \big\{G \in \cG~\big|~|G \cap  (H_w \cup I_w)| \geq \eps^{2w+4}\big\}$$ be the subset of groups that contain at least $\eps^{2w+4}$ $w$-heavy and $w$-medium items. The following holds for all $G \in \cB$:
	
	\begin{equation}
		\label{eq:feww}
		\begin{aligned}
			s(G) \geq{} & s\left(G \cap (H_w \cup I_w)\right) = \sum_{\ell \in G \cap (H_w \cup I_w)} s(\ell) \geq \sum_{\ell \in G \cap (H_w \cup I_w)} \eps^{w+1} =  \eps^{w+1} \cdot |G \cap (H_w \cup I_w)|  \\ \geq{} &  \eps^{w+1} \cdot \eps^{2w+4} \cdot \OPT(\cj) = \eps^{3w+5} \cdot \OPT(\cj). 
		\end{aligned}
	\end{equation} The second inequality is because the size of a $w$-medium item is at least $\eps^{w+1}$. The third inequality is because $G \in \cB$. Therefore, 
	
	$$|\cB| = \sum_{G \in \cB} 1 \leq   \sum_{G \in \cB} \frac{s(G)}{\eps^{3w+5} \cdot \OPT(\cj)} \leq  \frac{s(I)}{\eps^{3w+5} \cdot \OPT(\cj)} \leq  \frac{\OPT(\cj)}{\eps^{3w+5} \cdot \OPT(\cj)} = \eps^{-3w-5}.$$ The first inequality is by \eqref{eq:feww}. \qed

	The next two lemmas combined are essentially the proof of Condition~1 in Lemma~\ref{lem:reductionReconstruction}.

	\begin{lemma}
		\label{lem:shiftingNotIncreaseOPT2}
		For $\cI = \textsf{Reduce}(\eps,\cj)$ it holds that $OPT(\mathcal{I}) \leq OPT({\cal J})$.
	\end{lemma}

	\begin{proof}
		Let $A = (A_1, \ldots, A_m)$ be a packing of $\cj$ and let $w \in \{2, \ldots, \eps^{-1}+1\}$ chosen in Step~\ref{step:pivot} of Algorithm~\ref{alg:reduction}. We show that $A$ is also a packing of $\cI$ which concludes the proof. By the definition of packing, we prove the following. 
		
		\begin{itemize}
			
			\item It holds that $A$ is a partition of $I$. Follows because $A$ is a packing of $\cj$ and because that by Step~\ref{Ir} the sets of items in $\cI,\cj$ are identical. 
			
			\item For all $i \in [m]$ it holds that $s(A_i) \leq 1$. By Step~\ref{Ir} the set of items and the size functions in $\cI,\cj$ are identical and the proof follows.

			\item For all $G \in \cG_R$ and $i \in [m]$ it holds that $|G \cap A_i| \leq k_R(G)$. We split into three cases by Step~\ref{kr}. 
			
			\begin{enumerate}
				\item $G \in \cG_L(w)$. Then, $|G \cap A_i| \leq k(G) = k_R(G)$ where the inequality is because $A$ is a packing of $\cj$ and the equality is by Step~\ref{kr}. 
				
				\item $G = G' \setminus H_w \text{ s.t. } G' \in \cG \setminus \cG_L(w)$. Then, 
				
				$$|G \cap A_i| \leq |(G \cup (G' \cap H_w)) \cap A_i| = |G' \cap A_i| \leq k(G') = k_R(G).$$  The second inequality is because $A$ is a packing of $\cj$. The last equality is by Step~\ref{kr}. 
				
				\item $G = \Gamma_w$. Then, $|G \cap A_i| = |\Gamma_w \cap A_i| \leq |\Gamma_w| < |\Gamma_w+1| = k_R(\Gamma_w) = k_R(G)$ where the second equality is by Step~\ref{kr}. 
			\end{enumerate}
		\end{itemize}
	\end{proof}

	\begin{lemma}
		\label{lem:rounding12}
		$\cI$ is $\eps$-structured. 
	\end{lemma}

	\begin{proof}
		Let $w \in \{2, \ldots, \eps^{-1}+1\}$ chosen in Step~\ref{step:pivot} of Algorithm~\ref{alg:reduction}. Therefore, 
		\begin{equation*}
			\begin{aligned}
				\big|\{G{} & \in \cG_R~|~ \exists \ell \in G \text{ s.t. } s(\ell) \geq \eps^2\}\big| \leq \big|\{G \in \cG_R~|~ G \cap H_w \neq \emptyset\}\big| \leq |\cG_L(w)|+|\{\Gamma_w\}|  \\ \leq{} &   \eps^{-3w-5}+1 \leq \eps^{-3w-6} \leq \eps^{-3\eps^{-1}-3-6} \leq  \eps^{-4\eps^{-1}} \leq  \eps^{-\eps^{-2}} =  K(\eps).
			\end{aligned}
		\end{equation*} The first inequality is because for all $\ell \in I \setminus H_w$ it holds that $s(\ell) \leq \eps^2$ because $w \geq 2$ by \eqref{eq:k}. 
	\end{proof}

	\subsection{Proofs From the Reconstruction}
	
	For the following, let $A = (A_1, \ldots, A_m)$ be a packing of $\cI$ and let $F_A = (U_1, \ldots, U_{m'})$ (see \eqref{F_A}). In addition, for all $w \in \{2,\ldots, \eps^{-1}+1\}$ let $Y_w = I \setminus (H_w \cup I_w)$.

\begin{claim}
	\label{claim:tu2}
	For all $ i\in [m']$ and $G \in \cG$ it holds that $|G \cap Y_w\cap U_i| \leq k(G)$.
\end{claim}

\begin{proof}
If $i \in [m'] \setminus [m]$, then by \eqref{F_A} it holds that $|G \cap Y_w\cap U_i| =0 \leq k(G)$. Otherwise, we split into two cases by Step~\ref{step:groups} of Algorithm~\ref{alg:reduction}.

	\begin{enumerate}
		\item $G \in \cG_{L}(w)$. Then, \begin{equation*}
			\label{eq:FAHelp}
			|G \cap Y_w\cap U_i| = |G \cap Y_w \cap \left(( A_i \setminus (\Gamma_w \cup \Omega_w)) \cup B_i\right)| \leq |G \cap Y_w \cap A_i| \leq |G \cap A_i| \leq k_R(G) = k(G).
		\end{equation*} The first inequality is because $(\Omega_w \cup \Gamma_w) \cap Y_w = \emptyset$ by \eqref{Iu}, \eqref{X:X}, and because $B_i \subseteq \Gamma_w$ by Lemma~\ref{lem:Bgreedy}. The third inequality is because $A$ is a packing of $\cI$. The last equality is by Step~\ref{kr} of Algorithm~\ref{alg:reduction}. 
		
		\item $G  \in \cG \setminus \cG_{L}(w)$. Then, by \eqref{Is} it holds that $G \setminus H_w \in \cG_{S}(w)$. Therefore, \begin{equation*}
			|G \cap Y_w\cap U_i| = |G \cap Y_w \cap \left(( A_i \setminus (\Gamma_w \cup \Omega_w)) \cup B_i\right)| \leq |G \cap Y_w \cap A_i| \leq |G \cap A_i| \leq k_R(G \setminus H_w) = k(G).
		\end{equation*}  The first inequality is because $B_i \subseteq \Gamma_w$ by Lemma~\ref{lem:Bgreedy}. The third inequality is because $A$ is a packing of $\cI$. The last equality is by Step~\ref{kr} of Algorithm~\ref{alg:reduction}. 
		
	\end{enumerate}

\end{proof}
\begin{claim}
	\label{claim:tu}
	For all $ i\in [m']$ and $G \in \cG$ it holds that $|F_A(i,G)| \leq |G \cap U_i \cap H_w|$. 
\end{claim}

\begin{proof}
	
	We split into two cases by Step~\ref{step:groups} of Algorithm~\ref{alg:reduction}.

\begin{enumerate}
	\item $G \in \cG_{L}(w)$.

		\begin{equation*}
		\begin{aligned}
	|F_A(i,G)| =  \max\{|G \cap U_i|-k(G),0\} \leq{} &  \max\{k_R(G)-k(G),0\}  =  \max\{k(G)-k(G),0\} = 0.
		\end{aligned}
	\end{equation*}
	
	The first equality is because for any bin $i \in [m']$ and group $G \in \cG$ it holds that $F_A(i,G)$ is an exclusion-minimal subset of $w$-light items from $U_i$ such that $| (U_i \setminus F_A(i,G) )  \cap G| \leq k(G)$. The first inequality is because $A$ is a packing of $\cI$. The second equality is by Step~\ref{kr} of Algorithm~\ref{alg:reduction}.  
	
	\item $G  \in \cG \setminus \cG_{L}(w)$. Then, 	\begin{equation*}
		\begin{aligned}
			|F_A(i,G)| = {} & \max\{ |G \cap U_i|-k(G),0\} = \max\{ |G \cap U_i \cap H_w|+|G \cap U_i \cap Y_w|-k(G),0\} \\ \leq{} & \max\{ |G \cap U_i \cap H_w|+k(G)-k(G),0\} = |G \cap U_i \cap H_w|.
		\end{aligned}
	\end{equation*} The first equality is because for any bin $i \in [m']$ and group $G \in \cG$ it holds that $F_A(i,G)$ is an exclusion-minimal subset of $w$-light items from $U_i$ such that $| (U_i \setminus F_A(i,G) )  \cap G| \leq k(G)$. The second equality is because $\Omega_w \cap U_i = \emptyset$ by \eqref{X:X}, \eqref{F_A}, and because \textsf{Fill}(A) is a partition of items in $\Gamma_w$ by Lemma~\ref{lem:Bgreedy}. The inequality is by Claim~\ref{claim:tu2}.
	
\end{enumerate}

\end{proof}

\begin{claim}
	\label{du:help}
	For all $G \in \cG$ it holds that $|D(A) \cap G| \leq \eps^{2w+4} \OPT(\cj)$. 
\end{claim}

\begin{proof}
	We split into three cases by Step~\ref{kr} of Algorithm~\ref{alg:reduction}.

	\begin{enumerate}
		\item $G \in \cG_{L}(w)$. Then, we use the following inequality.
		
		\begin{equation}
			\label{DAS}
			|G \cap U_i| = |G \cap \left(( A_i \setminus (\Gamma_w\cup \Omega_w) ) \cup B_i\right)| = |G \cap A_i| \leq k_R(G) = k(G).
		\end{equation} 
		
		The second equality is because $(\Gamma_w \cup \Omega_w) \cap G = \emptyset$ by \eqref{Iu} and $B_i \subseteq \Gamma_w$ by Lemma~\ref{lem:Bgreedy}. Therefore,

		$$|D(A) \cap G| = \sum_{i \in [m']} |F_A(i,G)| = \sum_{i \in [m']}  \max\left\{ |G \cap U_i| - k(G),0 \right\} = 0.$$ The second equality is because for any bin $i \in [m']$ and group $G \in \cG$ it holds that $F_A(i,G)$ is an exclusion-minimal subset of $w$-light items from $U_i$ such that $| (U_i \setminus F_A(i,G) )  \cap G| \leq k(G)$. The last equality is by \eqref{DAS}.

		\item $G  \in \cG \setminus \cG_{L}(w)$. Now,
		
		\begin{equation*}
			\begin{aligned}
				|D(A) \cap G| ={} & \left|\bigcup_{i \in [m']} F_A(i,G)\right| \leq \sum_{i \in [m']} |G \cap U_i \cap H_w|  = |G \cap H_w| \leq \eps^{2w+4} \OPT(\cj).
			\end{aligned}
		\end{equation*}
		
		The first inequality is by Claim~\ref{claim:tu}. The second equality is by \eqref{F_A} and because $A$ is a packing of $\cI$. The last inequality is by Lemma~\ref{lem:few_large_groups}. 
		
	\end{enumerate}

\end{proof}

\begin{claim}
	\label{sd:help}
	For all $i \in [m']$ and $G \in \cG$ it holds that $s(F_A(i,G)) \leq \eps \cdot s(U_i \cap H_w \cap G)$. 
\end{claim}

\begin{proof}
	
	If $F_A(i,G) = \emptyset$ then the claim is trivially satisfied. Otherwise, by Claim~\ref{claim:tu} it holds that $U_i \cap H_w \cap G \neq \emptyset$. Therefore, 
	\begin{equation*}
		\begin{aligned}
			s(F_A(i,G)) ={} & \sum_{\ell \in F_A(i,G)} s(\ell) \leq \sum_{\ell \in F_A(i,G)} \max_{\ell \in F_A(i,G)} s(\ell) =  |F_A(i,G)|  \cdot \max_{\ell \in F_A(i,G)} s(\ell) \leq |U_i \cap H_w \cap G|  \max_{\ell \in F_A(i,G)} s(\ell) \\ \leq{} &  |U_i \cap H_w \cap G| \cdot  \eps \cdot \min_{\ell \in U_i \cap H_w \cap G} s(\ell) \leq  \eps \cdot \sum_{\ell \in U_i \cap H_w \cap G} s(\ell) = \eps \cdot s(U_i \cap H_w \cap G).
		\end{aligned}
	\end{equation*} The second inequality is by Claim~\ref{claim:tu}. The third inequality is because the size of the largest item in $Y_w$ is at least $\eps^{-1}$ times smaller compared to the smallest item in $H_w$ by \eqref{Interval}.

\end{proof}

\begin{lemma}
	\label{lem:s(D)}
	$s(D(A)) \leq \eps \cdot \OPT(\cj)$.
\end{lemma}

\begin{proof}
	
	\comment{
	
	\begin{equation*}
		\begin{aligned}
			s(D(A)) ={} & \sum_{G \in \cG} \sum_{i \in [m']} s(F_A(i,G)) \leq \sum_{G \in \cG} \sum_{i \in [m']} \eps \cdot s(U_i \cap H_w \cap G) \leq  \eps \cdot \sum_{G \in \cG} \sum_{i \in [m']}  \cdot s(U_i \cap G) \\ ={} & \eps \cdot \sum_{G \in \cG} \sum_{i \in [m']}  \left( s(A_i \cap G)-s((\Gamma_w \cup \Omega_w) \cap A_i)-s(B_i \cap G) \right) \\ \leq{} & \eps \cdot \left( s(I)-s(\Gamma_w)+s\left(\bigcup_{b \in [m]} B_i\right)\right) \leq \eps \cdot s(I) \leq \eps \cdot \OPT(\cj).
		\end{aligned}
	\end{equation*} 
}

	\begin{equation*}
	\begin{aligned}
		s(D(A)) ={} & \sum_{G \in \cG} \sum_{i \in [m']} s(F_A(i,G)) \leq \sum_{G \in \cG} \sum_{i \in [m']} \eps \cdot s(U_i \cap H_w \cap G) \leq  \eps \cdot \sum_{G \in \cG} \sum_{i \in [m']}  \cdot s(U_i \cap G) \\ \leq{} & \eps \cdot s(I) \leq \eps \cdot \OPT(\cj).
	\end{aligned}
\end{equation*} 
The first equality is by \eqref{D}. The first inequality is by Claim~\ref{sd:help}. 
\end{proof}

\noindent{\bf Proof of Lemma~\ref{thm:D}:} The proof follows by Claim~\ref{lem:s(D)} and Claim~\ref{du:help}.

\begin{lemma}
	\label{lem:shiftingCanBeUsedForI2}
	Given a packing $A$ for $\mathcal{I} = \textsf{Reduce}(\eps, \mathcal{J})$ of size $m$, Algorithm~\ref{Alg:RECON} returns in time $\textnormal{poly}(|\cj|, \frac{1}{\eps})$ a packing $F$ for the instance ${\cal J}$ of size $\phi$ where $\phi<m+13\eps\cdot \OPT(\cj)+1$. 
\end{lemma}

\begin{proof}
	
	In the proof of the lemma we use several auxiliary claims. For simplicity, by Step~\ref{step:3}, let $F_A \setminus E = (Z_1 \ldots Z_{m'})$, $\textsf{Greedy}(\mathcal{E}) = (X_1 \ldots X_{x})$,  $F = \left( F_A \setminus E \right) \oplus \textsf{Greedy}(\mathcal{E})$, and $F = (F_1, \ldots, F_{\phi})$. Also, for the following let $(B_1, \ldots, B_m, R) = \textsf{Fill}(\mathcal{J}, A)$. Finally, recall the definition of $\mathcal{\cI}$-partition from Section~\ref{sec:reduction}. For $\beta = \eps^{-w-2}$, we define $P_1,\ldots, P_q$ to be the $\mathcal{\cI}$-partition: a partition of $\Gamma_w$ by a non decreasing order of item sizes such that for all $1\leq i \leq j<q$ and $\ell \in P_i, y \in P_j$ it holds that $s(\ell) \geq s(y)$, $|P_i| = \ceil{\frac{|\Gamma_w|}{\beta}}$, and $P_{q}$ contains the remaining items from $\Gamma_w$. It follows that $q \leq \beta$ and we define $q(\cI) = q$.
	
	\begin{claim}
		\label{lem:F}
		$F$ is a packing of $\cj$.
	\end{claim}
	
	\begin{proof}
		We prove the necessary conditions of packing as follows. 
		\begin{enumerate}
			\item $F$ is a partition of $I$.  Let $\ell \in I$. By the following three complementary cases, we conclude that $F$ is a partition of $I$.
			
			\begin{itemize}
				\item $\ell \in E$. Then, there is exactly one $j \in [x]$  such that $\ell \in X_j$ by Lemma~\ref{thm:greedy} and \eqref{Ie}; hence,  there is exactly one $i \in [\phi]$ such that $\ell \in F_i$ because by Step~\ref{step:3} of Algorithm~\ref{Alg:RECON} for all $i \in [m']$ it holds that $\ell \notin Z_i$. %
				
				\item $\ell \in \Gamma_w$. Then, there is exactly one $i \in [m]$ such that $\ell \in A_i$ because $A$ is a packing of $\cI$. Then, by Lemma~\ref{lem:Bgreedy}  and by \eqref{F_A} there is exactly one $i \in [m']$ such that $\ell \in Z_i$. Hence,  there is exactly one $i \in [\phi]$ such that $\ell \in F_i$ because by \eqref{Ie} for all $j \in [x]$ it holds that $\ell \notin X_j$.   
				
				\item $\ell \in I \setminus (\Gamma_w \cup E)$. Then, there is exactly one $i \in [m]$ such that $\ell \in A_i$ because $A$ is a packing of $\cI$. Then, by \eqref{F_A} there is exactly one $i \in [m']$ such that $\ell \in Z_i$. Hence,  there is exactly one $i \in [\phi]$ such that $\ell \in F_i$ because by \eqref{Ie} for all $j \in [x]$ it holds that $\ell \notin X_j$.   
			\end{itemize}

			\item Let $i \in [\phi]$. By Step~\ref{step:reconl} of Algorithm~\ref{Alg:RECON} it holds that $F_i = Z_i$ or $F_i = X_j$ for $j = i-m'$. We split into two cases based on the latter distinction.

			\begin{itemize}
				\item $F_i = Z_i$. 	Then, by \eqref{F_A} if $U_i = \{\ell\}$ for $\ell \in R$ it holds that $s(F_i) = s(Z_i) = s(\ell) \leq 1$ because the sizes of items are bounded by $1$. Otherwise, by \eqref{F_A} it holds that $Z_i = ((A_i \setminus (\Gamma_w\cup \Omega_w)) \cup B_i) \setminus E$. Therefore, we use the following inequalities. 
				
				\begin{equation}
					\label{XW}
					\begin{aligned}
						s(B_i) ={} & \sum_{j \in \{2, \ldots, q(\cI)\}~} \sum_{\ell \in B_i \cap P_j} s(\ell) \leq \sum_{j \in \{2, \ldots, q(\cI)\}~} \sum_{\ell \in B_i \cap P_j} \max_{\ell \in P_j} s(\ell) \\ \leq{} & \sum_{j \in \{2, \ldots, q(\cI)\}~} |B_i \cap P_j| \cdot \max_{\ell \in P_j} s(\ell). 
					\end{aligned}
				\end{equation} The first equality is because $B_i \cap P_1 = \emptyset$ by Lemma~\ref{lem:Bgreedy}. In addition, 
				
				\begin{equation}
					\label{XW2}
					\begin{aligned}
						\sum_{j \in \{2, \ldots, q(\cI)\}~}{} & |B_i \cap P_j| \cdot \max_{\ell \in P_j} s(\ell) \leq \sum_{j \in \{2, \ldots, q(\cI)\}~} |A_i \cap P_{j-1}| \cdot \max_{\ell \in P_{j}} s(\ell)  \\ \leq{} & \sum_{j \in \{2, \ldots, q(\cI)\}~} |A_i \cap P_{j-1}| \cdot \min_{\ell \in P_{j-1}} s(\ell) \\ \leq{} & \sum_{j \in \{2, \ldots, q(\cI)\}~} \sum_{\ell \in A_i \cap P_{j-1}} s(\ell) \leq s(A_i \cap \Gamma_w). 
					\end{aligned}
				\end{equation} The first inequality is by Lemma~\ref{lem:Bgreedy}. The third inequality is because for all $j \in [q(\cI)-1]$, $\ell \in P_j, y \in P_{j+1}$ it holds that $s(\ell) \geq s(y)$ by the definition of the $\cI$-partition.  Then,

				\begin{equation*}
					\begin{aligned}
						s(F_i) ={} & s(U_i) = s\big(((A_i \setminus (\Gamma_w \cup \Omega_w)) \cup B_i) \setminus E\big) \\ \leq{} & s\big((A_i \setminus (\Gamma_w \cup \Omega_w) \cup B_i)\big)   \\ ={} & s(A_i \setminus (\Gamma_w \cup \Omega_w))+s(B_i) \\ \leq{} & s(A_i \setminus (\Gamma_w \cup \Omega_w))+s(A_i \cap \Gamma_w) \leq s(A_i) \leq 1. 
					\end{aligned}
				\end{equation*} The second inequality is by \eqref{XW} and \eqref{XW2}. The last inequality is since $A$ is a packing of~$\cI$.

				\item $F_i =X_j$ for $j = i-m'$. Then, It holds that $s(F_i) \leq 1$ by Lemma~\ref{thm:greedy}.
			\end{itemize}

			\item Let $G \in \cG$ and $i \in [\phi]$. We split into two cases by Step~\ref{step:reconl}.  \begin{itemize}
				\item $F_i = X_j$ for $j = i-m'$. Then, $$|F_i \cap G| = |X_j \cap G| = |X_j \cap E \cap G| \leq k_E(G \cap E) = k(G).$$ The second and the last equalities are by \eqref{Ie}. The inequality is by Lemma~\ref{thm:greedy}.

				\item $F_i = Z_i$. Then, we split into two complementary cases. If $G \in \cG_{L}(w)$:

				\begin{equation*}
					\begin{aligned}
						|F_i \cap G| = |Z_i \cap G| = |((A_i \setminus (\Gamma_w \cup \Omega_w)) \cup B_i) \cap G| = |A_i \cap G| \leq k_R(G) = k(G). 
					\end{aligned}
				\end{equation*} The second equality is by \eqref{F_A}. The third equality is because $G \in \cG_{L}(w)$; hence, $G \cap \Gamma_w = \emptyset$ and $G \cap B_i \cap \Omega_w = \emptyset$ by \eqref{Iu} and \eqref{X:X}, respectively. The inequality is because $A$ is a packing of $\cI$. The last equality is by Step~\ref{kr} of Algorithm~\ref{alg:reduction}. Otherwise, $G \in \cG \setminus \cG_{L}(w)$. Then,  	\begin{equation}
					\label{XW3}
					\begin{aligned}
						|F_i \cap G| \leq{} & |F_i \cap G \cap H_w|+|F_i \cap G \cap Y_w| \\ \leq{} &  |F_i \cap G \cap H_w|+\max\{0, k(G)-|F_i \cap G \cap H_w|\}  \\ ={} & \max\{ |F_i \cap G \cap H_w|,  k(G)\} \leq \max\{  |B_i \cap G|,  k(G)\} \leq k(G).
					\end{aligned}
				\end{equation}
				
				The  first equality is because $\Omega_w \cap  G \cap U_i = \emptyset$ by \eqref{F_A} since $G \in \cG \setminus \cG_{L}(w)$. The second inequality is by \eqref{D}: we discard to $F_A(i,G)$ only $w$-light items until reaching at most $k(G)$ items in $U_i$. The last equality is by \eqref{F_A} and because $G \in \cG \setminus \cG_L(w)$. The last inequality is by Lemma~\ref{lem:Bgreedy}. 
			\end{itemize}
		\end{enumerate}
		
	\end{proof}

	\begin{claim}
		\label{claim:IESize}
		$s(E) \leq 2\eps \cdot \OPT(\cj)$. 
	\end{claim}

	\begin{proof}
		
		\begin{equation*}
			\begin{aligned}
				s(E) \leq{} & s(I_w)+s(D(A)) \leq 2 \eps \cdot \OPT(\cj).
			\end{aligned}
		\end{equation*} The first inequality is by \eqref{Ie}. The second inequality is by Lemma~\ref{lem:k_val} and  Lemma~\ref{thm:D}.
	
	\end{proof}

	\begin{claim}
		\label{claim:IESize3}
		For all $G \in \cG$ it holds that $|G \cap E| \leq2\eps \cdot \OPT(\cj)$. 
	\end{claim}
	\begin{proof}
		
		We split into two cases. 
		
		\begin{enumerate}
			\item $G \in \cG_{L}(w)$. Then,

			\begin{equation*}
				\begin{aligned}
					|G \cap E| ={} & |G \cap D(A)| = \left| \bigcup_{i \in [m']} F_A(i,G)\right| = \sum_{i \in [m']} \max\{|G \cap U_i|-k(G),0\} \\ \leq{} & \sum_{i \in [m']} \max\{k_R(G)-k(G),0\}  = \sum_{i \in [m']} \max\{k(G)-k(G),0\} = 0.
				\end{aligned}
			\end{equation*}
			
			The first equality is by \eqref{Ie}. The first inequality is because $A$ is a packing of $\cI$ and by \eqref{F_A}. The fourth equality is by Step~\ref{kr} of Algorithm~\ref{alg:reduction}.

			\item $G \in \cG \setminus \cG_{L}(w)$. We use the following inequality.

			\begin{equation}
				\label{eq:||d}
				|G \cap E \cap Y_w| = |G \cap D(A)| \leq \eps \cdot \OPT (\cj).
			\end{equation}
			
			The equality is by \eqref{D} and \eqref{Ie}. The inequality is by Lemma~\ref{thm:D}. Therefore, 
			
			$$|G \cap E| = |G \cap E \cap Y_w|+|G \cap E \cap (H_w \cup I_w)| \leq  2\eps \cdot \OPT (\cj).$$

			The inequality is by Lemma~\ref{lem:few_large_groups} and by \eqref{eq:||d}. 
		\end{enumerate}
	\end{proof}

	\begin{claim}
		\label{claim:IESize2}
		$V(\mathcal{E}) \leq  2\eps \cdot \OPT (\cj)$. 
	\end{claim}
	\begin{proof}
		
		$$V(\mathcal{E}) = \max_{G \in \cG_E} \ceil{\frac{|G|}{|k_E(G)|}} \leq  \max_{G \in \cG_E} |G| =  \max_{G \in \cG} |G \cap E| \leq   2\eps \cdot \OPT (\cj).$$ The first inequality is because for all $G \in \cG_E$ it holds that $k_E(G) \geq 1$. The second equality is by \eqref{Ie}. The last inequality is by Claim~\ref{claim:IESize3}. 
		
	\end{proof}

	\begin{claim}
		\label{lem:sizeF}
		$\phi<m+13\eps\cdot \OPT(\cj)+1$
	\end{claim}
	
	\begin{proof}

		\begin{equation*}
			\begin{aligned}
				\phi = {} & m'+x \leq m+|R|+x \leq m+\eps\cdot\OPT(\cj)+1+(1+2) \cdot \max\{2s(E),s(E)+V(\mathcal{E})\} \\ \leq{} & m+13\eps\cdot\OPT(\cj)+1.
			\end{aligned}
		\end{equation*}

		The first equality is by Step~\ref{step:reconl} of Algorithm~\ref{Alg:RECON}. The first inequality is by \eqref{F_A}. The second inequality is by Lemma~\ref{lem:Bgreedy} and by Lemma~\ref{thm:greedy}. The last inequality is by Claim~\ref{claim:IESize} and Claim~\ref{claim:IESize2}. 
	\end{proof}
	
	\begin{claim}
		\label{claim:reconRuningTime}
		The running time of Algorithm~\ref{Alg:RECON} is $\textnormal{poly}(|\cj|,\frac{1}{\eps})$. 
	\end{claim}
	
	\begin{proof}
		By Lemma~\ref{lem:Bgreedy}, Step~\ref{step:1} can be computed in time $\textnormal{poly}(|\cj|,\frac{1}{\eps})$. By \eqref{Ie}, the discarded instance in Step~\ref{step:2} can be computed in time $\textnormal{poly}(|\cj|,\frac{1}{\eps})$. By Lemma~\ref{thm:greedy}, Step~\ref{step:3} can be computed in time $\textnormal{poly}(|\cj|,\frac{1}{\eps})$; also, concatenating two tuples takes linear time in the size of the instance. The overall running time is therefore $\textnormal{poly}(|\cj|,\frac{1}{\eps})$. 
	\end{proof}
	
	The proof of Lemma~\ref{lem:shiftingCanBeUsedForI2} follows by Claim~\ref{lem:F}, Claim~\ref{lem:sizeF}, and Claim~\ref{claim:reconRuningTime}.
\end{proof}

\posA
\noindent{\bf Proof of Lemma~\ref{lem:reductionReconstruction}:}

Condition~\ref{condition:reduction1} holds By Lemma~\ref{lem:shiftingNotIncreaseOPT2} and Lemma~\ref{lem:rounding12}. Condition~\ref{condition:reduction2} holds By Lemma~\ref{lem:shiftingCanBeUsedForI2}. \qed

\subsection{Algorithm \textsf{Fill}}
\label{sec:fill}

In this section we prove Lemma~\ref{lem:Bgreedy}. Recall the definition of $\mathcal{\cI}$-partition from Section~\ref{sec:reduction}. For $\beta = \eps^{-w-2}$, we define $P_1,\ldots, P_q$ to be the $\mathcal{\cI}$-partition: a partition of $\Gamma_w$ by a non decreasing order of item sizes such that for all $1\leq i \leq j<q$ and $\ell \in P_i, y \in P_j$ it holds that $s(\ell) \geq s(y)$, $|P_i| = \ceil{\frac{|\Gamma_w|}{\beta}}$, and $P_{q}$ contains the remaining items from $\Gamma_w$. It follows that $q \leq \beta$ and we define $q(\cI) = q$. In Algorithm \textsf{Fill} we start with $B_1, \ldots, B_m$ empty sets and with $R = \Gamma_w$. In each iteration, we try to move an item from $R$ to some $B_i, i \in [m]$ if it is possible to do so without violating one of the conditions of the lemma. When the above cannot be done anymore, we return $(B_1, \ldots, B_m,R)$. The pseudocode of Algorithm \textsf{Fill} is given in Algorithm~\ref{alg:fill}. For the proof, we use the following claims. Let $(B_1, \ldots, B_m, R) = \textsf{Fill}(\cj,A)$. 

\begin{algorithm}[h]
	\caption{$\textsf{Fill}(\mathcal{J}, A = (A_1, \ldots, A_m))$}
	\label{alg:fill}

	Initialize $R \leftarrow \Gamma_w$,  $B_i \leftarrow \emptyset ~:~ \forall i \in [m]$.\label{fill:init}
	
	\While{$\exists i \in [m], j \in \{2,\ldots, q(\cI)\}, \ell \in R \cap P_j \textnormal{ s.t. } |B_i \cap P_j| < |A_i \cap P_{j-1}| ~\textnormal{and}~ |\textnormal{\textsf{group}}(\ell) \cap  B_i| < k(G)$ \label{fill:while}}{

			$B_i \leftarrow B_i \cup \{\ell\}$.\label{fill:1}
			
			$R \leftarrow R \setminus \{\ell\}$.\label{fill:2}

	}
	
	Return $(B_1, \ldots, B_m, R)$.\label{fill:return}
\end{algorithm}

\begin{claim}
	\label{claim:fillRunningTime}
	The running time of Algorithm~\ref{alg:fill} is $\textnormal{poly}(|\cj|, \frac{1}{\eps})$.
\end{claim}

\begin{proof}
	
	By Step~\ref{fill:while}, Step~\ref{fill:1}, and Step~\ref{fill:2}, each iteration of the while loop of the algorithm takes linear time by trying at most all items. In addition, there are $\textnormal{poly}(|\cj|, \frac{1}{\eps})$ iterations of the while loop because by Step~\ref{fill:while}, Step~\ref{fill:1}, and Step~\ref{fill:2} in each iteration an item is moved from $R$ to some $B_i, i \in [m]$ and never moved again. Therefore, there are at most $|I| = \textnormal{poly}(|\cj|, \frac{1}{\eps})$ iterations and the claim follows. 
	
\end{proof}

\begin{claim}
	\label{claim:fill1}
	For each $k \in \mathbb{N}$, after $k$ iteration of Step~\ref{fill:while} the following hold.
	
	\begin{enumerate}
		
		\item  $(B_1, \ldots, B_m, R)$ is a partition of $\Gamma_w$.  
		
		\item For all $i \in [m]$ and $j \in \{2, \ldots,q(\cI)\}$ it holds that $|B_i \cap P_j| \leq |A_i \cap P_{j-1}|$ and $|B_i \cap P_1| = 0$.
		
		\item For all $i \in [m]$ and $G \in \cG$ it holds that $|B_i| \leq k(G)$.

	\end{enumerate}   
\end{claim}

\begin{proof}
	We prove the claim by induction on $k$. For the base case, let $k = 0$. Therefore, 
	
	\begin{enumerate}
		
		\item  By Step~\ref{fill:init} it holds that $(B_1, \ldots, B_m, R) = (\emptyset, \ldots, \emptyset, \Gamma_w)$, which is a partition of $\Gamma_w$.  
		
		\item By Step~\ref{fill:init}, for all $i \in [m]$ and $j \in \{2, \ldots,q(\cI)\}$ it holds that $|B_i \cap P_j| = |\emptyset \cap P_j| = 0 \leq |A_i \cap P_{j-1}|$ and $|B_i \cap P_1| = 0$.
		
		\item By Step~\ref{fill:init}, for all $i \in [m]$ and $G \in \cG$ it holds that $|B_i| = |\emptyset| =  0 \leq 1 \leq k(G)$.

	\end{enumerate}   
	
	Assume that the claim holds for some $k \in \mathbb{N}$. Now, for the step of the induction observe $k+1$. Let $(B^k_1, \ldots, B^k_m,R^k)$, $(B^{k+1}_1, \ldots, B^{k+1}_m,R^{k+1})$ be the object $(B_1, \ldots, B_m, R)$ before and after iteration $k+1$, respectively. Let $i \in [m], j \in \{2, \ldots, q(\cI)\}$ and $\ell \in P_j$ such that $\ell$ is moved from $R^k$ to $B^{k+1}_i$ in iteration $k+1$ by Step~\ref{fill:1} and Step~\ref{fill:2}.

	\begin{enumerate}
		
		\item  By the assumption of the induction, it holds that $(B^k_1, \ldots, B^k_m,R^k)$ is a partition of $\Gamma_w$. Let $y \in \Gamma_w$. If $y = \ell$, then $y \in B^{k+1}_i$ and do not belong to any other set in the partition by Step~\ref{fill:1} and Step~\ref{fill:2}. Otherwise, $y \neq \ell$; then, there is exactly one set in  $(B^k_1, \ldots, B^k_m,R^k)$ to which $y$ belongs by the assumption of the induction; also, in iteration $k+1$, by Step~\ref{fill:1} and Step~\ref{fill:2} it holds that $y$ remains in the same set since only $\ell$ is relocated in this iteration.
		
		\item Let  $i' \in [m]$ and $j' \in \{2, \ldots,q(\cI)\}$. If $i' = i$ and $j' = j$, then  
		
		$$|B^{k+1}_i \cap P_j| = |\{\ell\}|+|B^k_i \cap P_{j-1}| \leq |A_i \cap P_{j-1}|.$$
		
		The equality is by Step~\ref{fill:1} and Step~\ref{fill:2}. The inequality is by Step~\ref{fill:while}. Otherwise, it holds that $i' \neq i$ or $j' \neq j$. Therefore, 
		
		$$|B^{k+1}_i \cap P_j| = |B^k_i \cap P_{j-1}| \leq |A_i \cap P_{j-1}|.$$
		
		The equality is by Step~\ref{fill:1} and Step~\ref{fill:2}. The inequality is by the assumption of the induction.
		
		\item Let  $i' \in [m]$ and $G \in \cG$. If $i' = i$ and $G = \textsf{group}(\ell)$, then  
		
		$$|B^{k+1}_i \cap G| = |\ell \cap G|+|B^k_i \cap G| \leq k(G).$$
		
		The equality is by Step~\ref{fill:1} and Step~\ref{fill:2}. The inequality is by Step~\ref{fill:while}. Otherwise, it holds that $i' \neq i$ or $G \neq \textsf{group}(\ell)$. Therefore, 
		
		$$|B^{k+1}_i \cap G| = |B^k_i \cap G| \leq k(G).$$
		
		The equality is by Step~\ref{fill:1} and Step~\ref{fill:2}. The inequality is by the assumption of the induction.

	\end{enumerate}   
\end{proof}

\begin{claim}
	\label{claim:fill2}
	$|R \setminus P_{q(\cI)}| \leq \eps^2 \cdot \OPT(\cj)$. 
\end{claim}

\begin{proof}
	Let $j \in \{2,\ldots, q(\cI)\}$, and let $\ell \in R \cap P_j$ at the end of Algorithm~\ref{alg:fill}; in addition, let $G = \textsf{group}(\ell)$. Therefore, by Step~\ref{fill:while} at least one of the following holds for all $i \in [m]$.
	
	\begin{itemize}
		\item (i) $|B_i \cap P_j| = |A_i \cap P_{j-1}|$.
		
		\item (ii) $|G \cap B_i| = k(G)$. 
	\end{itemize}

	Now, define the set of indices in $[m]$ in which the second condition defined above holds. 
	
	\begin{equation}
		\label{fill:T}
		T = \{i \in [m]~|~ |G \cap B_i| = k(G)\}.
	\end{equation}
	
	We use the following inequality

	\begin{equation}
		\label{FF1}
		\begin{aligned}
			|P_{j-1}| ={} & \sum_{i \in [m]} |P_{j-1} \cap A_i| = \sum_{i \in [m] \setminus T} |P_{j-1} \cap A_i|+\sum_{i \in T} |P_{j-1} \cap A_i| =  \sum_{i \in [m] \setminus T} |P_j \cap B_i|+\sum_{i \in T} |P_{j-1} \cap A_i| \\ \leq{} & \sum_{i \in [m]} |P_j \cap B_i|  +\eps^{-w} \cdot |T| \leq  \sum_{i \in [m]} |P_j \cap B_i|+ \eps^{-w} \cdot \eps^{2w+4} \cdot \OPT(\cj) \\ ={} &  \sum_{i \in [m]} |P_j \cap B_i|+ \eps^{w+4} \cdot \OPT(\cj).
		\end{aligned}
	\end{equation} The first equality is because $A$ is a packing of $P_j$ in particular. The third equality is by \eqref{fill:T} and by (i).  The first inequality is because $\Gamma_w \subseteq H_w$; hence, the sizes of all items are at least $\eps^{w}$ and there can be at most $\eps^{-w}$ items from $\Gamma_w$ in $A_i, \forall i \in [m]$ because $A$ is a packing. The second inequality is because for each $i \in T$ it holds that $G \cap B_i \neq \emptyset$ by \eqref{fill:T}; thus, there are at least $|T|$ $w$-heavy items in $G$ because $B_i \subseteq \Gamma_w$. Therefore, we get by Lemma~\ref{lem:few_large_groups} that $|T| \leq  \eps^{2w+4} \cdot \OPT(\cj)$ since each $w$-small group contains at most $ \eps^{2w+4} \cdot \OPT(\cj)$ $w$-heavy items in particular. 
	
	Therefore, using the above for any $j \in \{2,\ldots, q(\cI)\}$:
	
	\begin{equation*}
		\begin{aligned}
			|R \setminus P_{q(\cI)}| ={} &\sum_{j \in \{2,\ldots,q(\cI)\}} |P_{j-1}| - \sum_{j \in \{2,\ldots,q(\cI)-1\}} \sum_{i \in [m]} |P_j \cap B_i| \leq  \eps^{w+4} \cdot \OPT(\cj) \cdot q(\cI) \\ \leq{} & \eps^{w+4} \cdot \OPT(\cj) \cdot \eps^{-w-2}  = \eps^{2} \cdot \OPT(\cj)
		\end{aligned}
	\end{equation*} The first equality is because $(B_1, \ldots, B_m,R)$ is a partition of $\Gamma_w$ by Claim~\ref{claim:fill1}. The first inequality is by \eqref{FF1}. 
\end{proof}

\begin{claim}
	\label{claim:fillR}
	$|R| \leq \eps \cdot \OPT(\cj)+1$. 
\end{claim}

\begin{proof}
	
	We use the following inequality.
	
	\begin{equation}
		\label{FF2}
		\begin{aligned}
			|R \cap P_{q(\cI)}| \leq{} & |P_{q(\cI)}| \leq \ceil{\beta^{-1} \cdot |\Gamma_w|} \leq \beta^{-1} \cdot |\Gamma_w|+1 = \eps^{w+2}  \cdot |\Gamma_w|+1 \\ \leq{} & \eps^{w+2}  \cdot \eps^{-w} \cdot \OPT(\cj) +1 = \eps^2 \cdot \OPT(\cj)+1. 
		\end{aligned}
	\end{equation} The second inequality is by the definition of $\cI$-partition. The fourth inequality is because $\Gamma_w \subseteq H_w$ by \eqref{Iu}; hence, in each bin in a packing of $\cj$ there can be at most $\eps^{-w}$ items from $\Gamma_w$ since each item is $w$-heavy and has a size at least $\eps^{-w}$; it follows that $ |\Gamma_w|  \leq \eps^{-w} \cdot \OPT(\cj)$. Now, 
	
	$$|R| = |R \cap P_{q(\cI)}|+|R \setminus P_{q(\cI)}| \leq |P_{q(\cI)}|+|R \setminus P_{q(\cI)}| \leq \eps^2 \cdot \OPT(\cj)+1+2\eps^2 \cdot \OPT(\cj) \leq \eps \cdot \OPT(\cj)+1.$$ The second inequality is by \eqref{FF2} and by Claim~\ref{claim:fill2}. 
\end{proof}

\noindent{\bf Proof of Lemma~\ref{lem:Bgreedy}:} By Claim~\ref{claim:fill1} and Claim~\ref{claim:fillR} the output of the algorithm satisfies the required properties, since by Step~\ref{fill:return} the returned value is the direct outcome of the while loop. In addition, the running time satisfies the properties by Claim~\ref{claim:fillRunningTime}.

\section{Algorithm \textsf{Greedy}}
\label{sec:greedy}

In this section we give the proof of Lemma~\ref{thm:greedy}.  Algorithm \textsf{Greedy} is indeed a greedy algorithm, which sequentially adds a new bin and tries to (i) maximize the total size packed in the bin while (ii) packing items from groups which have more items w.r.t. the matroid constraint. Some of the proofs from this section are deferred to Section~\ref{sec:proofsGreedy}.  

 For this section, fix a BPP instance $\cI = (I, \cG,s,k)$ and let $\delta \in (0,0.5)$ such that for all $\ell \in I$ it holds that $s(\ell) \leq \delta$.  We use $\cC(\cj)$ to denote the set of configurations of a BPP instance $\cj$. Also, recall that the operator $\oplus$ denotes tuple (or, packing) concatenation. In addition, for any $S \subseteq I$ we define $\cI \setminus S = (I_S,\cG_S,s_S,k_S)$ as the BPP instance such that \begin{equation}
		\label{IIS}
	\begin{aligned}
		I_S = I \setminus S, ~~~~~~~~~~~~~~~~~~~~ ~~~~~~
		\mathcal{G}_S = \{G \cap I_S\neq \emptyset~|~ G \in \cG\}, ~~~~~~~~~~~~~~~~~~~~~~~~~~~~~~~~~
		\\
		s_S(\ell) = s(\ell), ~~\forall \ell \in I_S, ~~~~~~~~~~~
		k_S(G \cap I_S) = k(G), ~~\forall G \in \cG \text{ s.t. } G \cap I_S\neq \emptyset.~~~~~~~~~~~
	\end{aligned}
\end{equation}
\begin{claim}
	\label{clm:confappend}
	For $S \in \cC(\cI)$ and a packing $X$ of $\cI \setminus S$ it holds that $T = (S) \oplus X$ is a packing of $\cI$.
\end{claim}

Using Claim~\ref{clm:confappend}, the algorithm first finds a configuration  $S$ of $\cI$ and then packs recursively the instance $\cI \setminus A$. Recall that $V(\mathcal{I}) = \max_{G \in \cG} \ceil{\frac{|G|}{|k(G)|}}$ and let the {\em promise} of $\cI$ be a maximum bound over $\OPT(\cI)$ combining a bound over item sizes and over the matroid: \begin{eqnarray}
	\label{promise}
	p(\cI) = \max \left\{    (1+2\delta) \cdot s(I)+2, V(\cI)    \right\}. 
\end{eqnarray} We show that the proposed algorithm \textsf{Greedy} constructs a packing of $\cI$ in at most $p(\cI)$ bins. For doing so, in each bin we must pack enough items from groups $G \in \cG$ which satisfy that $\ceil{  \frac{|G|}{k(G)} } > p(\cI)-1$; otherwise, the packing might require more bins than $p(\cI)$. Thus, we define the set of {\em bounding groups} as

 \begin{eqnarray}
	\label{bounding}
	\cB (\cI) =\left\{    G \in \cG ~\bigg|~ \ceil{ \frac{|G|}{k(G)} } > p(\cI)-1   \right\}
\end{eqnarray} 

To guarantee that after packing the first bin by the algorithm we remain with not too many items from some group, we must take certain number of items from each bounding group. Specifically, for each bounding group $G \in \cB(\cI)$ we define a {\em bounding subset} of $G$ as a subset of items $S \subseteq G$ such that $ \ceil{ \frac{|G \setminus S|}{k(G)} } \leq p(\cI)-1  $. Also, let $b(G)$ be the set of all bounding subsets of $G$. Finally, define a {\em minimal bounding subset} of $G$ as a bounding subset $S^* \in b(G)$ of $G$ such that (i) $|S^*| \leq k(G)$ and (ii) $s(S^*) \leq \frac{s(G)}{s(I)}$. For any $G \in \cB(\cI)$ let $b^*(G)$ be the set of minimal bounding subsets of $G$.  \begin{lemma}
	\label{bG}
	For any $G \in \cB(\cI)$ it holds that $b^*(G) \neq \emptyset$.
\end{lemma}

We generalize the definition of minimal bounding subset for the entire instance. Specifically, define the set of  {\em bounding subsets} of $\cI$ as all unions of minimal bounding subsets of all bounding groups: \begin{equation}
	\label{bounding:A}
	b^*(\cI) = \left\{    S \subseteq \bigcup_{G \in \cB(\cI)} G~\bigg|~ \forall G \in \cB(\cI)~:~ S \cap G \in b^*(G) \right\}
\end{equation}

\begin{lemma}
	\label{bI}
	 $b^*(\cI) \neq \emptyset$.
\end{lemma}

\begin{proof}
	
	For all $G \in \cB(\cI)$ let $S_G \in b^*(G)$; there is such $S_G \in b^*(G)$ by Lemma~\ref{bG}. Now, define $S = \bigcup_{G \in \cB(\cI)} S_G$; it follows by \eqref{bounding:A} that $S \in b^*(\cI)$. 
\end{proof}

\begin{algorithm}[h]
	\caption{$\textsf{Greedy}(\mathcal{I} = (I,\cG,s,k), \delta)$}
	\label{alg:GREEDY}

	\If{$I \in \cC(\cI)$\label{gr:if1}}{
		
		return $(I)$.\label{gr:if}
		
	}

	Choose arbitrary $A \in b^*(\cI)$ bounding subset of $\cI$.\label{gr:end0}

		\While{$s(A) \leq 1-\delta ~\textnormal{and}~ \exists \ell \in I \setminus A \textnormal{ s.t. } |\textnormal{\textsf{group}}(\ell) \cap A| < k\left( \textnormal{\textsf{group}}(\ell) \right)$ \label{gr:while2}}{

		$A \leftarrow A \cup \{\ell\}$.\label{gr:swap2}
		
	}\label{gr:end1}

	\While{$\exists \ell \in A, y \in I \setminus A \textnormal{ s.t. } s(y) > s(\ell), \textnormal{\textsf{group}}(\ell)  = \textnormal{\textsf{group}}(y), ~\textnormal{and}~ s(A) \leq 1-\delta $ \label{gr:while1}}{
		
		$A \leftarrow \left(A \setminus \{\ell\} \right) \cup \{y\}$.\label{gr:swap}
		
	}

	Return $(A) \oplus \textsf{Greedy}\left(\mathcal{I} \setminus A, \delta \right)$.\label{gr:return}
	
\end{algorithm}

Given the definitions of bounding groups and bounding subsets, we describe Algorithm \textsf{Greedy}. The algorithm is recursive. The base case is when the set of items $I$ forms a configuration of $\cI$; in this case, we simply return a packing of one bin containing all items. Otherwise, we construct a bounding subset $A$ of $\cI$, which exists by Lemma~\ref{bI}. Then, we try to increase the total size of items in $A$, without exceeding the capacity of the bin, in two ways. First, by adding items to $A$ from groups that do not meet the cardinality constraint. Second, by replacing items from the same group between $A$ and $I \setminus A$ such that the item moved into $A$ is of larger size. Finally, we define the resulting bin $A$ as the first bin in the packing, and pack recursively the instance $\cI \setminus A$ using Algorithm \textsf{Greedy}. We give the pseudocode of Algorithm \textsf{Greedy} in Algorithm~\ref{alg:GREEDY}. 

The proof of Lemma~\ref{thm:greedy} is based on viewing two complementary cases. First, if $A$, the first bin in the packing, has {\em small} total size of items, namely $s(A) < 1-\delta$, then in this case $A$ is a {\em basis} of maximum total size of the matroid over the items. This also holds for the following bins of the packing, or we would be able to increase the size of $A$. Thus, an inductive argument shows that the algorithm finds a packing of $\cI$ in at most $V(\cI)$ bins.

The complementary case is where $s(A) \geq 1-\delta$; then, the total size deducted from $s(I)$ by $A$ is significant, and there can be at most $(1+2\delta) \cdot s(I)+1$ bins satisfying this property. In addition, it also holds that $V(\cI \setminus A) \leq p(\cI)-1$ because we take a bounding subset for $A$ with additional items; using the first two cases, it can be inductively deduced that the returned packing in case that $s(A) \geq 1-\delta$ is of at most $p(\cI)$ bins.

\section{Omitted Proofs of Section~\ref{sec:greedy}}
\label{sec:proofsGreedy}

 For this section as in Section~\ref{sec:greedy}, we have a BPP instance $\cI = (I, \cG,s,k)$ and $\delta \in (0,0.5)$ such that for all $\ell \in I$ it holds that $s(\ell) \leq \delta$. Also, assume that $I \neq \emptyset$ otherwise all proofs here are trivial.  
 
\noindent{\bf Proof of Lemma~\ref{clm:confappend}:} We prove the claim by the definition of packing. For convenience, let $T = (S, X_1, \ldots, X_m)$ and also $T = (T_1, \ldots, T_{m+1})$.

\begin{enumerate}
	\item It holds that $T$ is a partition of $I$. Let $\ell \in I$. If $\ell \in S$, then $\ell \in T_1$ and for all $i \in [m+1] \setminus \{1\}$ it holds that $\ell \notin T_i$ because $X_{i-1} \subseteq I \setminus S$. Otherwise, it holds that $\ell \in I \setminus S$; therefore, because $X$ is a packing of $\cI \setminus S$ there is exactly one $i \in [m]$ such that $\ell \in X_i$ and it follows that $\ell \in T_{i+1}$ (only).
	
	\item For all $i \in [m+1]$ it holds that $s(T_i) \leq 1$. If $i = 1$ then $s(T_1) = s(S) \leq 1$ because $S \in \cC(\cI)$. Otherwise, it holds that $s(T_i) = s(X_{i-1}) \leq 1$ because $X$ is a packing of $\cI \setminus S$. 
	
	\item For all $G \in \cG$ and $i \in [m+1]$ it holds that $|T_i \cap G| \leq k(G)$.  If $i = 1$ then $|T_i \cap G| = |G \cap S| \leq k(G)$ because $S \in \cC(\cI)$. Otherwise, it holds that $|T_i \cap G| = |X_{i-1} \cap G|\leq k_S(G \setminus S) = k(G)$. The inequality is because $X$ is a packing of $\cI \setminus S$. The last equality is by \eqref{IIS}.\qed
\end{enumerate}

\noindent{\bf Proof of Lemma~\ref{bG}:} Let $n = |G|$; by \eqref{promise} and \eqref{bounding} it holds that $n \geq 1$. Now, let $\ell_1, \ldots, \ell_n$ be the items in $G$ in an increasing order of item sizes. Define	\begin{equation}
	\label{bounding:psi}
	\psi = \argmin_{\left\{i \in [n] ~\big|~  \ceil{  \frac{n-i}{k(G)}} \leq V(\cI)-1\right\}} i.
\end{equation} Observe that for $i =n$ it holds that $ \ceil{  \frac{n-i}{k(G)}}  \leq V(\cI)-1$ because $n \geq 1$; thus, it follows that $\psi$ is well defined. Now, define $S^* = \{\ell_1, \ldots, \ell_{\psi}\}$. We show next that $S^* \in b^*(G)$.

\begin{itemize}
	\item $S^* \in b(G)$. First, $S^* \subseteq G$ by definition. Second, $$ \ceil{ \frac{|G \setminus S^*|}{k(G)} } = \ceil{  \frac{n-\psi}{k(G)}} \leq V(\cI)-1 \leq p(\cI)-1.$$ The first inequality is by \eqref{bounding:psi}. The last inequality is by \eqref{promise}.

	\item $|S^*| \leq k(G)$. \begin{equation}
		\label{S:psi}
		\begin{aligned}
			|S^*| ={} & \psi \leq |G|-(V(\cI)-1 ) \cdot k(G) \leq |G|-\left(\frac{|G|}{k(G)}-1 \right) \cdot k(G) = k(G).
		\end{aligned}
	\end{equation} The first inequality is by \eqref{bounding:psi}, since $G \in \cB(\cI)$ and $\ceil{\frac{|G|-\left( |G|-(V(\cI)-1) \cdot k(G)\right)}{k(G)}} \leq (V(\cI)-1)$ (recall that $V(\cI) \in \mathbb{N}$).

	\item $s(S^*) \leq \frac{s(G)}{s(I)}$. It holds that: $$s(S^*) =\psi \cdot \frac{ s(S^*)}{\psi} \leq \psi \cdot \frac{s(G)}{n} \leq \frac{k(G)}{|G|} \cdot s(G) \leq \frac{s(G)}{p(\cI)-2}  \leq \frac{s(G)}{(1+2\delta) \cdot s(I)+2-2} \leq \frac{s(G)}{ s(I)}.$$ The first inequality is because $S^*$ is a subset of items of $G$ where each item in $G \setminus S^*$ has a larger or equal size compared to the size of any item in $S^*$ by the sorted order of the items; thus, the size of the average item in $S^*$ is smaller or equal to the size of the average item in $G$. The second inequality is by \eqref{S:psi}. The third inequality is by \eqref{bounding} and that $G \in \cB(\cI)$:  it holds that $\frac{|G|}{k(G)}+1 \geq \ceil{\frac{|G|}{k(G)}} > p(\cI)-1$. The fourth inequality is by \eqref{promise}.\qed
\end{itemize}

Given a BPP instance $\cj$, let $A^{\cj}$ be the first entry (from the left) of the returned tuple by $\textsf{Greedy}(\cj,\delta)$ and let $A^{\cj}_0,A^{\cj}_1$ be the object $A^{\cj}$ after Step~\ref{gr:end0} and after Step~\ref{gr:end1} in the computation of $\textsf{Greedy}(\cj,\delta)$, respectively. Note that $A^{\cj}_0$ is well defined (the algorithm is guaranteed to construct $A^{\cj}_0$) by Lemma~\ref{bI}. For simplicity, we use $A, A_0,A_1$ instead of $A^{\cj}, A^{\cj}_0, A^{\cj}_1$, respectively, when $\cj = \cI$. Now, for the proof of Lemma~\ref{thm:greedy}, we use the following auxiliary claims.

\begin{claim}
	\label{clm:A00}
	For any BPP instance $\cI = (I,\cG,s,k)$ it holds that $A_0 \in \cC(\cI)$ and for each bounding group $G \in \cB(\cI)$ it holds that $   \ceil{ \frac{|G \setminus A_0|}{k(G)}    } \leq p(\cI)-1$.
\end{claim}

\begin{proof}
Observe that $A_0 \in b^*(\cI)$ by Step~\ref{gr:end0}. In addition, note that $A_0$ is well defined by Lemma~\ref{bI}. Then, by \eqref{bounding:A} for each $G \in \cB(\cI)$ there is $S_G \in b^*(G)$ such that $A_0 = \bigcup_{G \in \cB(\cI)} S_G$. We show the conditions of the lemma as follows. \begin{itemize}
		\item $s(A_0) \leq 1$. $$s(A_0) =  s\left(  \bigcup_{G \in \cB(\cI)} S_G \right)  \leq \sum_{G \in \cB(\cI)} \frac{s(G)}{s(I)} \leq \frac{s(I)}{s(I)} = 1.$$ The first equality is by \eqref{bounding:A}. The first inequality is by the definition of a minimal bounding subset of a bounding group. 

		\item For all $G \in \cG$ it holds that $|G \cap A_0| \leq k(G)$. If $G \in \cG \setminus \cB(\cI)$ then $|G \cap A_0| =0 < k(G)$ by \eqref{bounding:A}. Otherwise, $|G \cap A_0| = \left|  S_G  \right| \leq  k(G)$ where the  equality is by \eqref{bounding:A} and the inequality is by  the definition of a minimal bounding subset of a bounding group.

		\item For each bounding group $G \in \cB(\cI)$ it holds that $ \ceil{ \frac{|G \setminus A_0|}{k(G)}    } \leq p(\cI)-1$. $$ \ceil{ \frac{|G \setminus A_0|}{k(G)}    } = \ceil{ \frac{|G \setminus S_G|}{k(G)}    }\leq p(\cI)-1$$ The first equality is by \eqref{bounding:A}. The first inequality is because $S_G \in b^*(G)$ and in particular $S_G \in b(G)$; thus, the inequality follows by the definition of a bounding subset of $G$.  
	\end{itemize}

\end{proof}

\begin{claim}
	\label{clm:A1}
	For any BPP instance $\cI = (I,\cG,s,k)$ it holds that $A_1 \in \cC(\cI)$ and for each bounding group $G \in \cB(\cI)$ it holds that $   \ceil{ \frac{|G \setminus A_1|}{k(G)}    } \leq p(\cI)-1$.
\end{claim}

\begin{proof}
	We prove the claim using loop invariant for the while loop of Step~\ref{gr:while2}. Let $A(s),A(t)$ be the object $A$ before and after an iteration of the while loop of Step~\ref{gr:while2}, respectively. Now, assume that $A(s) \in \cC(\cI)$ and for each bounding group $G \in \cB(\cI)$ it holds that $   \ceil{ \frac{|G \setminus A(s)|}{k(G)}    } \leq p(\cI)-1$. We show the conditions of the claim below for $A(t)$. By Step~\ref{gr:swap2}, let $\ell \in I$ such that $A(t) = A(s) \cup \{\ell\}$. Therefore,  \begin{itemize}
		\item $s(A(t)) \leq 1$. $$s(A(t)) =  s\left(  A(s) \right) + s(\ell) \leq (1-\delta)+\delta =1.$$ The first equality is by Step~\ref{gr:swap2}. The inequality is because $s(A(s)) \leq 1-\delta$ by Step~\ref{gr:while2}; in addition, for all items the size is at most $\delta$.

		\item For all $G \in \cG$ it holds that $|G \cap A(t)| \leq k(G)$. If $G \in \cG \setminus \{\textsf{group}(\ell)\}$; then, $|G \cap A(t)| = |A(s) \cap G| \leq k(G)$. The equality is by Step~\ref{gr:swap2} and the inequality is by the assumption that $A(s) \in \cC(\cI)$. Otherwise, $|G \cap A(t)| = \left|  A(s) \cap G \right| +|\{\ell\} \cap G| = |A(s) \cap G|+1 \leq  k(G)$. The first equality is by Step~\ref{gr:swap2} and the inequality is by Step~\ref{gr:while2}.

		\item For each bounding group $G \in \cB(\cI)$ it holds that $   \ceil{ \frac{|G \setminus A(t)|}{k(G)}    } \leq p(\cI)-1$. If $G \in \cG \setminus \{\textsf{group}(\ell)\}$; then, $   \ceil{ \frac{|G \setminus A(t)|}{k(G)}    } =  \ceil{ \frac{|G \setminus A(s)|}{k(G)}    } \leq p(\cI)-1$. The equality is by Step~\ref{gr:swap2} and the inequality is by the assumption on $A(s)$. Otherwise, $$ \ceil{ \frac{|G \setminus A(t)|}{k(G)}    } =  \ceil{ \frac{|G \setminus A(s)|-|\{\ell\}|}{k(G)}    } \leq \ceil{ \frac{|G \setminus A(s)|}{k(G)}    } \leq p(\cI)-1.$$ The first equality is by Step~\ref{gr:swap2} and the last inequality is by the assumption on $A(s)$. Therefore, by Claim~\ref{clm:A00} and the above loop invariant the claim follows.  
	\end{itemize}

\end{proof}

\begin{claim}
	\label{clm:A}
	$A \in \cC(\cI)$ and for each bounding group $G \in \cB(\cI)$ it holds that $\ceil{ \frac{|G \setminus A|}{k(G)}    } \leq p(\cI)-1$.
\end{claim}

\begin{proof}
	We prove the claim using loop invariant for the while loop of Step~\ref{gr:while1}. Let $A(s),A(t)$ be the object $A$ before and after an iteration of the while loop of Step~\ref{gr:while1}, respectively. Now, assume that $A(s) \in \cC(\cI)$ and for each bounding group $G \in \cB(\cI)$ it holds that $   \ceil{ \frac{|G \setminus A(s)|}{k(G)}    } \leq p(\cI)-1$. We show the conditions of the claim below for $A(t)$.  By Step~\ref{gr:swap}, let $\ell,y \in I$ such that $A(t) = \left(A(s) \setminus \{\ell\} \right) \cup \{y\}$. Therefore,  \begin{itemize}
		\item $s(A(t)) \leq 1$. $$s(A(t)) =  s\left(  A(s) \right) -s(\ell)+s(y) \leq s\left(  A(s) \right) +s(y) \leq (1-\delta)+\delta =1.$$ The first equality is by Step~\ref{gr:swap}. The second inequality is because $s(A(s)) \leq 1-\delta$ by Step~\ref{gr:while1}; in addition, for all items the size is at most $\delta$.

		\item For all $G \in \cG$ it holds that $|G \cap A(t)| \leq k(G)$. If $G \in \cG \setminus \{\textsf{group}(\ell)\}$; then, $|G \cap A(t)| = |A(s) \cap G| \leq k(G)$. The equality is by Step~\ref{gr:swap} and the inequality is by the assumption that $A(s) \in \cC(\cI)$. Otherwise, $|G \cap A(t)| = \left|  A(s) \cap G \right| -|\{\ell\} \cap G|+|\{y\} \cap G| = |A(s) \cap G| \leq  k(G)$. The first equality is by Step~\ref{gr:swap} and the inequality is by the assumption that $A(s) \in \cC(\cI)$

		\item For each bounding group $G \in \cB(\cI)$ it holds that $\ceil{ \frac{|G \setminus A(t)|}{k(G)}    } \leq p(\cI)-1$. If $G \in \cG \setminus \{\textsf{group}(\ell)\}$; then, $\ceil{ \frac{|G \setminus A(t)|}{k(G)}    } = \ceil{ \frac{|G \setminus A(s)|}{k(G)}    }\leq p(\cI)-1$. The equality is by Step~\ref{gr:swap} and the inequality is by the assumption on $A(s)$. Otherwise, $$\ceil{ \frac{|G \setminus A(t)|}{k(G)}    } = \ceil{ \frac{|G \setminus A(s)|+|\{\ell\} \cap G|-|\{y\} \cap G|}{k(G)}    }= \ceil{ \frac{|G \setminus A(s)|}{k(G)}    } \leq p(\cI)-1$$ The equalities are by Step~\ref{gr:swap} and the inequality is by the assumption on $A(s)$. Therefore, by Claim~\ref{clm:A1} and the above loop invariant the claim follows.

	\end{itemize}

\end{proof}

\begin{claim}
	\label{ABCV}
	$V(\cI \setminus A)\leq p(\cI)-1$.
\end{claim}

\begin{proof}
	\begin{equation*}
		V(\cI \setminus A) = \max_{G \in \cG_A} \ceil{ \frac{|G|}{k_A(G)}  } = \max_{G \in \cG}        \ceil{ \frac{|G \setminus A|}{k(G)}    } \leq  p(\cI)-1
	\end{equation*} The second equality is by \eqref{IIS}. The last inequality is because for all bounding groups $G \in \cB(\cI)$ it holds that $   \ceil{ \frac{|G \setminus A|}{k(G)}    } \leq p(\cI)-1$ by Claim~\ref{clm:A} and for all $G \in \cG \setminus \cB(\cI)$ it holds that $   \ceil{ \frac{|G \setminus A|}{k(G)}    } \leq p(\cI)-1$ by \eqref{bounding}.
\end{proof}

\begin{lemma}
	\label{delta<}
	If $s(A) < 1-\delta$, then Algorithm~\ref{alg:GREEDY} returns a packing for $\cI$ with at most $p(\cI)$ bins.
\end{lemma}

\begin{proof}
	
	We use the following auxiliary claims. 
	
	\begin{claim}
		\label{VD1}
		For all $G \in \cG$ it holds that $|A \cap G| = \min\{k(G),|G|\}$. 
	\end{claim}
	
	\begin{proof}
		Assume towards a contradiction that there is $G \in \cG$ such that $|A \cap G| < \min\{k(G),|G|\}$. Therefore, there is $\ell \in G \setminus A$. Observe $A$ after Step~\ref{gr:end0}. Because $s(A) < 1-\delta$, it follows by Step~\ref{gr:while2} and Step~\ref{gr:swap2} that $\ell$ can be added to $A$ in contradiction. 
	\end{proof}
	
	\begin{claim}
		\label{VD2}
		For all $G \in \cG$, $\ell \in A \cap G$ and $y \in G \setminus A$ it holds that $s(y) \leq s(\ell)$. 
	\end{claim}
	
	\begin{proof}
		Assume towards a contradiction that there are  $\ell \in A, y \in I \setminus A$ such that  $s(y) > s(\ell)$ and $\textnormal{\textsf{group}}(\ell)  = \textnormal{\textsf{group}}(y)$. Observe $A$ after Step~\ref{gr:end1}.  Because $s(A) < 1-\delta$, it follows by Step~\ref{gr:while1} and Step~\ref{gr:swap} that $y$ can be added to $A$ and $\ell$ can be removed from $A$ in contradiction. 
		
	\end{proof}

	\begin{claim}
		\label{ABCV2}
		If $s(A) < 1-\delta$ then $V(\cI \setminus A)\leq V(\cI)-1$.
	\end{claim}
	
	\begin{proof}
		\begin{equation*}
			\begin{aligned}
					V(\cI \setminus A) ={} & \max_{G \in \cG_A} \ceil{ \frac{|G|}{k_A(G)}  } = \max_{G \in \cG}        \ceil{ \frac{|G \setminus A|}{k(G)}    } =   \max_{G \in \cG}        \ceil{ \frac{|G|-\min\{ |G|,k(G) \} }{k(G)}} \\ \leq{} & \max_{G \in \cG}        \ceil{ \frac{|G|}{k(G)} -1 } =  \max_{G \in \cG}        \ceil{ \frac{|G|}{k(G)}}-1 = V(\cI)-1.
			\end{aligned}
		\end{equation*} The second equality is by \eqref{IIS}. The third equality is by Claim~\ref{VD1}. The inequality is because if $\min\{ |G|,k(G) \} = |G|$ then  it follows that $V(\cI \setminus A) = 0 \leq V(\cI)-1$, assuming that $I \neq \emptyset$. 
	\end{proof}

	Now, back to the proof of Lemma~\ref{delta<}. We prove the lemma by induction on $V(\cI)$. For $V(\cI) = 1$, it follows that $V(\cI \setminus A) \leq V(\cI)-1$ by Claim~\ref{ABCV2} and it follows that $V(\cI \setminus A) = 0$. Therefore, by \eqref{IIS} it follows that $I \setminus A = \emptyset$ and therefore $A = I$. Thus, $I$ is a configuration of $\cI$ by Claim~\ref{clm:A}. Hence, by Step~\ref{gr:if1} and Step~\ref{gr:if} it holds that Algorithm~\ref{alg:GREEDY} returns $(I)$, which is a packing of $\cI$ because $I = A$ and $A$ is a configuration of $\cI$. 
	
	Recall that for any BPP instance $\cj$ we define $A^{\cj}$ as the first entry of the tuple returned by $\textsf{Greedy}(\cj,\delta)$. Now, for the assumption of the induction, assume that  $V(\cI) \geq 2$. Assume that for any BPP instance $\cj = (J,\cG_J,s_J,k_J)$ such that (i) $s_J(\ell) \leq \delta ~\forall \ell \in J$, (ii) $V(\cj) \leq V(\cI)-1$, and (iii) $s\left(A^{\cj}\right) < 1-\delta$, it holds that $\textsf{Greedy}(\cj,\delta)$ returns a packing of $\cj$ of at most $V(\cj)$ bins. For the step of the induction, we use the following auxiliary claim.

	\begin{claim}
		\label{VD3}
		$s\left(A^{\cI \setminus A}\right) < 1-\delta$.
	\end{claim}
	
	\begin{proof}
		Assume towards a contradiction that $s\left(A^{\cI \setminus A}\right) \geq 1-\delta$. Therefore, it holds that $s\left(A^{\cI \setminus A}\right) > s(A)$; thus, we reach a contradiction if one of the following conditions hold. 
		
		\begin{enumerate}
			
			\item $|A^{\cI \setminus A}| > |A|$. Therefore, $$|A^{\cI \setminus A}| = \sum_{G \in \cG} |A^{\cI \setminus A} \cap G| \leq \sum_{G \in \cG} \min\{k(G),|G|\} = \sum_{G \in \cG} |A \cap G| = |A|.$$ The inequality is because $A^{\cI \setminus A} \in \cC(\cI \setminus A)$ by Claim~\ref{clm:A} and $\cC(\cI \setminus A) \subseteq \cC(\cI)$ by \eqref{IIS}.\footnote{Claim~\ref{clm:A} is stated and proven for $A$ and not for $A^{\cI \setminus A}$ for simplicity; however, the exact arguments can be used for proving the claim also for $A^{\cI \setminus A}$.} The second equality is by Claim~\ref{VD1}. The above equation is a contradiction that $|A^{\cI \setminus A}| > |A|$.

			\item There are $G \in \cG$, $\ell \in A^{\cI \setminus A} \cap G$ and $y \in G \setminus A^{\cI \setminus A}$ such that $s(y) > s(\ell)$. This is a contradiction to Claim~\ref{VD2}. 
		\end{enumerate}
		
		Otherwise, if the two cases above are not satisfied:
		
		$$s\left(A^{\cI \setminus A}\right) = \sum_{G \in \cG~} s\left(G\cap A^{\cI \setminus A}\right) \leq  \sum_{G \in \cG~} s\left(G\cap A \right) = s(A).$$ The inequality is by Claim~\ref{VD1}, Claim~\ref{VD2} and that conditions 1. and 2. above are not satisfied. By the above equation, we reach a contradiction that  $s\left(A^{\cI \setminus A}\right) > s(A)$. 
		
	\end{proof}

	By Claim~\ref{VD3} we can apply the assumption of the induction on $\cI \setminus A$; thus, Algorithm~\ref{alg:GREEDY} returns a packing for $\cI \setminus A$ with at most $V(\cI \setminus A)$ bins. In addition, by Claim~\ref{ABCV2} it holds that $V(\cI \setminus A) \leq V(\cI)-1$. Hence, by Step~\ref{gr:return}, the returned object by the algorithm is $X = (A) \oplus \textsf{Greedy}\left(\mathcal{I} \setminus A, \delta \right)$ which contains at most $V(\cI)$ entries. Moreover, it holds that $X$ is a packing of $\cI$ by Claim~\ref{clm:confappend} using the following arguments: (i) $A \in \cC(\cI)$ by Claim~\ref{clm:A} and (ii) $ \textsf{Greedy}\left(\mathcal{I} \setminus A, \delta \right)$ is a packing of $\cI \setminus A$ by the assumption of the induction. 
\end{proof}

\begin{lemma}
	\label{delta>>}
	If $s(A) \geq 1-\delta$, then Algorithm~\ref{alg:GREEDY} returns a packing for $\cI$ with at most $p(\cI)$ bins.
\end{lemma}

\begin{proof}
	We prove the claim by induction on $p(\cI)$. For $p(\cI) \leq 1$ it holds by \eqref{promise} that $I$ is a configuration of $\cI$: it follows that $V(\cI \setminus A) = 0$ by Claim~\ref{ABCV} and therefore, by \eqref{IIS} it follows that $I \setminus A = \emptyset$ and $A = I$.  Hence, by Step~\ref{gr:if1} and Step~\ref{gr:if} it holds that Algorithm~\ref{alg:GREEDY} returns $(I)$, which is a packing of $\cI$ because $I$ is a configuration of $\cI$. Now, for the assumption of the induction, assume that for any BPP instance $\cj = (J,\cG_J,s_J,k_J)$ such that (i) $s(\ell) \leq \delta ~\forall \ell \in J$, (ii) $s\left(A^{\cj}\right) \geq 1-\delta$, and (iii) $p(\cj) \leq p(\cI)-1$, it holds that $\textsf{Greedy}(\cj,\delta)$ returns a packing of $\cj$ of at most $p(\cj)$ bins.
	
	For the step of the induction, assume that $p(\cI) > 1$. We use the following auxiliary claim.

	\begin{claim}
		\label{prom:2}
		$p\left(\cI \setminus A\right) \leq  p(\cI)-1$.
	\end{claim}
	
	\begin{proof}
		
		We use the following inequality.  \begin{equation}
			\label{ABC}
			(1+2\delta) \cdot s(I \setminus A) +2 \leq (1+2\delta) \cdot (s(I)-(1-\delta))+2 \leq (1+2\delta) \cdot s(I) -1+2\leq p(\cI)-1. 
		\end{equation} The first inequality is because $s(A) \geq 1-\delta$. The second inequality is because $\delta \in (0,0.5)$. The last inequality is by \eqref{promise}. Now, $$p\left(\cI \setminus A\right) = \max \left\{    (1+2\delta) \cdot s(I \setminus A)+2, V(\cI \setminus A)    \right\} \leq p(\cI)-1.$$ The first equality is by \eqref{IIS} and \eqref{promise}. The first inequality is by \eqref{ABC} and Claim~\ref{ABCV}
	\end{proof}
	
	To conclude the lemma, observe the following two complementary cases. \begin{enumerate}
		\item $s\left(A^{\cI \setminus A}\right) \geq 1-\delta$. Then, by Claim~\ref{prom:2} and that $s\left(A^{\cI \setminus A}\right) \geq 1-\delta$, we can apply the assumption of the induction on $\cI \setminus A$; thus, Algorithm~\ref{alg:GREEDY} applied on $\cI \setminus A$ returns a packing for $\cI \setminus A$ with at most $p\left(\cI \setminus A\right)$ bins. In addition, by Claim~\ref{prom:2} it holds that $p\left(\cI \setminus A\right) \leq  p(\cI)-1$. Hence, by Step~\ref{gr:return}, the returned object by the algorithm on $\cI$ is $X = (A) \oplus \textsf{Greedy}\left(\mathcal{I} \setminus A, \delta \right)$ which contains at most $ 1+p(\cI)-1  = p(\cI)$ entries. Moreover, it holds that $X$ is a packing of $\cI$ by Claim~\ref{clm:confappend} using the following arguments: (i) $A \in \cC(\cI)$ by Claim~\ref{clm:A} and (ii) $ \textsf{Greedy}\left(\mathcal{I} \setminus A, \delta \right)$ is a packing of $\cI \setminus A$ by the assumption of the induction. 
		
		\item $s\left(A^{\cI \setminus A}\right) < 1-\delta$. Then, by Lemma~\ref{delta<} and that $s\left(A^{\cI \setminus A}\right) < 1-\delta$, Algorithm~\ref{alg:GREEDY} applied on $\cI \setminus A$ returns a packing for $\cI \setminus A$ with at most $V(\cI \setminus A)$ bins. In addition, by Claim~\ref{ABCV} it holds that $V(\cI \setminus A) \leq  p(\cI)-1$. Hence, by Step~\ref{gr:return}, the returned object by the algorithm is $X = (A) \oplus \textsf{Greedy}\left(\mathcal{I} \setminus A, \delta \right)$ which contains at most $ 1+p(\cI)-1  \leq p(\cI)$ entries by \eqref{promise}. Moreover, it holds that $X$ is a packing of $\cI$ by Claim~\ref{clm:confappend} using the following arguments: (i) $A \in \cC(\cI)$ by Claim~\ref{clm:A} and (ii) $ \textsf{Greedy}\left(\mathcal{I} \setminus A, \delta \right)$ is a packing of $\cI \setminus A$ by Lemma~\ref{delta<}. 
	\end{enumerate}
\end{proof}

\begin{lemma}
	\label{gr:running}
	The running time of Algorithm~\ref{alg:GREEDY} is $\textnormal{poly} (|\cI |)$. 
\end{lemma}

\begin{proof}
	Step~\ref{gr:if} and Step~\ref{gr:if1} can be trivially computed in linear time in $|\cI|$. Step~\ref{gr:end0} can be computed in polynomial time in $|\cI|$ by (i) finding the bounding groups by \eqref{bounding}, and (ii) compute a bounding subset of $\cI$ as follows. For each bounding group $G$ we compute a bounding subset of $G$ by sorting the items and choosing a minimal number of items according to the sorted order that form a bounding subset of $G$ (for more details see the proof of Lemma~\ref{bG}).    
	
	the while loop of Step~\ref{gr:while2} runs at most $|I|$ times because in each iteration an item is added to the constructed bin and is not added again by Step~\ref{gr:while2} and Step~\ref{gr:swap2}. Finally, the while loop of Step~\ref{gr:while1} runs at most ${|I| \choose 2} = \textnormal{poly} (|\cI |)$ times because in each iteration two items are replaced in Step~\ref{gr:swap}, one is add to the constructed bin and the other one is removed; for any two items, this can happen at most once by Step~\ref{gr:while1}. We conclude that the running time of Algorithm~\ref{alg:GREEDY} is $\textnormal{poly} (|\cI |)$. 
\end{proof}

\noindent{\bf Proof of Lemma~\ref{thm:greedy}:} If $s(A) < 1-\delta$ it holds that $X = (A) \oplus \textsf{Greedy}\left(\mathcal{I} \setminus A, \delta \right)$ is a packing of $\cI$ of at most $V(\cI)$ bins by Lemma~\ref{delta<}; therefore, the number of bins is at most $p(\cI)$ by \eqref{promise}. Otherwise, it holds that $s(A) \geq 1-\delta$; thus, it holds that $X$ is a packing of $\cI$ of at most $p(\cI)$ bins by Lemma~\ref{delta>>}. Therefore, the number of bins in the packing is at most $p(\cI) \leq (1+2\delta) \cdot \max\left\{ s(I),~V(\mathcal{I})\right\}+2$ by \eqref{promise}. Finally, the running time is $\textnormal{poly} (|\cI |)$ by Lemma~\ref{gr:running}.\qed

		\bibliographystyle{splncs04}
	\bibliography{bibfile}
\end{document}